\documentclass[journal,12pt,onecolumn,draftcls]{IEEEtran}
\usepackage{array}
\usepackage{textcomp}
\usepackage{stfloats}
\usepackage{float}
\usepackage{url}
\usepackage{verbatim}

\usepackage{amsmath,amssymb,amsfonts}
\usepackage{graphicx}
\usepackage{textcomp}
\usepackage{xcolor}
\usepackage{comment}
\usepackage{indentfirst}
\usepackage{amsthm}
\usepackage{tikz}
\usepackage{bbm}
\usepackage{comment}
\usepackage{standalone}
\usepackage{dirtytalk}
\usepackage{enumitem}
\usepackage{changepage}
\usepackage{cite}
\usepackage{mathtools}
\usepackage{nicefrac}
\newtheorem{thm}{Theorem}
\newtheorem{lem}{Lemma}
\newtheorem{prop}{Proposition}
\newtheorem{cor}{Corollary}

\theoremstyle{definition}
\newtheorem{defn}{Definition}

\definecolor{NYUviolet}{HTML}{57068c} 	
\definecolor{NYUlight}{HTML}{8900e1} 	
\definecolor{NYUdark}{HTML}{330662} 	
\definecolor{NYUnight}{HTML}{220337} 	

\def\BibTeX{{\rm B\kern-.05em{\sc i\kern-.025em b}\kern-.08em
    T\kern-.1667em\lower.7ex\hbox{E}\kern-.125emX}}

\newenvironment{proofachievableWm}{%
   \proof}{\endproof}
\newenvironment{proofconverseWm}{%
  \proof}{\endproof}
\newenvironment{proofconversenoiseless}{%
  \proof}{\endproof}

 \newenvironment{proofachievablenoiseless}{%
   \proof}{\endproof}
\newenvironment{proofconverseW1}{%
  \proof}{\endproof}
\usepackage{etoolbox}

\usepackage[ruled,lined]{algorithm2e}
\usepackage{setspace}
\SetKwComment{Comment}{/* }{ */}
\SetKwInOut{Input}{Input}
\SetKwInOut{Output}{Output}

\SetCommentSty{mycommfont}
\newcommand{\blue}[1]{{\color{blue}#1}}

\newtoggle{singlecolumn}
\toggletrue{singlecolumn} 
    
\begin{document}

\title{Database Matching Under Noisy Synchronization Errors\\\thanks{This research was presented in part at the 2021 IEEE International Symposium on Information Theory (ISIT), the 2022 Asilomar Conference on Signals, Systems, and Computers and the 2022 \& 2023 IEEE Information Theory Workshops (ITW). It has also been in part submitted for conference publication. This work is supported in part by National Science Foundation grants 1815821 and 2148293, and NYU WIRELESS Industrial Affiliates.}} 

\author{Serhat Bakirtas, Elza Erkip\\
Dept. of Electrical and Computer Engineering, New York University\\
\{serhat.bakirtas\},\{elza\}@nyu.edu
}

\maketitle

\begin{abstract}
The re-identification or de-anonymization of users from anonymized data through matching with publicly available correlated user data has raised privacy concerns, leading to the complementary measure of obfuscation in addition to anonymization. Recent research provides a fundamental understanding of the conditions under which privacy attacks, in the form of database matching, are successful in the presence of obfuscation. Motivated by synchronization errors stemming from the sampling of time-indexed databases, this paper presents a unified framework considering both obfuscation and synchronization errors and investigates the matching of databases under noisy entry repetitions. By investigating different structures for the repetition pattern, replica detection and seeded deletion detection algorithms are devised and sufficient and necessary conditions for successful matching are derived. Finally, the impacts of some variations of the underlying assumptions, such as the adversarial deletion model, seedless database matching, and zero-rate regime, on the results are discussed. Overall, our results provide insights into the privacy-preserving publication of anonymized and obfuscated time-indexed data as well as the closely related problem of the capacity of synchronization channels.
\end{abstract}

\section{Introduction}
\label{sec:intro}

\IEEEPARstart{W}{ith} the exponential boom in smart devices and the growing popularity of big data, companies and institutions have been gathering more and more personal data from users which is then either published or sold for research or commercial purposes. Although the published data is typically \emph{anonymized}, \emph{i.e.,} explicit identifiers of the users, such as names and dates of birth are removed, there has been a growing concern over potential privacy leakage from anonymized data, approached from legal~\cite{ohm2009broken} and corporate~\cite{bigdata} points of view. These concerns are also articulated in the respective literature through successful practical de-anonymization attacks on real data~\cite{naini2015you,datta2012provable,narayanan2008robust,sweeney1997weaving,takbiri2018matching,wondracek2010practical,su2017anonymizing,shusterman2019robust,gulmezoglu2017perfweb,bilge2009all,srivatsa2012deanonymizing,cheng2010you,kinsella2011m,kim2016inferring,de2013unique}. \emph{Obfuscation}, which refers to the deliberate addition of noise to the database entries, has been suggested as an additional measure to protect privacy~\cite{sweeney1997weaving}. While extremely valuable, this line of work does not provide a fundamental and rigorous understanding of the conditions under which anonymized and obfuscated databases are prone to privacy attacks.

In the light of the above practical privacy attacks on databases, several groups initiated rigorous analyses of the graph matching problem~\cite{erdos1960evolution,babai1980random,janson2011random,czajka2008improved,yartseva2013performance,pedarsani2013bayesian,fiori2013robust,lyzinski2014seeded,onaran2016optimal,cullina2016improved}. Correlated graph matching has applications beyond privacy, such as image processing~\cite{sanfeliu2002graph}, computer vision~\cite{galstyan2021optimal}, single-cell biological data alignment~\cite{zhu2021robust,tran2020benchmark} and DNA sequencing, which is shown to be equivalent to matching bipartite graphs~\cite{blazewicz2002dna}. Matching of correlated databases, also equivalent to bipartite graph matching, has also been investigated from information-theoretic~\cite{cullina,shirani8849392,dai2019database,bakirtas2021database,bakirtas2022seeded,noiselesslonger} and statistical~\cite{kunisky2022strong} perspectives. In \cite{cullina}, Cullina \emph{et al.} introduced \textit{cycle mutual information} as a correlation metric and derived sufficient conditions for successful matching and a converse result using perfect recovery as the error criterion. In~\cite{shirani8849392}, Shirani \emph{et al.} considered a pair of anonymized and obfuscated databases and drew analogies between database matching and channel decoding. By doing so, they derived necessary and sufficient conditions on the \emph{database growth rate} for reliable matching, in the presence of noise on the database entries. In~\cite{dai2019database}, Dai \emph{et al.} considered the matching of a pair of databases with joint Gaussian attributes with perfect recovery constraint. Similarly, in~\cite{kunisky2022strong}, Kunisky and Niles-Weed considered the same problem from the statistical perspective in different regimes of database size and under several recovery criteria. In~\cite{zeynepdetecting2022}, Kahraman and Nazer investigated the necessary and the sufficient conditions for detecting whether two Gaussian databases are correlated. More recently, motivated by the need for aligning single-cell data obtained from multiple biological sources/experiments~\cite{zhu2021robust,tran2020benchmark}, in~\cite{chen2022one} Chen \emph{et al.} investigated the matching of two noisy databases which are the noisy observations of a single underlying database under the fractional-error criterion, where the noise is assumed to be the Gaussian. They proposed a data-driven approach and analytically derived minimax lower bounds for successful matching.

\begin{figure}[t]
\iftoggle{singlecolumn}{
\centerline{\includegraphics[page=2,width=0.75\textwidth,trim={0 13.5cm 2cm 0},clip]{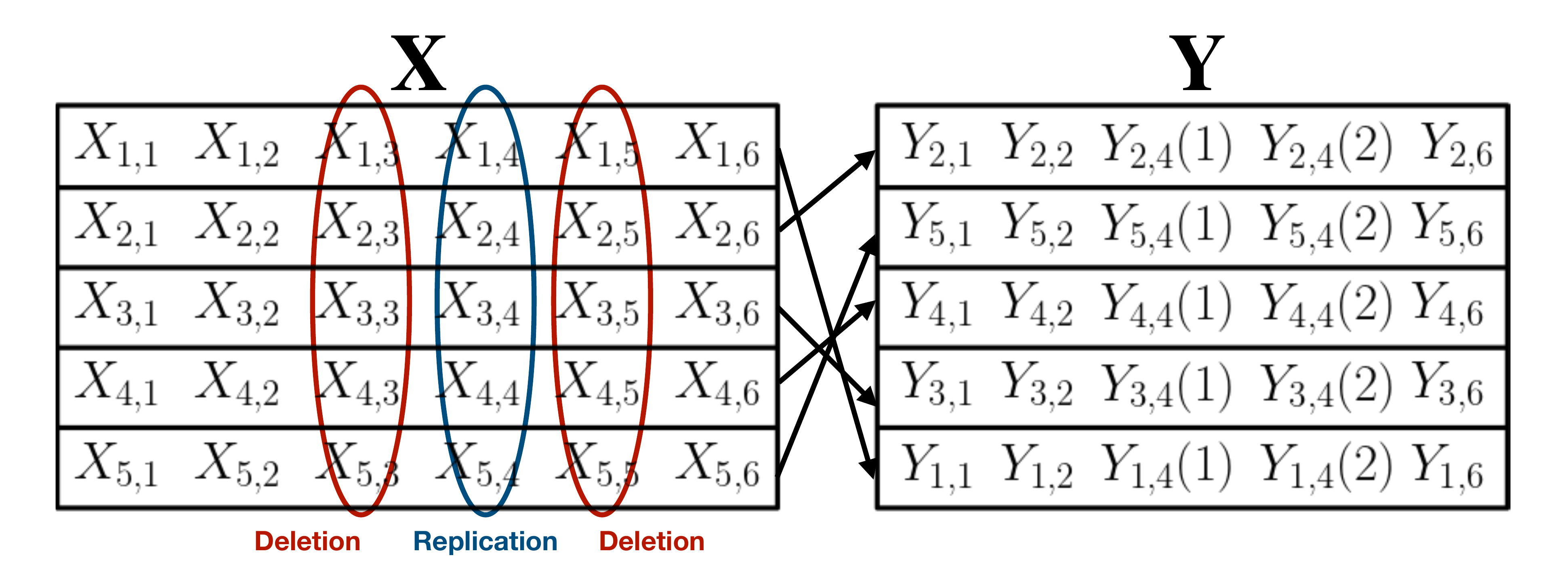}}
}{
\centerline{\includegraphics[page=2,width=0.5\textwidth,trim={0 13.5cm 2cm 0},clip]{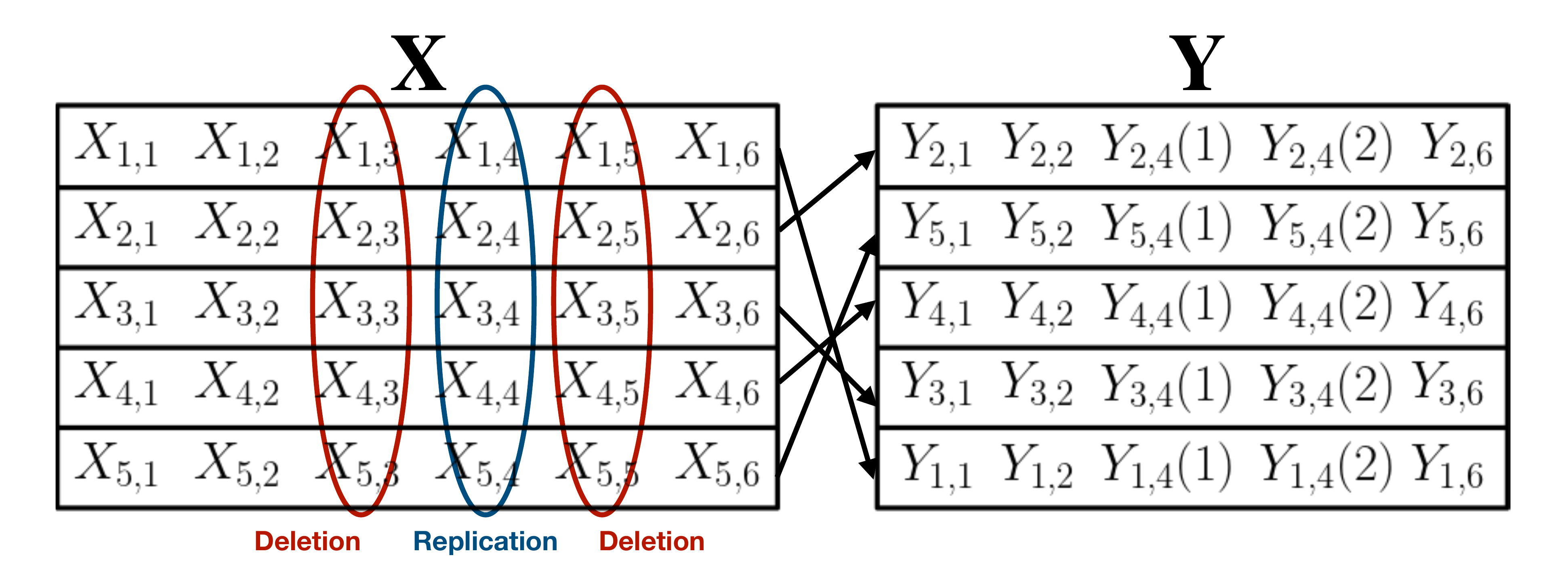}}
}
\caption{An illustrative example of database matching under identical repetition, where each row experiences the same synchronization error. The columns circled in red are deleted whereas the fourth column, which is circled in blue, is repeated twice, \emph{i.e.,} replicated. For each $(i,j)$, $Y_{i,j}$ is the noisy observation of $X_{i,j}$. Furthermore, for each $i$, $Y_{i,4}(1)$ and $Y_{i,4}(2)$ are noisy replicas of $X_{i,4}$. Our goal is to estimate the row permutation $\sigma_n$ which is in this example given as; $\sigma_n(1)=5$, $\sigma_n(2)=1$, $\sigma_n(3)=4$,
$\sigma_n(4)=3$ and $\sigma_n(5)=2$, by matching the rows of $\mathbf{X}$ and $\mathbf{Y}$. Here the $i$\textsuperscript{th} row of $\mathbf{X}$ corresponds to the $\sigma_n(i)$\textsuperscript{th} row of $\mathbf{Y}$.}
\label{fig:intro}
\end{figure}

Motivated by the synchronization errors in the sampling of time-indexed datasets, in this paper, we present a unified generalized framework of the database matching problem under noisy synchronization errors with near-exact recovery criterion. Specifically, we investigate the matching of Markov databases under arbitrary noise and synchronization errors. Our goal is to investigate necessary and sufficient conditions on the database growth rate~\cite{shirani8849392} for the successful matching of database rows. The generalized Markov database model captures correlations of the attributes (columns), where synchronization errors, in the form of random entry deletions and replications, are followed by noise. As such, this paper generalizes the aforementioned work on database matching under only noise. Our setting is illustrated in Figure~\ref{fig:intro}.

We consider two extreme regimes regarding the nature of synchronization errors, as results derived for these corner cases provide insights into the intermediate regime. To this end, first, we focus on the \emph{identical repetition} setting where the repetition pattern is constant across rows. In other words, in this setting, deletions and replications only take place columnwise. We consider a two-phase matching scheme, where we first infer the underlying repetition structure by using permutation-invariant features of columns. This is followed by the matching phase which relies on the known replica and deletion locations. We show that as long as the databases are not independent, in the first phase, replicas can be found with high probability through a series of hypothesis tests on the Hamming distances between columns. Furthermore, assuming \emph{seed} rows whose identities are known in both databases~\cite{shirani2017seeded,fishkind2019seeded} we show that if the seed size $\Lambda_n$ grows double-logarithmically with the number of rows $m_n$, where $n$ denotes the column size, deletion locations can also be extracted. In the absence of noise, seeds are not needed and column histograms can be used to detect both replicas and deletions. Once the repetition (including deletions and replications) locations are identified, in the second phase, we propose a joint typicality-based row matching scheme to derive sufficient conditions on the database growth rate for successful matching. Finally, we prove a tight converse result through a modified version of Fano's inequality, completely characterizing the matching capacity when the repetition pattern is constant across the rows. 

Next, we focus on the other extreme, namely the \emph{independent repetition} setting where the repetition pattern is independent in each row and there is no underlying repetition structure across rows. Under probabilistic side information on the deletion locations, we propose a row-matching scheme and derive an achievable database growth rate. This, together with an outer bound obtained through Fano's inequality, provides upper and lower bounds on the matching capacity in the independent repetition setting. Comparing the bounds in the two extremes, we show that the matching capacity is lower and hence matching is more difficult under the independent repetition model. Finally, based on these two extreme models, we state bounds on the matching capacity for any intermediate repetition structure.

We also discuss the adversarial repetition model, where we assume that synchronization errors, in the form of column deletions, are chosen by a constrained adversary whose goal is to hinder the matching of databases, where the constraint is in the form of a fractional column deletion budget which naturally provides a trade-off between utility and privacy. Since this adversarial model forces us to focus on the worst-case scenario and in turn, prohibits the use of typicality and Fano's inequality, we propose an exact sequence matching and perform a more careful analysis of the worst-case error, focusing on the Hamming distances between the rows (users) of the databases, as is the case in the adversarial channel literature~\cite{bassily2014causal}. Under the identical repetition model, we completely characterize the adversarial matching capacity.

In addition to the characterization of the matching capacity under various assumptions, our results provide sufficient conditions on the number and the size for column histograms to be asymptotically unique. Since histograms naturally show up frequently in information theory, probability theory, and statistics, this result could be of independent interest. In addition, our novel matching scheme in the independent repetition case can be directly converted to a decoding strategy for input-constrained noisy synchronization channels, a well-investigated model in the information theory literature~\cite{gallager1961sequential,5629489,6915855,9056064}.

\subsection{Paper Organization}
\label{subsec:organization}
The organization of this paper is as follows: Section~\ref{sec:problemformulation} contains the problem formulation and the preliminaries. In Section~\ref{sec:matchingcapacityWm}, our main results on the matching capacity under the identical repetition model are presented. Section~\ref{sec:matchingcapacityW1} contains our main results on the matching capacity under the independent repetition assumption. In Section~\ref{sec:discussion}, we discuss the underlying model assumptions and investigate how variations on these assumptions impact some of the results. Finally, in Section~\ref{sec:conclusion} the results and ongoing work are discussed.

\subsection{Notations}
\label{subsec:notations}
In this paper, we use the following notations:
\begin{itemize}
    \item $[n]$ denotes the set of integers $\{1,...,n\}$.
    \item Matrices are denoted with uppercase bold letters. For a matrix $\mathbf{X}$, $X_{i,j}$ denotes the $(i,j)$\textsuperscript{th} entry.
    \item $a^n$ denotes a row vector consisting of scalars $a_1,\dots,a_n$.
    \item Random variables are denoted by uppercase letters while their realizations are denoted by lowercase ones.
    \item The indicator of event $E$ is denoted by $\mathbbm{1}_E$.
    \item $H$ and $H_b$ denote the Shannon entropy and the binary entropy functions~\cite[Chapter 2]{cover2006elements}, respectively.
    \item $O$, $o$, $\Theta$, $\omega$ and $\Omega$ denote the standard asymptotic growth notations~\cite[Chapter 3]{cormen2022introduction}.

    \item $D_{KL}(p_X\|q_X)$ denotes the Kullback-Leibler divergence~\cite[Chapter 2.3]{cover2006elements} between the probability distributions $p_X$ and $q_X$. For scalars $p,q\in(0,1)$, $D(p\|q)$ denotes the Kullback-Leibler divergence between two Bernoulli distributions with respective parameters $p$ and $q$. More formally,
     \iftoggle{singlecolumn}{
    \begin{align}
    D(p\| q) &= (1-p) \log\frac{1-p}{1-q}+ p\log\frac{p}{q}
    \end{align} }
    {
    \begin{align}
    D(p\| q) &= (1-p) \log\frac{1-p}{1-q}+ p\log\frac{p}{q}
    \end{align} 
    }
    \item The logarithms, unless stated explicitly, are in base $2$.
\end{itemize}

\section{Problem Formulation \& Preliminaries}
\label{sec:problemformulation}

\subsection{Problem Formulation}
We use the following definitions, some of which are similar to~\cite{shirani8849392,bakirtas2021database,noiselesslonger}, to formally describe our problem. 

\begin{defn}{\textbf{(Unlabeled Markov Database)}}\label{defn:markovdb}
An ${(m_n,n,\mathbf{P})}$ \emph{unlabeled Markov database} is a randomly generated ${m_n\times n}$ matrix ${\mathbf{X}=\{X_{i,j}\in\mathfrak{X}:i\in[m_n],j\in[n]\}}$ whose rows are \emph{i.i.d.} and follow a first-order stationary Markov process defined over the alphabet ${\mathfrak{X}=\{1,\dots,|\mathfrak{X}|\}}$ with probability transition matrix $\mathbf{P}$ such that
\iftoggle{singlecolumn}{
\begin{gather}
    \mathbf{P} = \gamma \mathbf{I} + (1-\gamma) \mathbf{U}\label{eq:markovtransitionmatrix}\\
    U_{i,j} = u_j>0, \: \forall (i,j)\in \mathfrak{X}^2\\
    \sum\limits_{j\in\mathfrak{X}} u_j =1\\
    \gamma \in [0,1)
\end{gather}
}{
\begin{gather}
    \mathbf{P} = \gamma \mathbf{I} + (1-\gamma) \mathbf{U}\label{eq:markovtransitionmatrix}\\
    U_{i,j} = u_j>0, \: \forall (i,j)\in \mathfrak{X}^2\\
    \sum\limits_{j\in\mathfrak{X}} u_j =1\\
    \gamma \in [0,1)
\end{gather}
}
where $\mathbf{I}$ is the identity matrix. It is assumed that ${X_{i,1}\overset{\text{i.i.d.}}{\sim}\pi=[u_1,\dots,u_{|\mathfrak{X}|}]}$, $i=1,\dots,m_n$,  where $\pi$ is the stationary distribution associated with $\mathbf{P}$.
\end{defn}

Note that Definition~\ref{defn:markovdb} yields the following $n$-letter probability model for row generation:
\begin{align}
    \Pr(X^n=x^n)&= u_{x_1}\prod\limits_{j=2}^n \left[(1-\gamma)u_{x_j} + \gamma \mathbbm{1}_{[x_j=x_{j-1}]} \right], \hspace{2em} \forall x^n\in \mathfrak{X}^n
\end{align}

Observe that, the parameter $\gamma$ determines the correlation among the columns of $\mathbf{X}$. Specifically, $\gamma=0$ corresponds to the case where $X_{i,j}$ are \emph{i.i.d.}

In our work, we are mainly interested in two extreme cases of the repetition pattern:
\begin{itemize}
    \item Every row of $\mathbf{X}$ experiences the same repetition pattern which we call \emph{identical repetition}.
    \item Rows of $\mathbf{X}$ experience \emph{i.i.d.} repetition patterns which we call \emph{independent repetition}.
\end{itemize}
The formal definitions of these two scenarios are provided in Definitions~\ref{defn:labeleddbidenticalrepetition}-\ref{defn:labeleddbindependentrepetition} where the main difference comes from the repetition pattern $S^n$ (Definition~\ref{defn:labeleddbidenticalrepetition}) and repetition matrix (Definition~\ref{defn:labeleddbindependentrepetition}).

\begin{defn}{\textbf{(Labeled Repeated Database under Identical Repetition)}}\label{defn:labeleddbidenticalrepetition}
Let $\mathbf{X}$ be an ${(m_n,n,\mathbf{P})}$ unlabeled Markov database, $S^n$ be vector of length $n$ with $S_{j}$ being \emph{i.i.d.} entries drawn from a discrete probability distribution $p_S$ with a finite integer support ${\{0,\dots,s_{\max}\}}$, $\sigma_n$ be a uniform permutation of $[m_n]$ with $\mathbf{X}$, $S^n$ and $\sigma_n$ independently chosen. Also, let $p_{Y|X}$ be a conditional probability distribution with both $X$ and $Y$ taking values from $\mathfrak{X}$. Given $\mathbf{X}$, $S^n$ and $p_{Y|X}$, the random matrix $\mathbf{Y}$ is called the \emph{labeled repeated database under ıdentical repetition} if the $i$\textsuperscript{th} row $X^n_{i}$ of $\mathbf{X}$ and the $\sigma_n(i)$\textsuperscript{th} row $Y^{K_n}_{\sigma_n(i)}=[Y_{\sigma_n(i),1},\dots,Y_{\sigma_n(i),K_{n}}]$ of $\mathbf{Y}$ have the following relation:
\begin{align}
    \Pr(Y^{K_{n}}_{\sigma_n(i)} &= y^{K_n}|X_i^n=x^n)\notag\\
    &= \prod\limits_{j:S_j\neq 0} \Pr((Y_{\sigma_n(i),K_{j-1}+1},\dots,Y_{\sigma_n(i),K_{j}})=(y_{K_{j-1}+1},\dots,y_{K_{j}})|X_{i,j}=x_j)\label{eq:correlateddbident1}\\
    &= \prod\limits_{j:S_j\neq 0} \prod\limits_{s=1}^{S_{j}} p_{Y|X}(y_{K_{j-1}+s}|x_j)\label{eq:correlateddbident2}
\end{align}
where
\begin{align}
    K_{j}&\triangleq \sum\limits_{t=1}^j S_{t}
\end{align}
Here $S^n$ and $\sigma_n$ are called the \emph{repetition pattern} and \emph{labeling function}, respectively.

Note that $S_{j}$ indicates the times $X_{i,j}$ is repeated (including deletions and replications). When $S_{j}=0$, $X_{i,j}$ is said to be \emph{deleted} (repeated zero times) and when $S_{j}>1$, $X_{i,j}$ is said to be \emph{replicated} $S_{j}$ times (repeated $S_{j}$ times). $\delta\triangleq p_S(0)$ is called the \emph{deletion probability}.

The respective rows $X_{i_1}^n$ and $Y_{i_2}^{K_n}$ of $\mathbf{X}$ and $\mathbf{Y}$ are said to be \emph{matching rows}, if ${\sigma_n(i_1)=i_2}$.
\end{defn}

\begin{defn}{\textbf{(Labeled Repeated Database under Independent Repetition)}}\label{defn:labeleddbindependentrepetition}
Let $\mathbf{X}$ be an ${(m_n,n,\mathbf{P})}$ unlabeled Markov database, $\mathbf{S}$ be an $m_n\times n$ matrix with $S_{i,j}$ \emph{i.i.d.} from a discrete probability distribution $p_S$ with a finite integer support ${\{0,\dots,s_{\max}\}}$, $\sigma_n$ be a uniform permutation of $[m_n]$ with $\mathbf{X}$, $\mathbf{S}$ and $\sigma_n$ independently chosen. Also, let $p_{Y|X}$ be a conditional probability distribution with both $X$ and $Y$ taking values from $\mathfrak{X}$. Given $\mathbf{X}$, $\mathbf{S}$ and $p_{Y|X}$, the random matrix $\mathbf{Y}$ is called the \emph{labeled repeated database under independent repetition} if the $i$\textsuperscript{th} row $X^n_{i}$ of $\mathbf{X}$ and the $\sigma_n(i)$\textsuperscript{th} row $Y^{K_{\sigma_n(i),n}}_{\sigma_n(i)}=[Y_{\sigma_n(i),1},\dots,Y_{\sigma_n(i),K_{\sigma_n(i),n}}]$ of $\mathbf{Y}$ have the following relation:
\begin{align}
    &\Pr(Y^{K_{\sigma_n(i),n}}_{\sigma_n(i)} = y^{K_{\sigma_n(i),n}}|X_i^n=x^n) \notag \\
    &= \prod\limits_{j:S_{\sigma_n(i),j}\neq 0} \Pr((Y_{\sigma_n(i),K_{\sigma_n(i),j-1}+1},\dots,Y_{\sigma_n(i),K_{\sigma_n(i),j}})=(y_{K_{\sigma_n(i),j-1}+1},\dots,y_{K_{\sigma_n(i),j}})|X_{i,j}=x_j)\label{eq:correlateddbindep1}\\
    &= \prod\limits_{j:S_{\sigma_n(i),j}\neq 0} \prod\limits_{s=1}^{S_{\sigma_n(i),j}} p_{Y|X}(y_{K_{\sigma_n(i),j-1}+s}|x_j)\label{eq:correlateddbindep2}
\end{align}
where
\begin{align}
    K_{i,j}&\triangleq \sum\limits_{t=1}^j S_{i,t}
\end{align}
Here $\mathbf{S}$ and $\sigma_n$ are called the \emph{repetition matrix} and \emph{labeling function}, respectively.

Note that $S_{\sigma_n(i),j}$ indicates the times $X_{i,j}$ is repeated (including deletions and replications). When $S_{\sigma_n(i),j}=0$, $X_{i,j}$ is said to be \emph{deleted} (repeated zero times) and when $S_{\sigma_n(i),j}>1$, $X_{i,j}$ is said to be \emph{replicated} $S_{\sigma_n(i),j}$ times (repeated $S_{\sigma_n(i),j}$ times). $\delta\triangleq p_S(0)$ is called the \emph{deletion probability}.

The respective rows $X_{i_1}^n$ and $Y_{i_2}^{K_{i_2,n}}$ of $\mathbf{X}$ and $\mathbf{Y}$ are said to be \emph{matching rows}, if ${\sigma_n(i_1)=i_2}$.
\end{defn}

In our model, the labeled repeated database $\mathbf{Y}$ is obtained by permuting the rows of the unlabeled Markov database $\mathbf{X}$ with the uniform permutation $\sigma_n$ followed by repetition based on the repetition pattern $S^n$ (Definition~\ref{defn:labeleddbidenticalrepetition}) or repetition matrix $\mathbf{S}$ (Definition~\ref{defn:labeleddbindependentrepetition}) and introduction of noise through $p_{Y|X}$. The relationship between $\mathbf{X}$ and $\mathbf{Y}$, as described in Definitions~\ref{defn:labeleddbidenticalrepetition}-\ref{defn:labeleddbindependentrepetition}, is illustrated in Figure~\ref{fig:dmc}. As we formalize later, the goal is to recover the labeling function $\sigma_n$ based on the observations of $\mathbf{X}$ and $\mathbf{Y}$.

Equations \eqref{eq:correlateddbident1}-\eqref{eq:correlateddbident2} (resp. \eqref{eq:correlateddbindep1}-\eqref{eq:correlateddbindep2}) state that we can treat $Y_{\sigma_n(i),j}$ as the output of the discrete memoryless channel (DMC) $p_{Y|X}$ with input sequence consisting of $S_{j}$ (resp. $S_{\sigma_n(i),j}$) copies of $X_{i,j}$ concatenated together. We stress that $p_{Y|X}$ is a general model, capturing any distortion and noise on the database entries, though we refer to this as \say{noise} in this paper.

We will observe that these two models pose different challenges to matching and in turn necessitate different solutions with different implications.

In most of the paper, we assume a random repetition pattern as in Definitions~\ref{defn:labeleddbidenticalrepetition}-\ref{defn:labeleddbindependentrepetition}. In Section~\ref{subsec:adversarialrepetition}, we will discuss the effects of an adversarial worst-case repetition pattern. Note that in this paper, we assume that $p_{X,Y}$ and $p_S$ are available during the matching. For a study of distribution-agnostic database matching, see~\cite{bakirtas2023distribution}.

As discussed in Section~\ref{sec:matchingcapacityWm}, inferring the repetition pattern, particularly deletions, is a difficult task. Therefore, for the identical repetition pattern, we assume the availability of \emph{seeds} to help with the inference of the underlying repetition pattern, similar to database matching~\cite{bakirtas2021database} and graph matching \cite{shirani2017seeded, fishkind2019seeded} settings.

\begin{defn}{\textbf{(Seeds)}}
\label{defn:seeds}
A \emph{seed} is a pair of matching rows whose labels and entries are known universally. A \emph{batch of $\Lambda_n$ seeds} $(\mathbf{G}^{(1)},\mathbf{G}^{(2)})$ is a batch of $\Lambda_n$ correctly-matched row pairs. Here $\mathbf{G}^{(1)}\in\mathfrak{X}^{\Lambda_n\times n}$ has the same row generation process as $\mathbf{X}$, ($\mathbf{G}^{(1)}$,$\mathbf{G}^{(2)}$) have the same relation as $(\mathbf{X}$,$\mathbf{Y})$, as described in Definition~\ref{defn:labeleddbidenticalrepetition} with the same noise distribution $p_{Y|X}$ and repetition pattern $S^n$. $\Lambda_n$ is called the \emph{seed size}.
\end{defn}

Note that in Definition~\ref{defn:seeds}, for notational convenience, the seeds are assumed to be additional to the databases.

Throughout Section~\ref{sec:matchingcapacityWm}, we assume a double logarithmic seed size $\Lambda_n=\Omega(\log\log m_n)$. We will discuss the effects of not having seeds in Section~\ref{subsec:seedlessWm}.

In the independent repetition setting, the seeds offer no additional information, as the repetition pattern is independent in each row. Instead, we assume that the locations of some deleted entries are revealed. This is formalized in the following definition:

\begin{defn}{\textbf{(Partial Deletion Location Information)}} \label{defn:partialdellocinfo}
For a labeled repeated database under independent repetition (Definition~\ref{defn:labeleddbindependentrepetition}), the \emph{partial deletion location information} $\mathbf{A}$ is an ${m_n\times n}$ random matrix, with the following conditional distribution on repetition matrix $\mathbf{S}$:
\begin{align}
    \Pr(A_{i,j}= 1|\mathbf{S}) &= \alpha \mathbbm{1}_{[S_{i,j}=0]}
\end{align}
where ${A_{i,j}=1}$ corresponds to $X_{\sigma_n^{-1}(i),j}$ being revealed as deleted and ${A_{i,j}=0}$ corresponds to either $X_{\sigma_n^{-1}(i),j}$ not being deleted or not being revealed after deletion. The parameter ${\alpha\in[0,1]}$ is called the \emph{deletion detection probability}.
\end{defn}
Definition~\ref{defn:partialdellocinfo} states that the location of each deleted entry is revealed with probability $\alpha$. Since the entries of $\mathbf{S}$ are i.i.d. and $\mathbf{S}$ and $\mathbf{X}$ are independent, each deleted column is revealed independently of the other columns of $\mathbf{S}$ and $\mathbf{X}$. Furthermore, since $S_{i,j}$ are drawn \emph{i.i.d.}, so are $A_{i,j}$.
\begin{figure}[t]
\iftoggle{singlecolumn}{
\centerline{\includegraphics[width=0.7\textwidth]{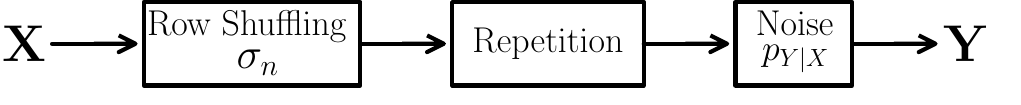}}
}{
\centerline{\includegraphics[width=0.5\textwidth]{Figures/DMCdiagram2.pdf}}
}
\caption{Relation between the unlabeled database $\mathbf{X}$ and the correlated repeated database $\mathbf{Y}$. Repetition is represented by the repetition pattern $S^n$ (Definition~\ref{defn:labeleddbidenticalrepetition}) or the repetition matrix $\mathbf{S}$ (Definition~\ref{defn:labeleddbindependentrepetition}).}
\label{fig:dmc}
\end{figure}

\begin{sloppypar}
\begin{defn}{\textbf{(Successful Matching Scheme)}}\label{defn:matchingschemeidentical}
In the identical (resp. independent) repetition setting, a \emph{matching scheme} is a sequence of mappings ${\phi_n: (\mathbf{X},\mathbf{Y},\mathbf{G}^{(1)},\mathbf{G}^{(2)})\mapsto \hat{\sigma}_n }$ (resp. ${\phi_n: (\mathbf{X},\mathbf{Y},\mathbf{A})\mapsto \hat{\sigma}_n }$) where $\mathbf{X}$ is the unlabeled Markov database, $\mathbf{Y}$ is the labeled repeated database, $(\mathbf{G}^{(1)},\mathbf{G}^{(2)})$ are seeds (resp. $\mathbf{A}$ is the partial deletion location information) and $\hat{\sigma}_n$ is the estimate of the correct labeling function $\sigma_n$. The scheme $\phi_n$ is \emph{successful} if 
\iftoggle{singlecolumn}{
\begin{align}
\Pr\left(\hat{\sigma}_n(J)\neq\sigma_n(J)\right)&\to 0  \text{ as }n\to\infty \label{eq:proberroridentical}
\end{align}}
{
\begin{align}
\Pr\left(\hat{\sigma}_n(J)\neq\sigma_n(J)\right)&\to 0  \text{ as }n\to\infty \label{eq:proberroridentical}
\end{align}
}
where the index $J$ is drawn uniformly from $[m_n]$.
\end{defn}
\end{sloppypar}

Observe that the performance criterion considered in Definition~\ref{defn:matchingschemeidentical} allows a sublinear fraction of the rows to be mismatched. This near-perfect performance criterion allows us to utilize communication and information-theoretic tools and work with arbitrary distributions whereas as far as we are aware the prior work considering the perfect recovery criterion mainly focuses on one specific distribution. This success criterion is also known as \emph{near-perfect} or \emph{almost-perfect recovery}~\cite{kunisky2022strong}. Other success definitions include \emph{perfect recovery}~\cite{cullina,dai2019database,kunisky2022strong}, where \emph{all} rows have to be perfectly aligned, and \emph{weak-recovery} or \emph{linear-error}~\cite{kunisky2022strong} where a constant fraction of the rows is allowed to be mismatched. For an extensive comparison of the Gaussian database alignment results under these different performance criteria, please see~\cite{kunisky2022strong}.

We stress that in database matching, the relationship between the row size $m_n$, the column size $n$, and the database distribution parameters are of interest~\cite{kunisky2022strong,zeynepdetecting2022,tamir2022joint}. Note that for fixed column size $n$, as the row size $m_n$ increases, matching becomes harder. This is because for a given column size $n$, as the row size $m_n$ increases, so does the probability of mismatch as a result of having a larger candidate row set. Furthermore, as stated in~\cite[Theorem 1.2]{kunisky2022strong}, for distributions with parameters constant in $n$ and $m_n$, the regime of interest is the logarithmic regime where $n\sim \log m_n$. Thus, we utilize the \emph{database growth rate} introduced in~\cite{shirani8849392} to characterize the relationship between the row size $m_n$ and the column size $n$.  

\begin{defn}\label{defn:dbgrowthrate}{\textbf{(Database Growth Rate)}}
The \emph{database growth rate} $R$ of an $(m_n,n,\mathbf{P})$ unlabeled Markov database is defined as 
\iftoggle{singlecolumn}{
\begin{align}
    R&=\lim\limits_{n\to\infty} \frac{1}{n}\log m_n
\end{align}
}{
\begin{align}
    R&=\lim\limits_{n\to\infty} \frac{1}{n}\log m_n
\end{align}
}
\end{defn}

In Sections~\ref{sec:matchingcapacityWm} and \ref{sec:matchingcapacityW1}, we assume that the database growth rate $R$ is positive and $m_n = 2^{nR}$ for notational simplicity. We will discuss the zero-rate regime $R=0$ in Section~\ref{subsec:zerorateWm}.

\begin{defn}{\textbf{(Achievable Database Growth Rate)}}\label{defn:achievable}
Consider a sequence of ${(m_n,n,\mathbf{P})}$ unlabeled Markov databases, a repetition probability distribution $p_S$, a noise distribution $p_{Y|X}$ and the resulting sequence of labeled repeated databases under identical (resp. independent) repetition. For a seed size $\Lambda_n$ (resp. a deletion detection probability $\alpha$), a database growth rate $R$ is said to be \emph{achievable} if there exists a successful matching scheme when the unlabeled database has a growth rate $R$.
\end{defn}

\begin{defn}{\textbf{(Matching Capacity)}}\label{defn:matchingcapacity}
Under identical (resp. independent) repetition, the \emph{matching capacity} $C$ is the supremum of the set of all achievable rates corresponding to a probability transition matrix $\mathbf{P}$, repetition probability distribution $p_S$, noise distribution $p_{Y|X}$, and seed size $\Lambda_n$ (resp. a deletion detection probability $\alpha$).
\end{defn}

In this paper, our goal is to characterize the matching capacity under the two extreme repetition structures, namely identical repetition and independent repetition, respectively, by providing database matching schemes as well as upper bounds on all achievable database growth rates.

\subsection{Preliminaries}
For the sake of completeness, we present below some of the classical information-theoretic definitions and results, most of which are borrowed from~\cite{shirani8849392,cover2006elements}, that will be used throughout this paper.

\begin{defn}{\textbf{(Entropy Rate)}} For the discrete random process $\mathcal{X}$ characterized by $p_{X^n}$, with $n\in\mathbb{N}$, the entropy rate is defined as:
\begin{align}
    H(\mathcal{X}) &\triangleq \lim\limits_{n\to\infty} \mathbb{E}[-\log p_{X_{n+1}|X^n}(X_{n+1}|X^n)].
\end{align}
when the limit exists.

\end{defn}

\begin{defn}{\textbf{(Typicality)}} The $\epsilon$-typical set associated with the discrete random process $\mathcal{X}$ is defined as
\begin{align}
    A_\epsilon^{(n)}(X) &\triangleq \left\{x^n: \left|-\frac{1}{n}\log p_{X^n}(x^n)-H(\mathcal{X})\right|\le \epsilon\right\}
\end{align}
 where $H(\mathcal{X})$ is the entropy rate of $\mathcal{X}$.
\end{defn}

\begin{defn}{\textbf{(Joint Typicality)}} The $\epsilon$-typical set associated with the discrete random processes $(\mathcal{X},\mathcal{Y})$ is defined as
\begin{align}
    A_\epsilon^{(n)}(X,Y) &\triangleq \left\{(x^n,y^n): \left|-\frac{1}{n}\log p_{X^n,Y^n}(x^n,y^n)-H(\mathcal{X},\mathcal{Y})\right|\le \epsilon\right\}
\end{align}
where $H(\mathcal{X},\mathcal{Y})$ is the entropy rate of $(\mathcal{X},\mathcal{Y})$.
\end{defn}

\begin{lem}{\textbf{(Generalized AEP~\cite[Theorem 1]{barron1985strong})}}\label{lem:AEP}
    For the discrete stationary random process $\mathcal{X}$ characterized by $p_{X^n}$, with $n\in\mathbb{N}$, we have
    \begin{align}
        -\frac{1}{n}\log p_{X^n}(x^n)&\overset{a.s.}{\to}H(\mathcal{X}).
    \end{align}
\end{lem}

Along with standard information-theoretical arguments, Lemma~\ref{lem:AEP} leads to the following:
\begin{prop}{\textbf{(Typicality)}}
    For a $\epsilon$-typical sequence $x^n\in A_{\epsilon}^{(n)}(X)$ we have
    \begin{align}
        2^{-n(H(\mathcal{X})+\epsilon)}&\le p_{X^n}(x^n) \le 2^{-n(H(\mathcal{X})-\epsilon)}
    \end{align}
    Furthermore, 
    \begin{align}
        2^{n(H(\mathcal{X})-\epsilon)}&\le |A_{\epsilon}^{(n)}(X)| \le 2^{n(H(\mathcal{X})+\epsilon)}
    \end{align}
    for large $n$.
\end{prop}

\begin{prop}{\textbf{(Joint Typicality)}}
    For a $\epsilon$-typical sequence pair $(x^n,y^n)\in A_{\epsilon}^{(n)}(X,Y)$ we have
    \begin{align}
        2^{-n(H(\mathcal{X},\mathcal{Y})+\epsilon)}&\le p_{X^n,Y^n}(x^n) \le 2^{-n(H(\mathcal{X},\mathcal{Y})-\epsilon)}
    \end{align}
    Furthermore, 
    \begin{align}
        2^{n(H(\mathcal{X},\mathcal{Y})-\epsilon)}&\le |A_{\epsilon}^{(n)}(X,Y)| \le 2^{n(H(\mathcal{X},\mathcal{Y})+\epsilon)}
    \end{align}
    for large $n$.
\end{prop}

\begin{prop}{\textbf{(Joint AEP)}}\label{prop:jointAEP}
    Consider a correlated pair of stochastic processes $(\mathcal{X},\mathcal{Y})$ characterized by $p_{X^n,Y^n}$, with $n\in\mathbb{N}$. Let $\Tilde{X}^n$ and $\Tilde{Y}^n$ be generated according to the marginal distributions $p_{X^n}$ and $p_{Y}^n$, independently. Then, the following holds:
    \begin{align}
        \Pr((\Tilde{X}^n,\Tilde{Y}^n)\in A_\epsilon^{(n)}(X,Y))&\le  2^{-n(I(X;Y)-3\epsilon)}
    \end{align}
    where $I(X;Y)\triangleq H(\mathcal{X})+H(\mathcal{Y})-H(\mathcal{X},\mathcal{Y})$ is the mutual information rate. Furthermore, 
    \begin{align}
         \Pr((\Tilde{X}^n,\Tilde{Y}^n)\in A_\epsilon^{(n)}(X,Y))&\ge  (1-\epsilon)2^{-n(I(X;Y)+3\epsilon)}
    \end{align}
    for large $n$.
\end{prop}

\section{Matching Capacity For Identical Repetition}
\label{sec:matchingcapacityWm}
In this section, we present the matching capacity $C$ for an identical repetition pattern with seed size $\Lambda_n=\Omega(\log\log m_n)$. We will show that when $\Lambda_n=\Omega(\log\log m_n)$, the repetition pattern, including the deletion locations, can be inferred.

We state the main result of this section in Theorem~\ref{thm:mainresultWm} and prove its achievability by proposing a three-step approach: \emph{i)} noisy replica detection and \emph{ii)} deletion detection using seeds, followed by \emph{iii)} row matching. Then, we prove the converse part. Finally, we focus on the noiseless setting as a special case where we prove that we can devise a new detection algorithm specific to the noiseless model which renders the seeds obsolete.

\begin{thm}{\textbf{(Matching Capacity for Identical Repetition})}\label{thm:mainresultWm}
Consider a probability transition matrix $\mathbf{P}$, a column repetition distribution $p_S$ with an identical repetition pattern, and a noise distribution $p_{Y|X}$. Then, for a seed size ${\Lambda_n=\Omega(\log\log m_n)}$, the matching capacity is
\iftoggle{singlecolumn}{
\begin{align}
    C &= \lim\limits_{n\to\infty} \frac{I(X^n;Y^{K_n},S^n)}{n}\label{eq:matchingcap}
\end{align}
}{
\begin{align}
    C &= \lim\limits_{n\to\infty} \frac{I(X^n;Y^{K_n},S^n)}{n}\label{eq:matchingcap}
\end{align}
}
where $X^n$ is a Markov chain with probability transition matrix $\mathbf{P}$ and stationary distribution $\pi$, $S_i\overset{\text{iid}}{\sim} p_S$ and 
\begin{align}
    \Pr(Y^{K_n} = y^{K_n}|X^n=x^n)
    &= \prod\limits_{j:S_j\neq 0} \Pr((Y_{K_{j-1}+1},\dots,Y_{K_j})=(y_{K_{j-1}+1},\dots,y_{K_j})|X_{j}=x_j)\\
    &= \prod\limits_{j:S_j\neq 0} \prod\limits_{s=1}^{S_j} p_{Y|X}(y_{K_{j-1}+s}|x_j)
\end{align}
where $K_{j}\triangleq \sum\limits_{t=1}^j S_{t}$.
\end{thm}
Because of the independence of $X^n$ and $S^n$, \eqref{eq:matchingcap} can also be represented as 
\iftoggle{singlecolumn}{
\begin{align}
    C &= \lim\limits_{n\to\infty} \frac{I(X^n;Y^{K_n}|S^n)}{n}.
\end{align}
}{
\begin{align}
    C &= \lim\limits_{n\to\infty} \frac{I(X^n;Y^{K_n}|S^n)}{n}.
\end{align}
}
Hence, Theorem~\ref{thm:mainresultWm} states that although the repetition pattern $S^n$ is not known apriori, for a seed size $\Lambda_n=\Omega(\log\log m_n)$, we can achieve a database growth rate as if we knew $S^n$. Since the utility of seeds increases with the seed size $\Lambda_n$, we will focus on $\Lambda_n=\Theta(\log\log m_n)$, which we show is sufficient to achieve the matching capacity.

    Even though the specific Markov row generation process, assumed in Definition~\ref{defn:markovdb}, does not show up in~\eqref{eq:matchingcap}, it plays a significant role in the estimation of the repetition pattern $S^n$, as can be seen in Appendices~\ref{proof:noisyreplicadetection}-\ref{proof:histogram}.

\begin{cor}{\textbf{(Matching Capacity for Identical Repetition with I.I.D. Database Entries)}}
    When $\gamma=0$, resulting in an \emph{i.i.d.} database distribution $p_X(x)=u_x$, $\forall x\in\mathfrak{X}$, the matching capacity is 
    \begin{align}
        C &= I(X;Y^S|S)
    \end{align}
    where $S\sim p_S$ and ${Y^S=Y_1,\dots,Y_S}$ such that
\begin{align}
    \Pr(Y^S=y_1,\dots,y_S|X=x)&=\begin{cases}
        \prod\limits_{i=1}^S p_{Y|X}(y_i|x),&\text{if }s>0\\
        \mathbbm{1}_{[y^S = E]},&\text{if }s=0
    \end{cases}
\end{align}
and $E$ denotes the empty string.
\end{cor}

The rest of this section is on the proof of Theorem~\ref{thm:mainresultWm}. In Section~\ref{subsec:replicadetection}, we discuss our noisy replica detection algorithm which does not utilize the seeds and prove its asymptotic performance. In Section~\ref{subsec:seededdeletiondetection}, we introduce a deletion detection algorithm that uses seeds and derive a seed size sufficient for an asymptotic performance guarantee. Then, in Section~\ref{subsec:matchingschemeWm}, we combine these two algorithms and prove the achievability of Theorem~\ref{thm:mainresultWm} by proposing a typicality-based matching scheme for rows, which is performed once replicas and deletions are detected. In Section~\ref{subsec:converseWm}, we prove the converse part of Theorem~\ref{thm:mainresultWm}. Finally, in Section~\ref{subsec:noiselessWm}, we focus on the special case of no noise on the repeated entries and provide a single repetition (replica and deletion) detection algorithm that does not require any seeds.

Note that when the two databases are independent, Theorem~\ref{thm:mainresultWm} states that the matching capacity becomes zero, hence our results trivially hold. As a result, throughout this section, we assume that the two databases are not independent. 

\subsection{Noisy Replica Detection}\label{subsec:replicadetection}
We propose to detect the replicas by extracting permutation-invariant features of the columns of $\mathbf{Y}$. Our algorithm only considers the columns of $\mathbf{Y}$ and as such, can only detect replicas, not deletions. Note that our replica detection algorithm does not require any seeds unlike seeded deletion detection discussed in Section~\ref{subsec:seededdeletiondetection}.

Our proposed replica detection algorithm (Algorithm~\ref{alg:noisyreplicadetection}) adopts the \emph{Hamming distance between consecutive columns} of $\mathbf{Y}$ as a permutation-invariant feature of the columns. The permutation-invariance allows us to perform replica detection on $\mathbf{Y}$ with no prior information on $\sigma_n$.

Let $K_n$ denote the number of columns of $\mathbf{Y}$, $C^{m_n}_j$ denote the $j$\textsuperscript{th} column of $\mathbf{Y}$, $j=1,\dots,K_n$. The replica detection algorithm works as follows: We first compute the Hamming distances $H_j$ between consecutive columns $C^{m_n}_j$ and $C^{m_n}_{j+1}$, for $j\in[K_n-1]$. More formally,
\begin{align}
    H_j & \triangleq \sum\limits_{t=1}^{m_n} \mathbbm{1}_{[Y_{t,j+1}\neq Y_{t,j}]},\hspace{2em} \forall j\in[K_n-1]\label{eq:RHD}
    \end{align}
For some average Hamming distance threshold $\tau\in(0,1)$ chosen based on $\mathbf{P}$ and $p_{Y|X}$ (See Appendix~\ref{proof:noisyreplicadetection}), the algorithm decides that $C^{m_n}_{j}$ and $C^{m_n}_{j+1}$ are replicas only if $H_j<m_n \tau$, and correspond to distinct columns of $\mathbf{X}$ otherwise. In the following lemma, we show that Algorithm~\ref{alg:noisyreplicadetection} can infer the replicas with high probability. Observe that the runtime of Algorithm~\ref{alg:noisyreplicadetection} is $O(m_n n)$, the computational bottleneck being the computation of $\{H_j\}_{j=1}^{K_n-1}$.

\begin{algorithm}[t]
\caption{Noisy Replica Detection Algorithm}\label{alg:noisyreplicadetection}
\Input{$(\mathbf{Y},\mathbf{P},p_{Y|X})$}
\Output{isReplica}
$H\gets $ RunningHammingDist($\mathbf{Y}$)\Comment*[r]{Eq.~\eqref{eq:RHD}}
$\tau \gets $ThresholdSelection$(\mathbf{P},p_{Y|X})$\Comment*[r]{Threshold selection. See Appendix~\ref{proof:noisyreplicadetection}.}
isReplica $\gets \varnothing$\;

\For{$j = 1$ \KwTo \textup{columnSize(}$\mathbf{Y}$\textup{)}$-1$}{
  \eIf{$H[j]\le \tau$ $*$ \textup{rowSize(}$\mathbf{Y}$\textup{)}}{
    isReplica$[j] \gets$ TRUE\;
  }{
      isReplica$[j] \gets $ FALSE\;
    }
  
}
\end{algorithm}

\begin{lem}{\textbf{(Noisy Replica Detection)}}\label{lem:noisyreplicadetection}
Let $F_j$ denote the event that the Hamming distance-based algorithm described above fails to infer the correct replica relationship between the columns $C^{m_n}_{j}$ and $C^{m_n}_{j+1}$ of $\mathbf{Y}$, $j=1,\dots,K_n-1$. The total probability of replica detection error of Algorithm~\ref{alg:noisyreplicadetection} diminishes as $n\to\infty$, that is
\iftoggle{singlecolumn}{
\begin{align}
    \Pr(\bigcup\limits_{j=1}^{K_n-1} F_j)&\to 0\text{ as }n\to\infty. \label{eq:replicadetection}
\end{align}
}{
\begin{align}
    \Pr(\bigcup\limits_{j=1}^{K_n-1} E_j)&\to 0\text{ as }n\to\infty. \label{eq:replicadetection}
\end{align}
}
\end{lem}
\begin{proof}
See Appendix~\ref{proof:noisyreplicadetection}.
\end{proof}

\subsection{Deletion Detection Using Seeds}\label{subsec:seededdeletiondetection}
The replica detection algorithm discussed in Section~\ref{subsec:replicadetection} only uses $\mathbf{Y}$ and infers only the replicas, not the deletions. We next propose a deletion detection algorithm that uses seeds.

Let $\smash{(\mathbf{G}^{(1)},\mathbf{G}^{(2)})}$ be a batch of $\smash{\Lambda_n=\Theta(\log\log m_n)}$ seeds with the identical repetition pattern $S^n$ as $(\mathbf{X},\mathbf{Y})$. Our deletion detection algorithm (Algorithm~\ref{alg:seededdeletiondetection}) works as follows: After finding the replicas as in Section~\ref{subsec:replicadetection}, we discard all extra copies, keeping only the original entry in a replica run with $S_j>1$ from $\mathbf{G}^{(2)}$, to obtain $\smash{\tilde{\mathbf{G}}^{(2)}}$, whose column size is denoted by $\hat{K}_n$. At this step, we only have deletions. 

Next, for each index pair $(i,j)\in[n]\times[\hat{K}_n]$, we compute the Hamming distance $L_{i,j}$ between the $i$\textsuperscript{th} column $\smash{G_i^{(1)}}$ of $\smash{\mathbf{G}^{(1)}}$ and the $j$\textsuperscript{th} column $\smash{G_j^{(2)}}$ of $\smash{\tilde{\mathbf{G}}^{(2)}}$. More formally, we compute
\iftoggle{singlecolumn}{
\begin{align}
    L_{i,j} &\triangleq \sum_{t=1}^{\Lambda_n} \mathbbm{1}_{\left[{G}^{(1)}_{t,i}\neq \tilde{{G}}^{(2)}_{t,j}\right]}.
\end{align}
}{
\begin{align}
    L_{i,j} &\triangleq \sum_{t=1}^{\Lambda_n} \mathbbm{1}_{\left[{G}^{(1)}_{t,i}\neq \tilde{{G}}^{(2)}_{t,j}\right]}.
\end{align}
}
Then, for each index $i\in[n]$, the algorithm decides $\smash{G_i^{(1)}}$ is retained (not deleted) only if there exists a column $\smash{G_j^{(2)}}$ in $\smash{\tilde{\mathbf{G}}^{(2)}}$ with $\smash{L_{i,j}\le \Lambda_n \bar{\tau}}$, for some average Hamming distance threshold $\bar{\tau}\in(0,1)$ chosen based on $\mathbf{P}$ and $p_{Y|X}$ (See Appendix~\ref{proof:seededdeletiondetection}). In this case, we assign $\hat{I}_i=0$, where $\hat{I}_i$ is the indicator of $\smash{G^{(1)}_i}$ being inferred as deleted. Otherwise, the algorithm decides $\smash{G_i^{(1)}}$ is deleted, assigning $\hat{I}_i= 1$. At the end of this procedure, the algorithm outputs an estimate $\hat{I}^n=(\hat{I}_1,\dots,\hat{I}_n)$ of the true deletion pattern $I^n_{\text{del}}=(I_1,\dots,I_n)$. Here, for each $i\in[n]$ we have 

\iftoggle{singlecolumn}{
\begin{align}
    I_i&\triangleq \mathbbm{1}_{[S_i=0]}\\
    \hat{I}_i &\triangleq \mathbbm{1}_{\left[\exists j\in [\hat{K}_n]:\: L_{i,j}\le \Lambda_n \bar{\tau}\right]}
\end{align}
}{
\begin{align}
    I_i&\triangleq \mathbbm{1}_{[S_i=0]}\\
    \hat{I}_i &\triangleq \mathbbm{1}_{\left[\exists j\in [\hat{K}_n]:\: L_{i,j}\le \Lambda_n \bar{\tau}\right]}
\end{align}
}

Note that such a Hamming distance-based strategy depends on pairs of matching entries in a pair of seed rows in $\smash{\mathbf{G}^{(1)}}$ and $\smash{\tilde{\mathbf{G}}^{(2)}}$ having a higher probability of being equal than non-matching entries. More formally, WLOG, let $S_j\neq 0$ and $\tilde{X}_{i,j}$ and $\tilde{Y}_{i,j}$ denote the respective $(i,j)$\textsuperscript{th} entries of $\mathbf{G}^{(1)}$ and $\tilde{\mathbf{G}}^{(2)}$. Given a matching pair ${(\tilde{X}_{i,j},\tilde{Y}_{i,j})}$ of entries and any non-matching pair ${(\tilde{X}_{i,l},\tilde{Y}_{i,j})}$, $l\neq j$ we need
\iftoggle{singlecolumn}{
\begin{align}
    \Pr(\tilde{Y}_{i,j}\neq \tilde{X}_{i,j})<\Pr(\tilde{Y}_{i,j}\neq \tilde{X}_{i,l})\label{eq:conditiondeletiondetection}
\end{align}
}{
\begin{align}
    \Pr(\tilde{Y}_{i,j}\neq \tilde{X}_{i,j})<\Pr(\tilde{Y}_{i,j}\neq \tilde{X}_{i,l})\label{eq:conditiondeletiondetection}
\end{align}
}
which may not be true in general. 

For example, suppose we have a binary uniform \emph{i.i.d.} distribution, \emph{i.e.,} ${\mathfrak{X}=\{0,1\}}$ with $\gamma=0$ and ${u_1=\nicefrac{1}{2}}$ (recall Definition~\ref{defn:markovdb}). Further assume that $p_{Y|X}$ follows BSC($q$), \emph{i.e.} ${p_{Y|X}(x|x)=1-q}$, ${x=0,1}$. Note that when ${q>\nicefrac{1}{2}}$, equation~\eqref{eq:conditiondeletiondetection} is not satisfied. However, in this example, we can flip the labels in $Y$ by applying the bijective remapping ${\Phi=\left(\begin{smallmatrix}
0 & 1\\
1 & 0
\end{smallmatrix}\right)}$ to $Y$ in order to satisfy equation~\eqref{eq:conditiondeletiondetection}. 
\begin{sloppypar}
Thus, as long as such a permutation ${\Phi}$ of $\mathfrak{X}$ satisfying equation~\eqref{eq:conditiondeletiondetection} exists, we can use Algorithm~\ref{alg:seededdeletiondetection}.
Now, suppose that such a mapping $\Phi$ exists. We apply $\Phi$ to the entries of $\smash{\tilde{\mathbf{G}}^{(2)}}$ to construct $\smash{\tilde{\mathbf{G}}_{\Phi}^{(2)}}$. Then, our deletion detection algorithm follows the above steps computing $L_{i,j}(\Phi)$ for each index pair $(i,j)\in[n]\times[\hat{K}_n]$ and outputs the deletion pattern estimate ${\hat{I}^n(\Phi)=(\hat{I}_1(\Phi),\dots,\hat{I}_n(\Phi))}$ where

\iftoggle{singlecolumn}{
\begin{align}
L_{i,j}(\Phi) &\triangleq \sum_{t=1}^{\Lambda_n} \mathbbm{1}_{\left[{G}^{(1)}_{t,i}\neq \tilde{{G}}^{(2)}_{\Phi_{t,j}}\right]}.\label{eq:XHD} \\
    \hat{I}_i(\Phi) &\triangleq \mathbbm{1}_{\left[\exists j\in [\hat{K}_n]:\: L_{i,j}(\Phi)\le \Lambda_n \bar{\tau}\right]}
\end{align}
}{
\begin{align}
L_{i,j}(\Phi) &\triangleq \sum_{t=1}^{\Lambda_n} \mathbbm{1}_{\left[{G}^{(1)}_{t,i}\neq \tilde{{G}}^{(2)}_{\Phi_{t,j}}\right]}.\label{eq:XHD} \\
    \hat{I}_i(\Phi) &\triangleq \mathbbm{1}_{\left[\exists j\in [\hat{K}_n]:\: L_{i,j}(\Phi)\le \Lambda_n \bar{\tau}\right]}
\end{align}}
and $G_j^{(2)}(\Phi)$ is the $j$\textsuperscript{th} column of 
$\smash{\tilde{\mathbf{G}}_{\Phi}^{(2)}}$. Note that the runtime of Algorithm~\ref{alg:seededdeletiondetection} is $O(n^2 \Lambda_n)$, the computational bottleneck being the computation of $\mathbf{L}(\Phi)$.
\end{sloppypar}
\begin{sloppypar}
The following lemma states that such a bijective mapping $\Phi$ always exists and for a seed size ${\Lambda_n=\Theta(\log n)=\Theta(\log\log m_n)}$, this algorithm can infer the deletion locations with high probability.
\end{sloppypar}

\begin{algorithm}[t]
\caption{Seeded Deletion Detection Algorithm}\label{alg:seededdeletiondetection}
\Input{$(\mathbf{G}^{(1)},\mathbf{G}^{(2)},\mathbf{P},p_{Y|X},\textup{isReplica})$}
\Output{isDeleted}
$\mathfrak{S}(\mathfrak{X})\gets$ SymmetryGroup($\mathfrak{X}$)\;
\For{$s \gets 1$ \KwTo $|\mathfrak{X}|!$}{
$\Phi\gets\mathfrak{S}(\mathfrak{X})[s]$\Comment*[r]{Pick a remapping.} 
\blue{/* isUseful checks if $\Phi$ satisfies ~\eqref{eq:conditiondeletiondetection}.*/}\\
 \If{\textup{isUseful(}$\Phi,\mathbf{P},p_{Y|X}$\textup{)}}{
    break\Comment*[r]{Move on with $\Phi$.}
  }
}
$\tilde{\mathbf{G}}^{(2)} \gets $ExtraReplicaRemoval($\mathbf{G}^{(2)}$,isReplica)\Comment*[r]{Remove extra copies.}
$\tilde{\mathbf{G}}^{(2)}_\Phi \gets$Remap($\tilde{\mathbf{G}}^{(2)},\Phi$)\Comment*[r]{Apply remapping $\Phi$.}
$\mathbf{L}(\Phi)\gets $ ComputeHammingDist($\mathbf{G}^{(1)},\tilde{\mathbf{G}}^{(2)}_\Phi$)\Comment*[r]{Eq.~\eqref{eq:XHD}}
$\bar{\tau} \gets $ThresholdSelection2$(\mathbf{P},p_{Y|X})$\Comment*[r]{Threshold selection. See Apprendix~\ref{proof:seededdeletiondetection}.}

\For{$i = 1$ \KwTo  \textup{columnSize($\mathbf{G}^{(1)}$)}}{
\For{$j = 1$ \KwTo  \textup{columnSize($\tilde{\mathbf{G}}^{(2)}_\Phi$)}}{
  \eIf{$\mathbf{L}(\Phi)[i][j]\le \bar{\tau}$ $*$ \textup{rowSize($\mathbf{G}^{(1)}$)} }{
    isDeleted$[i] \gets$ FALSE\;
    break\;
  }{
      isDeleted$[j] \gets $ TRUE\;
    }
  
}
}
\end{algorithm}
\begin{lem}{\textbf{(Seeded Deletion Detection)}}\label{lem:seededdeletiondetection}
For a repetition pattern ${S}^n$, let ${I_\text{del}=\{j\in[n]|S_j=0\}}$. Then there exists a bijective mapping $\Phi$ such that equation~\eqref{eq:conditiondeletiondetection} holds after the remapping. In addition, for a seed size $\Lambda_n=\Theta(\log n)$, using Algorithm~\ref{alg:seededdeletiondetection}, we have 
\iftoggle{singlecolumn}{
 \vspace{-0.5em}
\begin{align}
    \Pr\left(\hat{I}(\Phi)=I_\text{del}\right)&\to 1\text{ as }n\to\infty.
\end{align}
}{
 \vspace{-0.5em}
\begin{align}
    \Pr\left(\hat{I}(\Phi)=I_\text{del}\right)&\to 1\text{ as }n\to\infty.
\end{align}
}
\end{lem}
\begin{proof}
See Appendix~\ref{proof:seededdeletiondetection}.
\end{proof}

    We stress that the remapping $\Phi$ is utilized only on $\mathbf{G}^{(2)}$ to detect the deletions, and is not applied to $\mathbf{Y}$ during the matching process.

\subsection{Row Matching Scheme and Achievability}\label{subsec:matchingschemeWm}

Let $S^n$ be the underlying column repetition pattern and $K_n\triangleq\sum_{j=1}^n S_j$ be the number of columns in $\mathbf{Y}$. The matching scheme (Algorithm~\ref{alg:identicalrepetitionmatching}) we propose follows these steps:
\begin{enumerate}[label=\textbf{ \arabic*)},leftmargin=1.3\parindent]
\item Perform replica detection as in Section~\ref{subsec:replicadetection}. The probability of error in this step is denoted by $\rho_n$.
\item Perform deletion detection using seeds as in Section~\ref{subsec:seededdeletiondetection}. The probability of error is denoted by $\mu_n$. At this step, we have an estimate $\hat{S}^n$ of $S^n$.
\item Using $\hat{S}^n$, place markers between the noisy replica runs of different columns to obtain $\tilde{\mathbf{Y}}$. If a run has length 0, \emph{i.e.} deleted, introduce a column consisting of erasure symbol $\ast\notin\mathfrak{X}$. Note that provided that the detection algorithms in Steps~1 and 2 have performed correctly, there are exactly $n$ such runs, where the $j$\textsuperscript{th} run in $\tilde{\mathbf{Y}}$ corresponds to the noisy copies of the $j$\textsuperscript{th} column of $\sigma_n\circ\mathbf{X}$ if $S_j\neq 0$, and an erasure column otherwise. 
\begin{figure}[t]
\centerline{\includegraphics[width=0.6\textwidth]{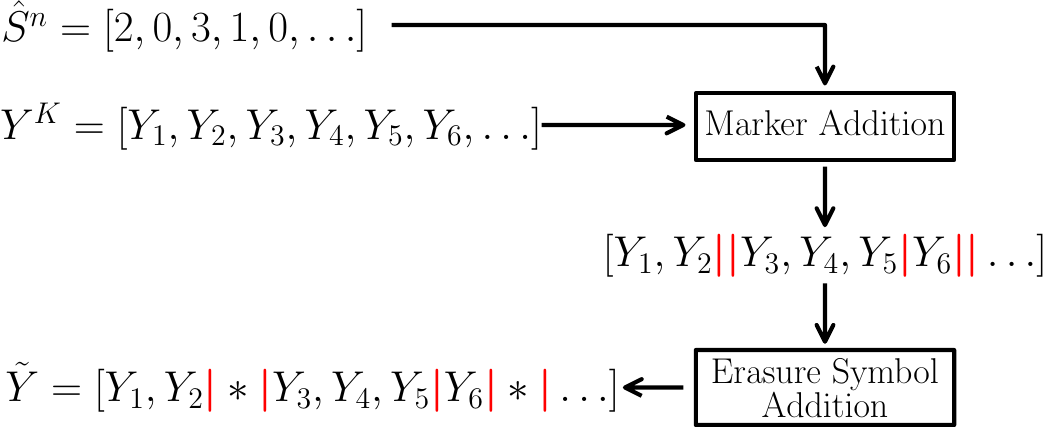}}
\caption{An example of the construction of $\tilde{\mathbf{Y}}$, as described in Step~3 of the proof of Theorem~\ref{thm:mainresultWm} in Section~\ref{subsec:matchingschemeWm}, illustrated over a pair of rows $X^n$ of $\mathbf{X}$ and $Y^K$ of $\mathbf{Y}$. After these steps, in Step~4 we check the joint typicality of the rows $X^n$ of $\mathbf{X}$ and $\tilde{Y}$ of $\tilde{\mathbf{Y}}$.}
\label{fig:marker}
\end{figure}
\item Fix $\epsilon>0$. Match the $l$\textsuperscript{th} row $Y^{K_n}_{l}$ of $\tilde{\mathbf{Y}}$ with the $i$\textsuperscript{th} row $X^n_i$ of $\mathbf{X}$ if $X_i^n$ is the only row of $\mathbf{X}$ jointly $\epsilon$-typical with $Y^{K_n}_l$ according to $p_{X^n,Y^{K_n},S^n}$,
where $S_i\overset{\text{iid}}{\sim} p_S$ and ${Y^{K_n}=Y^{S_1}_1,\dots,Y^{S_n}_n}$ such that
\iftoggle{singlecolumn}{
 \begin{align}
    p_{X^n,{Y}^K|S^n}(x^n,y^k|s^n)&=p_{X^n}(x^n) \prod\limits_{i: s_i>0}\left(\prod\limits_{j=1}^{s_i} p_{Y|X}((y^{s_i})_j|x_i)\right)
       \prod\limits_{i: s_i=0} \mathbbm{1}_{[y^{s_i} = \ast]} 
\end{align}
}{
 \begin{align}
    p&_{X^n,{Y}^K|S^n}(x^n,y^k|s^n)\notag\\&=p_{X^n}(x^n) \prod\limits_{i: s_i>0}\left(\prod\limits_{j=1}^{s_i} p_{Y|X}((y^{s_i})_j|x_i)\right)\notag\\
    &\hspace{4.35em}   \prod\limits_{i: s_i=0} \mathbbm{1}_{[y^{s_i} = \ast]} 
\end{align}
}
with $y^k=y^{s_1}\dots y^{s_n}$.
Assign $\hat\sigma_n(i)=l$. If there is no such jointly typical row, or there is more than one, declare an error.
\end{enumerate}
The runtime of Algorithm~\ref{alg:identicalrepetitionmatching} is $O(m_n^2 n)$ due to the typicality check (each $O(n)$) for all row pairs $(X^n_i,Y_j^{K_n})$ $(i,j)\in[m_n]^2$.

The column discarding and the marker addition as described in Steps~3-4, are illustrated in Figure~\ref{fig:marker}.

We are now ready to prove the achievability of Theorem~\ref{thm:mainresultWm}.

\begin{algorithm}[t]
\caption{Typicality-Based Matching Scheme (Identical Repetition)}\label{alg:identicalrepetitionmatching}
\Input{$(\mathbf{X},\mathbf{Y},\mathbf{P},p_{Y|X},p_S,\epsilon,\mathbf{G}^{(1)},\mathbf{G}^{(2)})$}
\Output{$\hat{\sigma}_n$}
\textup{isReplica}$\gets$ Algorithm~\ref{alg:noisyreplicadetection}($\mathbf{Y},\mathbf{P},p_{Y|X}$)\Comment*[r]{Step 1.}
\textup{isDeleted}$\gets$ Algorithm~\ref{alg:seededdeletiondetection}($\mathbf{G}^{(1)},\mathbf{G}^{(2)},\mathbf{P},p_{Y|X}$,\textup{isReplica})\Comment*[r]{Step 2.}
$\hat{S}^n \gets $ EstimateRepetitionPattern(isReplica,isDeleted)\;
$\tilde{\mathbf{Y}}\gets$ MarkerAddition($\mathbf{Y},\hat{S}^n$)\Comment*[r]{Step 3.}

\For{$i = 1$ \KwTo  \textup{rowSize($\mathbf{X}$)}}{
\textup{count}$\gets 0$\;
\For{$j = 1$ \KwTo  \textup{rowSize($\tilde{\mathbf{Y}}$)}}{
  \If{\textup{isJointlyTypical(}$\mathbf{X}[i][:],\mathbf{Y}[j][:],\mathbf{P},p_{Y|X},p_S,\epsilon$\textup{)}}{
    $\hat{\sigma}_n[i] \gets$ j\;
    count$\gets$ count + 1\;
  }
  
}
\blue{/* count = 0: no row in $\tilde{\mathbf{Y}}$ jointly typical with $\mathbf{X}[i][:]$. ERROR! */}\\
\blue{/* count $>$ 1: multiple rows in $\tilde{\mathbf{Y}}$ jointly typical with $\mathbf{X}[i][:]$. ERROR! */}\\
\If{\textup{count} $\neq 1$}{
$\hat{\sigma}_n[i] \gets 0$\Comment*[r]{Matching error.}
}
}

\end{algorithm}

\begin{proofachievableWm}

From the union bound and Proposition~\ref{prop:jointAEP}, the total probability of error of this scheme (as in~\eqref{eq:proberroridentical}) can be bounded for large $n$ as follows
\iftoggle{singlecolumn}{
\begin{align}
    P_e 
    &\le 2^{n R}  2^{-n(\bar{I}(X;Y^S,S)-3 \epsilon)}+\epsilon+\rho_n+\mu_n\label{eq:perrunion}
\end{align}
}{
\begin{align}
    P_e 
    &\le 2^{n R}  2^{-n(\bar{I}(X;Y^S,S)-3 \epsilon)}+\epsilon+\rho_n+\mu_n\label{eq:perrunion}
\end{align}
}
where $\bar{I}(X;Y^S,S)$ is the mutual information rate~\cite{graymutualinforate} defined as 
\iftoggle{singlecolumn}{
\begin{align}
   \bar{I}(X;Y^S,S)&\triangleq\lim\limits_{n\to\infty} \frac{1}{n} I(X^n;Y^{K_n},S^n)
\end{align}
}{
\begin{align}
   \bar{I}(X;Y^S,S)&\triangleq\lim\limits_{n\to\infty} \frac{1}{n} I(X^n;Y^{K_n},S^n)
\end{align}
}
Note that since $m_n$ is exponential in $n$, from Lemma~\ref{lem:noisyreplicadetection} we have $\rho_n\to0$. Furthermore, since $\Lambda_n=\Theta(\log n)$, from Lemma~\ref{lem:seededdeletiondetection} we have $\mu_n\to0$ as $n\to\infty$. Thus $P_e\le \epsilon$ as $n\to\infty$ if 
\iftoggle{singlecolumn}{
\begin{align}
    R&<\lim\limits_{n\to\infty} \frac{1}{n} I(X^n;Y^{K_n},S^n)
\end{align}
}{
\begin{align}
    R&<\lim\limits_{n\to\infty} \frac{1}{n} I(X^n;Y^{K_n},S^n)
\end{align}
}
concluding the proof of the achievability part.
\end{proofachievableWm}

\subsection{Converse}
\label{subsec:converseWm}
In this subsection, we prove that the database growth rate achieved in Theorem~\ref{thm:mainresultWm} is in fact tight using a genie-aided proof where the column repetition pattern $S^n$ is known. Since the rows are \emph{i.i.d.} conditioned on the repetition pattern $S^n$, the seeds $(\mathbf{G}^{(1)},\mathbf{G}^{(2)})$ do not offer any additional information when $S^n$ is given. Thus, the genie-aided proof holds for any seed size $\Lambda_n$.

\begin{proofconverseWm}
While Theorem~\ref{thm:mainresultWm} is stated for $\Lambda_n=\Omega(\log\log m_n)$, in the converse we assume any seed size $\Lambda_n$. We prove the converse using the modified Fano's inequality presented in~\cite{shirani8849392}.
Let $R$ be the database growth rate and $P_e$ be the probability that the scheme is unsuccessful for a uniformly selected row pair. More formally,
\iftoggle{singlecolumn}{
\begin{align}
   P_e&\triangleq \Pr\left(\sigma_n(J)\neq\hat{\sigma}_n(J)\right),\hspace{1em} J\sim\text{Unif}([m_n])
\end{align}
}{
\begin{align}
   P_e&\triangleq \Pr\left(\sigma_n(J)\neq\hat{\sigma}_n(J)\right),\hspace{1em} J\sim\text{Unif}([m_n])
\end{align}
}
Suppose $P_e\to0$ as $n\to\infty$. Furthermore, let $S^n$ be the repetition pattern and $K_n=\sum_{j=1}^n S_j$. Since $\sigma_n$ is a uniform permutation, from Fano's inequality, we have
\iftoggle{singlecolumn}{
\begin{align}
    H(\sigma_n)&\le 1+m_n P_e\log m_n+I(\sigma_n;\mathbf{X},\mathbf{Y},\mathbf{G}^{(1)},\mathbf{G}^{(2)},S^n)\label{eq:converseWmfirst}
\end{align}
}{
\begin{align}
    H(\sigma_n)\le & 1+m_n P_e\log m_n\notag\\
    &+I(\sigma_n;\mathbf{X},\mathbf{Y},\mathbf{G}^{(1)},\mathbf{G}^{(2)},S^n)\label{eq:converseWmfirst}
\end{align}
}

From the independence of $\mathbf{Y}$, $S^n$, $(\mathbf{G}^{(1)},\mathbf{G}^{(2)})$ and $\sigma_n$, we get
\iftoggle{singlecolumn}{
\begin{align}
    I(\sigma_n;\mathbf{X},\mathbf{Y},\mathbf{G}^{(1)},\mathbf{G}^{(2)},S^n)
    &= I(\sigma_n;\mathbf{X}|\mathbf{Y},\mathbf{G}^{(1)},\mathbf{G}^{(2)},S^n)\\
    &\le I(\sigma_n,\mathbf{Y},\mathbf{G}^{(1)},\mathbf{G}^{(2)},S^n;\mathbf{X})\\
    &\le I(\sigma_n,\mathbf{Y},S^n;\mathbf{X})\label{eq:converseassumeS}\\
    &= I(\sigma_n,\mathbf{Y};\mathbf{X}|S^n)\\
    &=\sum\limits_{i=1}^{m_n} I(X_i^n;Y_{\sigma_n(i)}^{K_n}|S^n) \label{eq:converseWmindeprows}\\
    &= m_n  I(X^n;Y^{K_n}|S^n)\label{eq:converseWmidrows}\\
    &= m_n  I(X^n;Y^{K_n},S^n)\label{eq:converseWmidrows2}
\end{align}
}{
\begin{align}
    I(\sigma_n;\mathbf{X}&,\mathbf{Y},\mathbf{G}^{(1)},\mathbf{G}^{(2)},S^n)
    \notag\\&= I(\sigma_n;\mathbf{X}|\mathbf{Y},\mathbf{G}^{(1)},\mathbf{G}^{(2)},S^n)\\
    &\le I(\sigma_n,\mathbf{Y},\mathbf{G}^{(1)},\mathbf{G}^{(2)},S^n;\mathbf{X})\\
    &\le I(\sigma_n,\mathbf{Y},S^n;\mathbf{X})\label{eq:converseassumeS}\\
    &= I(\sigma_n,\mathbf{Y};\mathbf{X}|S^n)\\
    &=\sum\limits_{i=1}^{m_n} I(X_i^n;Y_{\sigma_n(i)}^{K_n}|S^n) \label{eq:converseWmindeprows}\\
    &= m_n  I(X^n;Y^{K_n}|S^n)\label{eq:converseWmidrows}\\
    &= m_n  I(X^n;Y^{K_n},S^n)\label{eq:converseWmidrows2}
\end{align}
}
where \eqref{eq:converseassumeS} follows from the fact that given the repetition pattern $S^n$, the seeds $(\mathbf{G}^{(1)},\mathbf{G}^{(2)})$ do not offer any additional information on $\sigma_n$. Equation~\eqref{eq:converseWmindeprows} follows from the conditional independence of the non-matching rows given $S^n$. Equation~\eqref{eq:converseWmidrows} follows from the fact that the matching rows are identically distributed conditioned on the repetition pattern ${S}^n$. Finally, \eqref{eq:converseWmidrows2} follows from the independence of $X^n$ and $S^n$.

Note that from Stirling's approximation~\cite[Chapter 3.2]{cormen2022introduction} and the uniformity of $\sigma_n$, we get
\iftoggle{singlecolumn}{
\begin{align}
    H(\sigma_n)&= \log m_n!\\
    &= m_n\log m_n - m_n \log e + O(\log m_n)
\end{align}
\begin{align}
\lim\limits_{n\to\infty}\frac{1}{m_n n}H(\sigma_n)&=\lim\limits_{n\to\infty}\frac{1}{m_n n} \left[m_n\log m_n - m_n \log e + O(\log m_n)\right]\\
        &= \lim\limits_{n\to\infty}\frac{1}{n}\log m_n\\
        &= R \label{eq:converseWmlast}
\end{align}
}{
\begin{align}
    H(\sigma_n)&= \log m_n!\\
    &= m_n\log m_n - m_n \log e + O(\log m_n)
\end{align}
\begin{align}
\lim\limits_{n\to\infty}\frac{1}{m_n n}H(\sigma_n)&=\lim\limits_{n\to\infty}\frac{1}{m_n n} [m_n\log m_n \notag\\&\hspace{2.5em} - m_n \log e + O(\log m_n)]\\
        &= \lim\limits_{n\to\infty}\frac{1}{n}\log m_n\\
        &= R \label{eq:converseWmlast}
\end{align}
}
Finally, from \eqref{eq:converseWmfirst}-\eqref{eq:converseWmlast} we obtain
\iftoggle{singlecolumn}{
\begin{align}
    R&= \lim\limits_{n\to\infty}\frac{1}{m_n n}H(\sigma_n)\\
    &\le \lim\limits_{n\to\infty}\left[ \frac{1}{m_n n}+P_e R+\frac{1}{n} I(X^n;Y^{K_n},S^n)\right]\\
    &= \lim\limits_{n\to\infty}\frac{I(X^n;Y^{K_n},S^n)}{n} \label{eqn:converselast}
\end{align}
}{
\begin{align}
    R&= \lim\limits_{n\to\infty}\frac{1}{m_n n}H(\sigma_n)\\
    &\le \lim\limits_{n\to\infty}\left[ \frac{1}{m_n n}+P_e R+\frac{1}{n} I(X^n;Y^{K_n},S^n)\right]\\
    &= \lim\limits_{n\to\infty}\frac{I(X^n;Y^{K_n},S^n)}{n} \label{eqn:converselast}
\end{align}
}
where \eqref{eqn:converselast} follows from the fact that $P_e\to 0$ as $n\to\infty$.
\end{proofconverseWm}

\subsection{Noiseless Setting}
\label{subsec:noiselessWm}
Lemmas~\ref{lem:noisyreplicadetection} and \ref{lem:seededdeletiondetection} state that given a seed size $\Lambda_n$ double logarithmic with the row size $m_n$, the repetition pattern can be inferred through the aforementioned replica and deletion detection algorithms for any noise distribution $p_{Y|X}$. Thus, the results of Section~\ref{subsec:replicadetection} through Section~\ref{subsec:matchingschemeWm} trivially apply to the noiseless setting where
\iftoggle{singlecolumn}{
\begin{align}
    p_{Y|X}(y|x) &=\mathbbm{1}_{[y=x]}\:\forall (x,y)\in\mathfrak{X}^2.
\end{align}
}{
\begin{align}
    p_{Y|X}(y|x) &=\mathbbm{1}_{[y=x]}\:\forall (x,y)\in\mathfrak{X}^2.
\end{align}
}

We note that when there is no noise, the capacity expression of Theorem~\ref{thm:mainresultWm} (Equation \ref{eq:matchingcap}) can be further simplified as
\iftoggle{singlecolumn}{
\begin{align}
    C &= (1-\delta)^2 \sum\limits_{r=0}^\infty \delta^r H(X_0|X_{-r-1}).
\end{align}
}{
\begin{align}
    C &= (1-\delta)^2 \sum\limits_{r=0}^\infty \delta^r H(X_0|X_{-r-1}).
\end{align}
}

In this subsection, we show that in the noiseless setting, seeds can be made obsolete by the use of a novel detection algorithm.
In other words, in the noiseless setting, we show that Theorem~\ref{thm:mainresultWm} can be extended to any seed size $\Lambda_n$. 

\begin{thm}{\textbf{(Noiseless Matching Capacity for Identical Repetition})}\label{thm:noiselesscapacityWm}
Consider a probability transition matrix $\mathbf{P}$ and a repetition probability distribution $p_S$. Suppose there is no noise, \emph{i.e.,}
\iftoggle{singlecolumn}{
\begin{align}
    p_{Y|X}(y|x) &=\mathbbm{1}_{[y=x]}\:\forall (x,y)\in\mathfrak{X}^2.
\end{align}
}{
\begin{align}
    p_{Y|X}(y|x) &=\mathbbm{1}_{[y=x]}\:\forall (x,y)\in\mathfrak{X}^2.
\end{align}
}
Then, the matching capacity under identical repetition is
\iftoggle{singlecolumn}{
\begin{align}
    C &= (1-\delta)^2 \sum\limits_{r=0}^\infty \delta^r H(X_0|X_{-r-1})\label{eq:noiselesscapacity}
\end{align}
}{
\begin{align}
    C &= (1-\delta)^2 \sum\limits_{r=0}^\infty \delta^r H(X_0|X_{-r-1})\label{eq:noiselesscapacity}
\end{align}
}
for any seed size $\Lambda_n$. Here $\delta\triangleq p_S(0)$ is the deletion probability and $H(X_0|X_{-r-1})$ is the conditional entropy associated with the probability transition matrix
\iftoggle{singlecolumn}{
\begin{align}
    \mathbf{P}^{r+1}&=\gamma^{r+1} \mathbf{I}+(1- \gamma^{r+1}) \mathbf{U} \label{eq:Ppower}
\end{align}
}{
\begin{align}
    \mathbf{P}^{r+1}&=\gamma^{r+1} \mathbf{I}+(1- \gamma^{r+1}) \mathbf{U} \label{eq:Ppower}
\end{align}
}
The capacity can further be simplified as
\iftoggle{singlecolumn}{
\begin{align}
    C &= \frac{(1-\delta)(1-\gamma)}{(1-\gamma\delta)} [H(\pi)+\sum\limits_{i\in\mathfrak{X}} u_i^2\log u_i]- (1-\delta)^2 \sum\limits_{r=0}^\infty \delta^r \sum\limits_{i\in \mathfrak{X}} u_i \eta_{r,i} \log \eta_{r,i}\label{eq:thm2eval}
\end{align}
}{
\begin{align}
    C&= \frac{(1-\delta)(1-\gamma)}{(1-\gamma\delta)} [H(\pi)+\sum\limits_{i\in\mathfrak{X}} u_i^2\log u_i]\notag\\&\hspace{1.75em}- (1-\delta)^2 \sum\limits_{r=0}^\infty \delta^r \sum\limits_{i\in \mathfrak{X}} u_i \eta_{r,i} \log \eta_{r,i}\label{eq:thm2eval}
\end{align}
}
where
\iftoggle{singlecolumn}{
\begin{align}
    \eta_{r,i}\triangleq (1-u_i) \gamma^{r+1}+ u_i.
\end{align}
}{
\begin{align}
    \eta_{r,i}\triangleq (1-u_i) \gamma^{r+1}+ u_i.
\end{align}
}
\end{thm}

\begin{cor}{\textbf{(Noiseless Matching Capacity for Identical Repetition with I.I.D. Database Entries)}}
    When $\gamma=0$, resulting in an \emph{i.i.d.} database distribution $p_X(x)=u_x$, $\forall x\in\mathfrak{X}$, the matching capacity in the noiseless setting is 
    \begin{align}
        C&= (1-\delta) H(X)
    \end{align}
    where $H(X)=H(\pi)$ is the entropy of the stationary distribution $\pi=[u_1,\dots,u_{|\mathfrak{X}|}]$.
\end{cor}

Observe that the RHS of~\eqref{eq:noiselesscapacity} is the mutual information rate for an erasure channel with erasure probability $\delta$ with first-order Markov $(\mathbf{P})$ inputs, as stated in~\cite[Corollary II.2]{li2014input}. Thus, Theorem~\ref{thm:noiselesscapacityWm} states that we can achieve the erasure bound which assumes prior knowledge of the column repetition pattern.

The proof of Theorem~\ref{thm:noiselesscapacityWm} hinges on the observation that in the noiseless setting deletion and replica detection can be performed without seeds. Inspired by the idea of extracting permutation-invariant features as done in Section~\ref{subsec:replicadetection}, our noiseless repetition detection algorithm uses the histogram (and equivalently the type) of each column of $\mathbf{X}$ and $\mathbf{Y}$ as the permutation-invariant feature. Our repetition detection algorithm works as follows: First, for tractability, we \say{collapse} the Markov chain into a binary-valued one. We pick a symbol $x$ from the alphabet $\mathfrak{X}$, WLOG $x=1$, and define the \emph{collapsed} databases $\tilde{\mathbf{X}}$ and $\tilde{\mathbf{Y}}$ as follows:

\iftoggle{singlecolumn}{
\begin{align}
       \tilde{\mathbf{M}}_{i,j} &= \begin{cases}
    1 & \text{if } {\mathbf{M}}_{i,j} = 1\\
    2 & \text{if } {\mathbf{M}}_{i,j} \neq 1
    \end{cases}, \: \forall (i,j),\: \mathbf{M}\in\{\mathbf{X},\mathbf{Y}\} \label{eq:collapsedb}
\end{align}
}{
\begin{align}
    \tilde{\mathbf{M}}_{i,j} &= \begin{cases}
    1 & \text{if } {\mathbf{M}}_{i,j} = 1\\
    2 & \text{if } {\mathbf{M}}_{i,j} \neq 1
    \end{cases}, \: \forall (i,j),\: \mathbf{M}\in\{\mathbf{X},\mathbf{Y}\} \label{eq:collapsedb}
\end{align}
}

Next, we construct the collapsed histogram vectors $\tilde{{H}}^{(1),n}$ and $\tilde{{H}}^{(2),{K_n}}$ as

\iftoggle{singlecolumn}{
\begin{align}
    \tilde{H}_j^{(1)}&=\sum\limits_{i=1}^{m_n} \mathbbm{1}_{\left[\tilde{X}_{i,j}=2 \right]},\:
    \forall j\in [n]
    \label{eq:histogramdefn}\\
    \tilde{H}_j^{(2)}&=\sum\limits_{i=1}^{m_n} \mathbbm{1}_{\left[\tilde{Y}_{i,j}=2 \right]},\:
    \forall j\in [K_n]\label{eq:histogramdefn2}
\end{align}
}{
\begin{align}
    \tilde{H}_j^{(1)}&=\sum\limits_{i=1}^{m_n} \mathbbm{1}_{\left[\tilde{X}_{i,j}=2 \right]},\:
    \forall j\in [n]
    \label{eq:histogramdefn}\\
    \tilde{H}_j^{(2)}&=\sum\limits_{i=1}^{m_n} \mathbbm{1}_{\left[\tilde{Y}_{i,j}=2 \right]},\:
    \forall j\in [K_n]\label{eq:histogramdefn2}
\end{align}
}

Then, the algorithm declares the $j$\textsuperscript{th} column deleted if $\tilde{H}^{(1)}_j$ is absent in $\tilde{{H}}^{(2),{K_n}}$ and declares the $j$\textsuperscript{th} column replicated $s$ times if $\tilde{H}^{(1)}_j$ is present $s\ge 1$ times in $\tilde{{H}}^{(2),{K_n}}$.

Note that as long as column histograms $\tilde{H}^{(1)}_j$ of the collapsed database $\tilde{\mathbf{X}}$ are unique, this detection process is error-free.

The following lemma provides conditions for the asymptotic uniqueness of column histograms ${\tilde{H}_j^{(1)}}$, ${j\in[n]}$.

\begin{lem}{\textbf{(Asymptotic Uniqueness of the Column Histograms)}}\label{lem:histogram}
Let $\tilde{H}^{(1)}_j$ denote the histogram of the $j$\textsuperscript{th} column of $\tilde{\mathbf{X}}$, as in~\eqref{eq:histogramdefn}.
Then, for $m_n=\omega(n^4)$, we have
\iftoggle{singlecolumn}{
\begin{align}
    \Pr\left(\exists i,j\in [n],\: i\neq j,\tilde{H}^{(1)}_i=\tilde{H}^{(1)}_j\right)\to 0 \text{ as }n\to \infty.
\end{align}
}{
\begin{align}
    \Pr\left(\exists i,j\in [n],\: i\neq j,\tilde{H}^{(1)}_i=\tilde{H}^{(1)}_j\right)\to 0 \text{ as }n\to \infty.
\end{align}
}
\end{lem}
\begin{proof}
See Appendix~\ref{proof:histogram}.
\end{proof}

When the databases are not collapsed, the order relation given in Lemma~\ref{lem:histogram} can be tightened. See Section~\ref{subsec:zerorateWm} for more details.

Note that by Definition~\ref{defn:dbgrowthrate}, the row size $m_n$ is exponential in the column size $n$ and the order relation of Lemma~\ref{lem:histogram} is automatically satisfied. 

Next, we present the proof of the achievability part of Theorem~\ref{thm:noiselesscapacityWm} via Algorithm~\ref{alg:noiselessidenticalrepetition}.

\begin{algorithm}[H]
\begin{singlespace}
\caption{Typicality-Based Matching Scheme (Identical Repetition, Noiseless Setting)}\label{alg:noiselessidenticalrepetition}
\Input{$(\mathbf{X},\mathbf{Y},\mathbf{P},p_S,\epsilon)$}
\Output{$\hat{\sigma}_n$}
$(\tilde{\mathbf{X}},\tilde{\mathbf{Y}})\gets$ CollapseDatabases($\mathbf{X},\mathbf{Y}$)\Comment*[r]{Eq.~\eqref{eq:collapsedb}.}
$(\tilde{\mathbf{H}}^{(1)},\tilde{\mathbf{H}}^{(2)})\gets$ ColumnHistograms($\tilde{\mathbf{X}},\tilde{\mathbf{Y}}$)\Comment*[r]{Eq.~\eqref{eq:histogramdefn}.}
\blue{/* Histogram-based repetition detection */}\\
\For{$i=1$ \KwTo \textup{columnSize($\tilde{\mathbf{H}}^{(1)}$)}}{
count$\gets 0$\;
\For{$j=1$ \KwTo \textup{columnSize($\tilde{\mathbf{H}}^{(2)}$)}}{
\If{$\tilde{\mathbf{H}}^{(2)}[:][j]=\tilde{\mathbf{H}}^{(1)}[:][i]$}{
count$\gets$ count + 1\;
}
}
$\hat{S}[i]\gets $ count\;
}
\blue{/* Erasure symbol addition \& Extra replica removal */}\\
\For{$j=1$ \KwTo \textup{columnSize(}$\mathbf{X}$\textup{)}}{
\eIf{$\hat{S}[j]=0$}{
$\bar{\mathbf{Y}}[:][j]\gets\ast$\;
}{
$\bar{\mathbf{Y}}[:][j]\gets \mathbf{X}[:][j]$\;
}
}

\blue{/* Typicality matching w.r.t. erasure channel */}\\
\For{$i = 1$ \KwTo  \textup{rowSize($\mathbf{X}$)}}{
\textup{count}$\gets 0$\;
\For{$j = 1$ \KwTo  \textup{rowSize($\mathbf{Y}$)}}{
  \If{\textup{isJointlyTypical2(}$\mathbf{X}[i][:],\bar{\mathbf{Y}}[j][:],\mathbf{P},p_S,\epsilon$\textup{)}}{
    $\hat{\sigma}_n[i] \gets$ j\;
    count$\gets$ count + 1\;
  }
  
}
\blue{/* count = 0: no row in $\bar{\mathbf{Y}}$ jointly typical with $\mathbf{X}[i][:]$. ERROR! */}\\
\blue{/* count $>$ 1: multiple rows in $\bar{\mathbf{Y}}$ jointly typical with $\mathbf{X}[i][:]$. ERROR! */}\\
\If{\textup{count} $\neq 1$}{
$\hat{\sigma}_n[i] \gets 0$\Comment*[r]{Matching error.}
}
}
\end{singlespace}
\end{algorithm}

\begin{proofachievablenoiseless}
Let $S^n$ be the underlying repetition pattern and ${K_n}\triangleq\sum_{j=1}^n S_j$ be the number of columns in $\mathbf{Y}$. Our matching scheme consists of the following steps:
\begin{enumerate}[label=\textbf{\arabic*)},leftmargin=1.3\parindent]
    \item Construct the collapsed histogram vectors $\tilde{{H}}^{(1),n}$ and $\tilde{{H}}^{(2),{K_n}}$ as in~\eqref{eq:histogramdefn}.
\item Check the uniqueness of the entries $\tilde{H}^{(1)}_j$ $j\in[n]$ of $\tilde{{H}}^{(1),n}$. If there are at least two that are identical, declare a \emph{detection error} whose probability is denoted by $\mu_n$. Otherwise, proceed with Step~3.
\item If $\tilde{H}^{(1)}_j$ is absent in $\tilde{{H}}^{(2),{K_n}}$, declare it deleted, assigning $\hat{S}_j=0$. Note that, conditioned on the uniqueness of the column histograms $\tilde{H}^{(1)}_j$ $\forall j\in[n]$, this step is error-free.
\item If $\tilde{H}^{(1)}_j$ is present $s\ge 1$ times in $\tilde{{H}}^{(2),{K_n}}$ , assign $\hat{S}_j=s$. Again, if there is no detection error in Step~2, this step is error-free. Note that at the end of this step, provided there are no detection errors, we recover $S^n$, \emph{i.e.}, $\hat{{S}}^n={S}^n$.
\item Based on $\hat{{S}}^n$, $\mathbf{X}$ and $\mathbf{Y}$, construct $\bar{\mathbf{Y}}$ as the following:
\begin{itemize}
    \item If $\hat{S}_j = 0$, the $j$\textsuperscript{th} column of $\bar{\mathbf{Y}}$ is a column consisting of erasure symbol $\ast\notin\mathfrak{X}$.
    \item If $\hat{S}_j \ge 1$, the $j$\textsuperscript{th} column of $\bar{\mathbf{Y}}$ is the $j$\textsuperscript{th} column of $\mathbf{X}$.
\end{itemize}
Note that after the removal of the additional replicas and the introduction of the erasure symbols, $\bar{\mathbf{Y}}$ has $n$ columns.
\item Fix $\epsilon>0$. Let $q_{\bar{Y}|X}$ be the probability transition matrix of an erasure channel with erasure probability $\delta$, that is $\forall (x,\bar{y})\in\mathfrak{X}\times(\mathfrak{X}\cup \{\ast\})$
\iftoggle{singlecolumn}{
\begin{align}
    q_{\bar{Y}|X}(\bar{y}|x) &= \begin{cases}
    1-\delta &\text{if }\bar{y}=x\\
    \delta &\text{if }\bar{y}=\ast
    \end{cases}. \label{eq:erasure}
\end{align}
}{
\begin{align}
    q_{\bar{Y}|X}(\bar{y}|x) &= \begin{cases}
    1-\delta &\text{if }\bar{y}=x\\
    \delta &\text{if }\bar{y}=\ast
    \end{cases}. \label{eq:erasure}
\end{align}
}
We consider the input to the memoryless erasure channel as the $i$\textsuperscript{th} row $X^n_i$ of $\mathbf{X}$. The output $\bar{Y}^n$ is the matching row of $\bar{\mathbf{Y}}$. For our row matching algorithm, we match the $l$\textsuperscript{th} row $\bar{{Y}}^n_{l}$ of $\bar{\mathbf{Y}}$ with the $i$\textsuperscript{th} row $X^n_i$ of $\mathbf{X}$, if $X^n_i$ is the only row of $\mathbf{X}$ jointly $\epsilon$-typical~\cite[Chapter 3]{cover2006elements} with $\bar{{Y}}^n_l$ with respect to $p_{X^n,Y^n}$, where
\iftoggle{singlecolumn}{
\begin{align}
    p_{X^n,\bar{Y}^n}(x^n,\bar{y}^n) &= p_{X^n}(x^n) \prod\limits_{j=1}^n q_{Y|X}(\bar{y}_j|x_j)\label{eq:markovinput}
\end{align}
}{
\begin{align}
    p_{X^n,\bar{Y}^n}(x^n,\bar{y}^n) &= p_{X^n}(x^n) \prod\limits_{j=1}^n q_{Y|X}(\bar{y}_j|x_j)\label{eq:markovinput}
\end{align}
}
where $X^n$ denotes the Markov chain of length $n$ with probability transition matrix $\mathbf{P}$. This results in $\hat\sigma_n(i)=l$. 
Otherwise, declare \emph{collision error}.
\end{enumerate}

Similar to~\eqref{eq:perrunion}, from the union bound and Proposition~\ref{prop:jointAEP}, the total probability of error of this scheme can be bounded for large $n$ as follows
\iftoggle{singlecolumn}{
\begin{align}
    P_e &\le \mu_n + \epsilon + 2^{n(R-\bar{I}(X;\bar{Y})+3\epsilon)}
\end{align}
}{
\begin{align}
    P_e &\le \mu_n + \epsilon + 2^{n(R-\bar{I}(X;\bar{Y})+3\epsilon)}
\end{align}
}

Since $m_n$ is exponential in $n$, by Lemma~\ref{lem:histogram}, ${\mu_n\to0}$ as ${n\to\infty}$. Thus
\iftoggle{singlecolumn}{
\begin{align}
    P_e&< 3 \epsilon \text{ as }n\to\infty
\end{align}
}{
\begin{align}
    P_e&< 3 \epsilon \text{ as }n\to\infty
\end{align}
}
if 
$R<\bar{I}(X;\bar{Y})-3\epsilon$. Thus, we can argue that any database growth rate $R$ satisfying
\iftoggle{singlecolumn}{
\begin{align}
    R&<\bar{I}(X;\bar{Y})\label{eq:achievable}
\end{align}
}{
\begin{align}
    R&<\bar{I}(X;\bar{Y})\label{eq:achievable}
\end{align}
}
is achievable, by taking $\epsilon$ small enough. From~\cite[Corollary II.2]{li2014input} we have
\iftoggle{singlecolumn}{
\begin{align}
    \bar{I}(X;\bar{Y})&=(1-\delta)^2 \sum\limits_{r=0}^\infty \delta^r H(X_0|X_{-r-1}) \label{eq:MIrate}
\end{align}
}{
\begin{align}
    \bar{I}(X;\bar{Y})&=(1-\delta)^2 \sum\limits_{r=0}^\infty \delta^r H(X_0|X_{-r-1}) \label{eq:MIrate}
\end{align}
}
where $H(X_0|X_{-r-1})$ is the conditional entropy associated with the probability transition matrix $\mathbf{P}^{r+1}$. 

Now, we argue that \eqref{eq:Ppower} can be proven via induction on $r$ by taking \eqref{eq:markovtransitionmatrix} as a base case and observing that $\mathbf{U}^2 = \mathbf{U}$. Finally, plugging $\pi$ and $\mathbf{P}^{r+1}$ directly into~\cite[Theorem 4.2.4]{cover2006elements} yields \eqref{eq:thm2eval}, concluding the achievability part of the proof.
\end{proofachievablenoiseless}

Next, we move on to prove the converse part of Theorem~\ref{thm:noiselesscapacityWm}.
\begin{proofconversenoiseless}
Since the converse part of Theorem~\ref{thm:mainresultWm} holds for any seed size $\Lambda_n$, in the noiseless setting, we trivially have
\iftoggle{singlecolumn}{
\begin{align}
    C&\le \lim\limits_{n\to\infty} \frac{I(X^n;Y^{K_n},S^n)}{n}.
\end{align}
}{
\begin{align}
    C&\le \lim\limits_{n\to\infty} \frac{I(X^n;Y^{K_n},S^n)}{n}.
\end{align}
}

Next, note that there is a bijective mapping between $(Y^{K_n},S^n)$ and $(\bar{{Y}}^n,{S}^n)$. Therefore, we have
\iftoggle{singlecolumn}{
\begin{align}
    I({X}^n;{Y}^{K_n},{S}^n) &= I({X}^n;\bar{{Y}}^n,{S}^n)\label{eq:noadditionalinfo1}\\
    &= I({X}^n;\bar{{Y}}^n) + I({X}^n;{S}^n|\bar{{Y}}^n)\\
    &= I({X}^n;\bar{{Y}}^n) \label{eq:noadditionalinfo2}
\end{align}
}{
\begin{align}
    I({X}^n;{Y}^{K_n},{S}^n) &= I({X}^n;\bar{{Y}}^n,{S}^n)\label{eq:noadditionalinfo1}\\
    &= I({X}^n;\bar{{Y}}^n) + I({X}^n;{S}^n|\bar{{Y}}^n)\\
    &= I({X}^n;\bar{{Y}}^n) \label{eq:noadditionalinfo2}
\end{align}
}
where \eqref{eq:noadditionalinfo2} follows from the independence of ${S}^n$ and ${X}^n$ conditioned on $\bar{{Y}}^n$. This is because since $\bar{{Y}}^n$ is stripped of all extra replicas, from $(X^n,\bar{{Y}}^n)$ we can only infer the zeros of $S^n$, which is already known through $\bar{{Y}}^n$ via erasure symbols. Thus, we have
\iftoggle{singlecolumn}{
\begin{align}
    C &\le \bar{I}(X;\bar{Y})
\end{align}
}{
\begin{align}
    C &\le \bar{I}(X;\bar{Y})
\end{align}
}
where $\bar{I}(X;\bar{Y})$ is defined in~\eqref{eq:MIrate}, concluding the proof of the converse part.
\end{proofconversenoiseless}

The runtimes of the histogram-based detection algorithm and the typicality-based matching algorithm (Algorithm~\ref{alg:noiselessidenticalrepetition}) are $O(m_n n)$ and $O(m_n^2 n)$, respectively.

\section{Matching Capacity For Independent Repetition}
\label{sec:matchingcapacityW1}

In this section, we investigate the upper and the lower bounds on the matching capacity $C$ for independent repetition, where we assume a repetition pattern that is independent across all rows. For tractability, we focus on the special case where $\gamma=0$, resulting in an \emph{i.i.d.} database distribution $p_X(x)=u_x$, $\forall x\in\mathfrak{X}$.

We state our main result on the matching capacity for independent repetition in the following theorem: 
\begin{thm}{\textbf{(Matching Capacity Bounds for Independent Repetition})}\label{thm:mainresultW1}
Consider a probability transition matrix $\mathbf{P}$ with $\gamma=0$, a noise distribution $p_{Y|X}$ and a repetition distribution $p_S$.  Then the matching capacity satisfies
\iftoggle{singlecolumn}{
\begin{align}
    C&\ge\left[\frac{\mathbb{E}[S]}{s_{\max}} H(X)-(1-\alpha\delta)H_b\left( \frac{\mathbb{E}[S]}{(1-\alpha\delta)s_{\max}}\right)-\mathbb{E}[S]H(X|Y)\right]^+\label{eq:rowwiseachievable}\\
    C&\le \inf\limits_{n\ge 1}\frac{1}{n} I({X}^n;{Y}^{K_n},{A}^n)\label{eq:rowwiseconverse}
\end{align}
}{
\begin{align}
    C&\ge\Big[\frac{\mathbb{E}[S]}{s_{\max}} H(X)-\mathbb{E}[S]H(X|Y)\notag\\&\hspace{2em}-(1-\alpha\delta)H_b\left( \frac{\mathbb{E}[S]}{(1-\alpha\delta)s_{\max}}\right)\Big]^+\label{eq:rowwiseachievable}\\
    C&\le \inf\limits_{n\ge 1}\frac{1}{n} I({X}^n;{Y}^{K_n},{A}^n)\label{eq:rowwiseconverse}
\end{align}
}
where
\begin{align}
    \Pr(Y^{K_n} = y^{K_n}|X^n=x^n)
    &= \prod\limits_{j:S_j\neq 0} \Pr((Y_{K_{j-1}+1},\dots,Y_{K_j})=(y_{K_{j-1}+1},\dots,y_{K_j})|X_{j}=x_j)\\
    &= \prod\limits_{j:S_j\neq 0} \prod\limits_{s=1}^{S_j} p_{Y|X}(y_{K_{j-1}+s}|x_j),
\end{align}
 $K_{j}\triangleq \sum\limits_{t=1}^j S_{t}$
, and $\delta$ and $\alpha$ are the deletion and the deletion detection probabilities, respectively and $s_{\max}\triangleq \max \mathrm{supp}(p_S)$.
Furthermore, for repetition distributions with $\frac{1}{s_{\max}}\mathbb{E}[S]\ge \frac{1-\alpha\delta}{|\mathfrak{X}|}$, the lower bound in equation~\eqref{eq:rowwiseachievable} can be tightened as
\iftoggle{singlecolumn}{
\begin{align}
     C\ge\Big[(1-\alpha\delta) H(X)-\Big(1&-\alpha\delta-\frac{\mathbb{E}[S]}{s_{\max}} \Big)\min\{H(X),\log(|\mathfrak{X}|-1)\}\notag\\
     &-(1-\alpha\delta)H_b\left( \frac{\mathbb{E}[S]}{(1-\alpha\delta)s_{\max}}\right)-\mathbb{E}[S]H(X|Y)\Big]^+\label{eq:rowwiseachievable2}
\end{align}
}{
\begin{align}
     C&\ge\Big[(1-\alpha\delta) H(X)-\mathbb{E}[S]H(X|Y)\notag\\
     &-\Big(1-\alpha\delta-\frac{\mathbb{E}[S]}{s_{\max}} \Big)\min\{H(X),\log(|\mathfrak{X}|-1)\}\notag\\
     &-(1-\alpha\delta)H_b\left( \frac{\mathbb{E}[S]}{(1-\alpha\delta)s_{\max}}\right)\Big]^+\label{eq:rowwiseachievable2}
\end{align}
}
\end{thm}

We note that the upper bound given in Theorem~\ref{thm:mainresultW1} (equation~\eqref{eq:rowwiseconverse}) is an infimum over the column size $n$. Therefore, its evaluation for any $n\in\mathbb{N}$ yields an upper bound on the matching capacity.

With independent repetition, we cannot perform repetition detection as in Section~\ref{sec:matchingcapacityWm}, and hence we are restricted to using a single-step rowwise matching scheme as done in~\cite{shirani8849392}. This builds an analogy between database matching and channel decoding. In particular, our approach to database matching for independent repetition is related to decoding in the noisy synchronization channel~\cite{cheraghchi2020overview}.

We stress that there are several important differences between the database matching problem and the synchronization channel literature: \emph{i)} In database matching the database distribution is fixed and cannot be designed or optimized, whereas in channel coding the main goal is to optimize the input distribution to find the channel capacity \emph{ii)} The synchronization channel literature mostly focuses on code design with few works, such as~\cite{diggavi1603788}, focusing on random codebook arguments for only a few types of synchronization errors such as deletion~\cite{diggavi1603788} and duplication~\cite{drinea2007improved} and finally \emph{iii)} Our database matching result provides an achievability argument for all repetition distributions with finite support, whereas the synchronization channel literature mainly focuses on some families of repetition distributions. As a result, for input-constrained noisy synchronization channels, our generalized random codebook argument, presented in Section~\ref{subsec:achievabilityW1}, is novel and might be of independent interest.

In Section~\ref{subsec:achievabilityW1}, we prove the achievability part of Theorem~\ref{thm:mainresultW1} (equation~\eqref{eq:rowwiseachievable}) by proposing a rowwise matching scheme. Then, in Section~\ref{subsec:converseW1} we prove the converse part (equation~\eqref{eq:rowwiseconverse}). Then, we present strictly tighter upper bounds for a special case with only deletions, \emph{i.e.,} when $s_{\max}=1$.

\subsection{Row Matching Scheme and Achievability}
\label{subsec:achievabilityW1}
To prove the achievability, we consider the following matching scheme, also given in Algorithm~\ref{alg:rowwisematchingscheme}:

\begin{algorithm}[t]
\caption{Typicality-Based Matching Scheme (Independent Repetition)}\label{alg:rowwisematchingscheme}
\Input{$(\mathbf{X},\mathbf{Y},\mathbf{A},p_X,p_{Y|X},p_S,\epsilon)$}
\Output{$\hat{\sigma}_n$}

\For{$j = 1$ \KwTo  \textup{rowSize($\mathbf{Y}$)}}{
\textup{count}$\gets 0$\;
\blue{/* Remove revealed deleted columns */}\\
\For{$i=1$ \KwTo \textup{columnSize($\mathbf{A}$)}}{
\eIf{$\mathbf{A}[j][i]=0$}{
$\bar{\mathbf{X}}[:][i] \gets \mathbf{X}[:][i]$\;
}{
$\bar{\mathbf{X}}[:][i] \gets []$\;
}
}
\blue{/* Stretch $\bar{\mathbf{X}}$ $s_{\max}$ times */}\\
\For{$i=1$ \KwTo \textup{columnSize($\bar{\mathbf{X}}$)}}{
$\tilde{\mathbf{X}}[:][(i-1)s_{\max}+1:i s_{\max}]\gets \bar{\mathbf{X}}[:][i]$\;
}

\blue{/* Typical subsequence check (See Appendix~\ref{proof:achievabilityW1}). */}\\
\For{$i=1$ \KwTo \textup{rowSize($\tilde{\mathbf{X}}$)}}{
\If{\textup{isTypicalSubsequence(}$\tilde{\mathbf{X}}[i][:],\mathbf{Y}[j][:],\mathbf{P},p_{Y|X},p_S,\epsilon$\textup{)}}{
    $\hat{\sigma}^{-1}_n[j] \gets$ i\;
    count$\gets$ count + 1\;
  }
}

\blue{/* count = 0: $\mathbf{Y}[j][:]$ is not a typical subsequence of any $\tilde{\mathbf{X}}[i][:]$. ERROR! */}\\
\blue{/* count $>$ 1: $\mathbf{Y}[j][:]$ is a typical subsequence of multiple $\tilde{\mathbf{X}}[i][:]$. ERROR! */}\\
\If{\textup{count} $\neq 1$}{
$\hat{\sigma}^{-1}_n[j] \gets 0$\Comment*[r]{Matching error.}
}

}

\end{algorithm}

\begin{enumerate}[label=\textbf{\arabic*)},leftmargin=1.3\parindent]
    \item Given the $i$\textsuperscript{th} row $Y_i^{K_n}$ of $\mathbf{Y}$ and the corresponding row $A_i^n$ of the partial deletion location information $\mathbf{A}$, we discard the $j$\textsuperscript{th} column of $\mathbf{X}$ if $A_{i,j}=1$, $\forall j\in[n]$ to obtain $\bar{\mathbf{X}}$ since it does not offer any additional information due to the independent nature of the database entries.
    \item We convert the problem into a deletion-only one by elementwise repeating all the columns of $\bar{\mathbf{X}}$ $s_{\max}$ times, which we call \say{stretching by $s_{\max}$}, to obtain $\tilde{\mathbf{X}}$. At this step, $Y_i^{K_n}$ can be seen as the output of the noisy deletion channel where the $\sigma_n^{-1}(i)$\textsuperscript{th} row of $\tilde{\mathbf{X}}$ is the input.
    \item We perform a generalized version of the decoding algorithm introduced in~\cite{bakirtas2021database} for the noiseless deletions with deletion detection probability. Note that the latter itself is an extension of the one proposed in~\cite{diggavi1603788}. 
\end{enumerate} 

Observe that Algorithm~\ref{alg:rowwisematchingscheme} has a runtime of $O(m_n^2 n)$, similar to Algorithms~\ref{alg:identicalrepetitionmatching}-\ref{alg:noiselessidenticalrepetition}.

The full proof of the achievability part (equations~\eqref{eq:rowwiseachievable} and \eqref{eq:rowwiseachievable2}) via the matching scheme described above can be found in Appendix~\ref{proof:achievabilityW1}.

\begin{figure}[t]
\centerline{\includegraphics[width=0.6\textwidth]{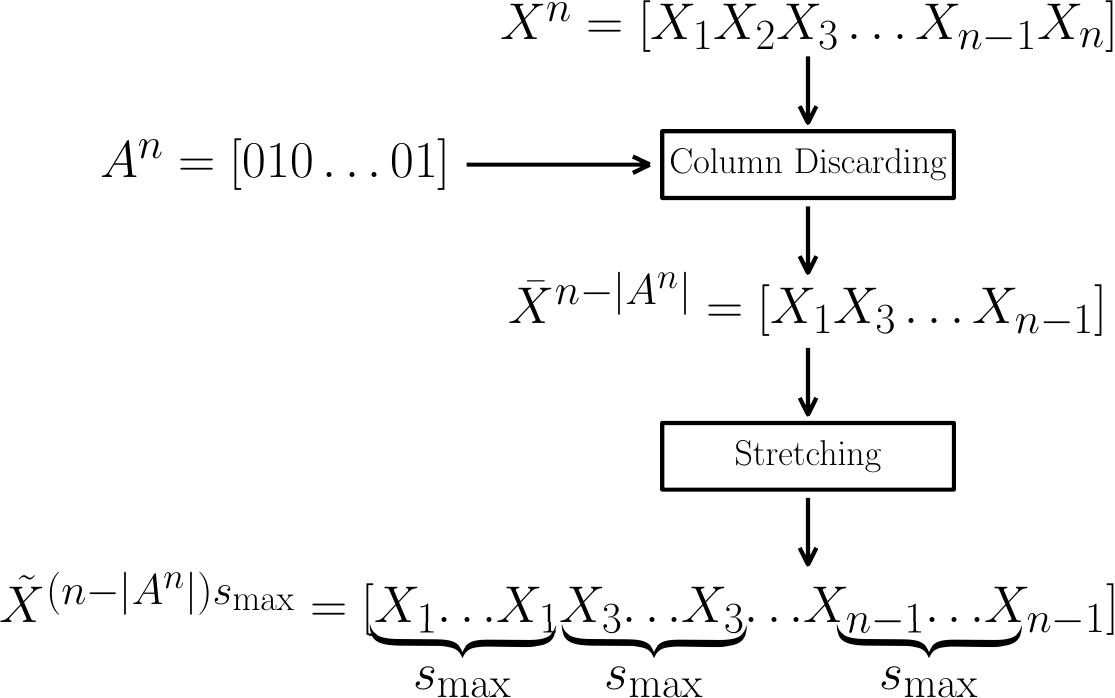}}
\caption{An illustrative example of the column discarding and the stretching of ${X}^n$ into $\tilde{{X}}^{(n-|{A}^n|) s_{\max}}$, for a given the deletion detection pattern ${A}^n$. First, we discard each know deleted element known from $X^n$ to obtain $\bar{{X}}^{n-|{A}^n|}$. Then, each element of $\bar{{X}}^{(n-|{A}^n|)}$ is repeated $s_{\max}$ times to obtain $\tilde{{X}}^{(n-|{A}^n|) s_{\max}}$.}
\label{fig:elemwiserep}
\end{figure}

An illustrative example of the \say{stretching} is given in Figure~\ref{fig:elemwiserep}. The idea behind this stretching is that since each entry can be repeated at most $s_{\max}$ times when we stretch $X^n$ $s_{\max}$ times to obtain $\tilde{X}^{n s_{\max}}$, the output of the synchronization channel (before the noise $p_{Y|X}$) is guaranteed to be a subsequence of $\tilde{X}^{n s_{\max}}$. This way, we can convert the general noisy synchronization problem into a noisy deletion-only problem. We note that when $s_{\max}$ becomes large compared to the alphabet size $|\mathfrak{X}|$, the lower bound given in~\eqref{eq:rowwiseachievable} goes to zero, even when $p_S(s_{\max})$ is very small.

Note that for any repetition structure, including the ones not considered in this work, one can simply ignore the underlying structure and apply the matching scheme described above. Therefore the achievable rate of Theorem~\ref{thm:mainresultW1} (equation~\eqref{eq:rowwiseachievable}) is achievable for any repetition structure.

\subsection{Converse}
\label{subsec:converseW1}
In this subsection, we prove the converse part of Theorem~\ref{thm:mainresultW1} and evaluate the given upper bound for some special cases. First, we observe that by following the genie argument provided in the converse of Theorem~\ref{thm:mainresultWm}, we can argue that Theorem~\ref{thm:mainresultWm} is an upper bound on $C$ for any $\alpha$ and for any repetition structure. 

We next prove the converse of Theorem~\ref{thm:mainresultW1} (equation \eqref{eq:rowwiseconverse}). We then analytically evaluate this for some $n\in\mathbb{N}$ and we argue that the evaluated upper bounds are strictly tighter than that in Theorem~\ref{thm:mainresultWm}.

\begin{proofconverseW1}
We start with the modified Fano's inequality used in Section~\ref{subsec:converseWm}. Let
\iftoggle{singlecolumn}{
\begin{align}
   P_e&\triangleq \Pr\left(\sigma_n(J)\neq\hat{\sigma}_n(J)\right),\hspace{1em} J\sim\text{Unif}([m_n])
\end{align}
}{
\begin{align}
   P_e&\triangleq \Pr\left(\sigma_n(J)\neq\hat{\sigma}_n(J)\right),\hspace{1em} J\sim\text{Unif}([m_n])
\end{align}
}
Then, we have
\iftoggle{singlecolumn}{
\begin{align}
    H(\sigma_n)\le 1+m_n P_e\log m_n+I&(\sigma_n;\mathbf{X},\mathbf{Y},\mathbf{A})
\end{align}
}{
\begin{align}
    H(\sigma_n)\le 1+m_n P_e\log m_n+I&(\sigma_n;\mathbf{X},\mathbf{Y},\mathbf{A})
\end{align}
}
where
\iftoggle{singlecolumn}{
\begin{align}
I(\sigma_n;\mathbf{X},\mathbf{Y},\mathbf{A})
    &= I(\sigma_n;\mathbf{X}|\mathbf{Y},\mathbf{A})\label{eq:converseW1first}\\
    &\le I(\sigma_n,\mathbf{Y},\mathbf{A};\mathbf{X})\label{eq:converseW1seedsnoinfo}\\
    &=\sum\limits_{i=1}^{m_n} I(X_i^n;Y_{\sigma_n(i)}^{K_n},A_{\sigma_n(i)}^n) \label{eq:converseW1indeprows}\\
    &= m_n  I(X^n;Y^{K_n},A^n)\label{eq:converseW1idrows}
\end{align}
}{
\begin{align}
I(\sigma_n;\mathbf{X},\mathbf{Y},\mathbf{A})
    &= I(\sigma_n;\mathbf{X}|\mathbf{Y},\mathbf{A})\label{eq:converseW1first}\\
    &\le I(\sigma_n,\mathbf{Y},\mathbf{A};\mathbf{X})\label{eq:converseW1seedsnoinfo}\\
    &=\sum\limits_{i=1}^{m_n} I(X_i^n;Y_{\sigma_n(i)}^{K_n},A_{\sigma_n(i)}^n) \label{eq:converseW1indeprows}\\
    &= m_n  I(X^n;Y^{K_n},A^n)\label{eq:converseW1idrows}
\end{align}
}
where \eqref{eq:converseW1indeprows} and \eqref{eq:converseW1idrows} follow from the fact that non-matching rows and their corresponding probabilistic side information on deletion locations are respectively independent and identically distributed.
Following similar steps to Section~\ref{subsec:converseWm}, we obtain
\iftoggle{singlecolumn}{
\begin{align}
    R&\le \lim\limits_{n\to\infty} \frac{I(X^n;Y^{K_n},A^n)}{n}
\end{align}
}{
\begin{align}
    R&\le \lim\limits_{n\to\infty} \frac{I(X^n;Y^{K_n},A^n)}{n}
\end{align}
}
whenever $P_e\to 0$ as $n\to\infty$.

Note that from Fekete's lemma~\cite{fekete1923verteilung}, for any subadditive sequence $\{a_n\}_{n\in\mathbb{N}}$, we have
\iftoggle{singlecolumn}{
\begin{align}
    \lim\limits_{n\to\infty}\frac{a_n}{n} = \inf\limits_{n\ge 1}\frac{a_n}{n}
\end{align}
}{
\begin{align}
    \lim\limits_{n\to\infty}\frac{a_n}{n} = \inf\limits_{n\ge 1}\frac{a_n}{n}
\end{align}
}
Therefore, it is sufficient to prove the subadditivity of $I(X^n;Y^{K_n},A^n)$. 

Choose an arbitrary $r\in[n-1]$ and let $M_r\triangleq \sum_{j=1}^r S_j$ where $S^n$ is the repetition pattern through which $Y^{K_n}$ is obtained from $X^n$. Note that $M_r$ denotes a marker, stating which part of $Y^{K_n}$ depends on the first $r$ elements of $X^n$, denoted by $X_1^r$. Therefore we have a bijective relation between $(Y^{K_n},M_r)$ and $(Y_1^{\sum_{j=1}^r S_j},Y_{\sum_{j=1}^r S_j +1}^{K_n})$ where the subscripts and the superscripts denote the starting and the ending points of the vectors, respectively. Thus,
\iftoggle{singlecolumn}{
\begin{align}
    I(X^n;Y^{K_n},A^n)&\le I(X^n;Y^{K_n},M_r,A^n)\\
    &= I(X^n;Y_1^{\sum_{j=1}^r S_j},Y_{\sum_{j=1}^r S_j +1}^{K_n},A^n)\\
    &= I(X_1^r,X_{r+1}^n;Y_1^{\sum_{j=1}^r S_j},Y_{\sum_{j=1}^r S_j +1}^{K_n},A_1^r,A_{r+1}^n)\\
    &= I(X_1^r;Y_1^{\sum_{j=1}^r S_j},A_1^r) +I(X_{r+1}^n;Y_{\sum_{j=1}^r S_j +1}^{K_n},A_{r+1}^n)\label{eq:subadditivityupperbound}
\end{align}
}{
\begin{align}
    I(&X^n;Y^{K_n},A^n)\notag\\&\le I(X^n;Y^{K_n},M_r,A^n)\\
    &= I(X^n;Y_1^{\sum_{j=1}^r S_j},Y_{\sum_{j=1}^r S_j +1}^{K_n},A^n)\\
    &= I(X_1^r,X_{r+1}^n;Y_1^{\sum_{j=1}^r S_j},Y_{\sum_{j=1}^r S_j +1}^{K_n},A_1^r,A_{r+1}^n)\\
    &= I(X_1^r;Y_1^{\sum_{j=1}^r S_j},A_1^r) +I(X_{r+1}^n;Y_{\sum_{j=1}^r S_j +1}^{K_n},A_{r+1}^n)\label{eq:subadditivityupperbound}
\end{align}
}
where \eqref{eq:subadditivityupperbound} follows from the fact that $X^n$ and $A^n$ have \emph{i.i.d.} entries and the noise $p_{Y|X}$ acts independently on the entries. Thus, $I(X^n;Y^{K_n},A^n)$ is a subadditive sequence.
Hence,
\iftoggle{singlecolumn}{
\begin{align}
    R&\le \inf\limits_{n\ge 1} \frac{I(X^n;Y^{K_n},A^n)}{n}\label{eq:W1infimum}
\end{align}
}{
\begin{align}
    R&\le \inf\limits_{n\ge 1} \frac{I(X^n;Y^{K_n},A^n)}{n}\label{eq:W1infimum}
\end{align}
}
whenever $P_e\to 0$ as $n\to\infty$, concluding the proof.
\end{proofconverseW1}

We note that since the upper bound given in Theorem~\ref{thm:mainresultW1} is the infimum over all $n\ge1$, its evaluation at any $n\in\mathbb{N}$ yields an upper bound on the matching capacity. In Corollaries~\ref{cor:W1conversenoiselessub} and \ref{cor:W1conversenoisybinaryub}, we analytically evaluate this upper bound at $n=2$ under some assumptions on $p_{X,Y}$ when $s_{\max}=1$, \emph{i.e.,} when we only have deletions, and explicitly demonstrate the gap between the upper bounds given in Theorem~\ref{thm:mainresultWm} and Theorem~\ref{thm:mainresultW1}.

First, we consider a noiseless deletion setting with arbitrary database distribution $p_X$ in Corollary~\ref{cor:W1conversenoiselessub}.
\begin{sloppypar}
\begin{cor}{\textbf{(Upper Bound for Noiseless Deletion)}}\label{cor:W1conversenoiselessub}
    Consider a noiseless deletion setting where ${p_{Y|X}(y|x) = \mathbbm{1}_{[x=y]}}$, ${\forall(x,y)\in\mathfrak{X}^2}$ and ${S\sim\text{Bernoulli}(1-\delta)}$. Then for any input distribution $p_X$, we have
\iftoggle{singlecolumn}{
\begin{align}
    C &\le \frac{1}{2} I(X^2;Y^K,A^2)\\&= (1-\delta)H(X)-(1-\alpha)\delta(1-\delta)\left(1-\hat{q}\right)\label{eq:W1noiselessub}
\end{align}
}{
\begin{align}
    C &\le \frac{1}{2} I(X^2;Y^K,A^2)\\
    &= (1-\delta)H(X)-(1-\alpha)\delta(1-\delta)\left(1-\hat{q}\right)\label{eq:W1noiselessub}
\end{align}
}
where $\hat{q} \triangleq \sum_{x\in\mathfrak{X}} p_X(x)^2$.
\end{cor}
\end{sloppypar}
\begin{proof}
See Appendix~\ref{proof:rowwiseconversenoiselessub}.
\end{proof}
Note that for any $\mathfrak{X}$ with $|\mathfrak{X}|\ge 2$ and $\alpha\in[0,1)$ the upper bound given in Corollary~\ref{cor:W1conversenoiselessub} is strictly lower than the one provided in Theorem~\ref{thm:mainresultWm} which is
\iftoggle{singlecolumn}{
\begin{align}
    I(X;Y,S) = (1-\delta) H(X).
\end{align}
}{
\begin{align}
    I(X;Y,S) = (1-\delta) H(X).
\end{align}
}

\begin{figure}[t]
\centerline{\includegraphics[width=0.60\textwidth]{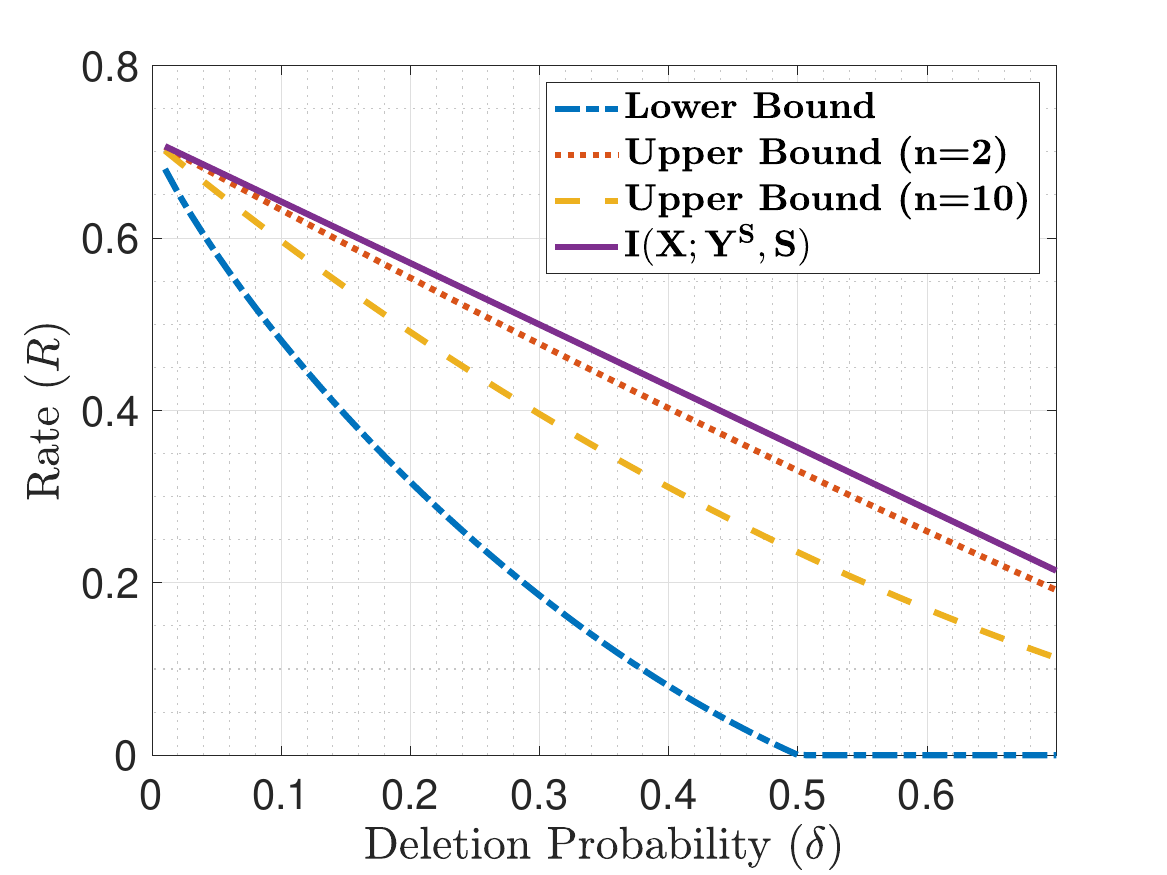}}\caption{The evaluation of the lower and upper bounds on the matching capacity for the binary noisy deletion case with  $p_X\sim$Bernoulli$(\nicefrac{1}{2})$, $p_{S}\sim \text{Bernoulli}(1-\delta)$, $\alpha=0.7$ and $p_{Y|X}\sim \text{BSC}(0.05)$. The dash-dotted (blue) curve is the achievable rate stated in Theorem~\ref{thm:mainresultW1}. The dashed (yellow) and the dotted (red) curves are the evaluations of the upper bound stated in Theorem~\ref{thm:mainresultW1}, at $n=10$ and $n=2$, respectively. The solid (purple) curve shows the loose upper bound given in Theorem~\ref{thm:mainresultWm}. We see that the gap between the lower and the upper bounds shrinks as $n$ increases.}
\label{fig:entryrates}
\end{figure}

Next, we consider a noisy deletion setting with binary $X$ and arbitrary noise $p_{Y|X}$ in Corollary~\ref{cor:W1conversenoisybinaryub}.

\begin{cor}{\textbf{(Upper Bound for Binary Noisy Deletion)}}\label{cor:W1conversenoisybinaryub}
Consider a binary noisy deletion setting where $X\sim\text{Bernoulli(p)}$ and ${S\sim\text{Bernoulli}(1-\delta)}$. Then, for any binary DMC $p_{Y|X}$, we have
\iftoggle{singlecolumn}{
\begin{align}
    C &\le \frac{1}{2} I(X^2;Y^K,A^2) \\&= (1-\delta) I(X;Y)-2 (1-\alpha) \delta(1-\delta)p(1-p) I(U;V) \label{eq:W1noisybinaryub}
\end{align}
}{
\begin{align}
    C &\le \frac{1}{2} I(X^2;Y^K,A^2)\\ & = (1-\delta) I(X;Y)\notag\\&\hspace{0.5em}-2 (1-\alpha) \delta(1-\delta)p(1-p) I(U;V) \label{eq:W1noisybinaryub}
\end{align}
}
where $U$ and $V$ are binary random variables with $U\sim\text{Bernoulli}(\nicefrac{1}{2})$ and ${p_{V|U}=p_{Y|X}}$.
\end{cor}
\begin{proof}
See Appendix~\ref{proof:W1conversenoisybinaryub}.
\end{proof}
Again, for any $p\in(0,1)$ and $\alpha\in[0,1)$, the upper bound given in Corollary~\ref{cor:W1conversenoisybinaryub} is strictly lower than the one provided in Theorem~\ref{thm:mainresultWm} which is
\iftoggle{singlecolumn}{
\begin{align}
    I(X;Y,S)=(1-\delta)I(X;Y)
\end{align}
}{
\begin{align}
    I(X;Y,S)=(1-\delta)I(X;Y)
\end{align}
}

We note that the tighter upper bounds in Corollaries~\ref{cor:W1conversenoiselessub} and \ref{cor:W1conversenoisybinaryub} become generalizations of the upper bound on the noiseless deletion channel mutual information, given in~\cite[Corollary 1]{drmota6283980}. Specifically,~\cite{drmota6283980} considers noiseless deletion channel with \emph{i.i.d.} Bernoulli inputs. Corollary~\ref{cor:W1conversenoiselessub} extends the results to noiseless deletion channels with arbitrary alphabet sizes. Furthermore, Corollary~\ref{cor:W1conversenoisybinaryub} extends the results to binary noisy deletion channels with arbitrary noise. 

For the binary noisy case considered in Corollary~\ref{cor:W1conversenoisybinaryub}, the numerical comparison of the lower bound and the two upper bounds on the matching capacity is provided in Figure~\ref{fig:entryrates}. Note that the upper bound provided by Corollary~\ref{cor:W1conversenoisybinaryub} is not tight as it can be shown that a larger value of $n$ gives a tighter upper bound, implying that the gap between the lower and the upper bounds in Theorem~\ref{thm:mainresultW1} is smaller than the one shown in Figure~\ref{fig:entryrates}.

\section{Extensions}
\label{sec:discussion}
In this section, we discuss extensions to the system model and results. Specifically, in Section~\ref{subsec:adversarialrepetition}, we investigate the adversarial repetition case instead of random repetitions, where the repetitions are not due to random sampling of the time-indexed data, but due to a constrained privacy mechanism. In Section~\ref{subsec:seedlessWm}, we consider the identical repetition model with no seeds. In Section~\ref{subsec:zerorateWm}, we discuss the zero-rate regime, where the row size $m_n$ is not necessarily exponential in the column size $n$, and derive conditions necessary for the detection algorithms discussed in Section~\ref{sec:matchingcapacityWm} to work.

\subsection{What If Repetitions Are Intentional?}
\label{subsec:adversarialrepetition}
So far, as stated in Definition~\ref{defn:labeleddbidenticalrepetition}, we have assumed that the identical repetitions occur randomly according to a discrete probability distribution $p_S$ with finite integer support. In this subsection, we study the case of an adversary who controls the repetition pattern (under some constraints) to make matching as difficult as possible. This could arise for example where a privacy-preserving mechanism denies the sampling of the geolocation data when that data contains the most information about the users, such as their home addresses. We consider the adversarial setting under identical repetition assumption.

We stress that in the identical repetition setting, the replicas either have no effect on the matching capacity as in the noiseless case (Theorem~\ref{thm:noiselesscapacityWm}) or offer additional information acting as a repetition code of random length, in turn increasing the matching capacity (Theorem~\ref{thm:mainresultWm}). Hence, it is expected that any adversary who tries to hinder the matching process to not allow the replication of entries. Therefore in the adversarial repetition setting, it is natural to focus on the deletion-only case. We assume an adversary with a $\delta$-\emph{deletion budget}, which can delete up to $\delta$ fraction of the columns, to maximize the mismatch probability. For tractability, we focus on the noiseless case with \emph{i.i.d.} database entries. More formally, we assume  $X_i\overset{\text{iid}}{\sim}p_X$ where
\iftoggle{singlecolumn}{
\begin{align}
    p_{Y|X}(y|x) &=\mathbbm{1}_{[y=x]},\hspace{1em}\forall (x,y)\in\mathfrak{X}^2
\end{align}
}{
\begin{align}
    p_{Y|X}(y|x) &=\mathbbm{1}_{[y=x]},\hspace{1em}\forall (x,y)\in\mathfrak{X}^2
\end{align}
}

Under these assumptions, we define the adversarial matching capacity as follows:

\begin{defn}{\textbf{(Adversarial Matching Capacity)}}\label{defn:matchingcapacityadversarial}
The \emph{adversarial matching capacity} $C^{\text{adv}}(\delta)$ is the supremum of the set of all achievable rates corresponding to a database distribution $p_X$ and an adversary with a $\delta$-\emph{deletion budget} when there is identical repetition. More formally,
\iftoggle{singlecolumn}{
\begin{align}
    C^{\text{adv}}(\delta) &\triangleq \sup \{R: \forall I_{\text{del}}=(i_1,\dots,i_{n\delta})\subseteq [n], \Pr(\hat{\sigma}_n(J)\neq \sigma_n(J))\overset{n\to\infty}{\longrightarrow} 0, J\sim\text{Unif}([m_n]))\}
\end{align}
}{
\begin{align}
    C^{\text{adv}}(\delta) &\triangleq \sup \{R: \forall I_{\text{del}}=(i_1,\dots,i_{n\delta})\subseteq [n],\notag\\& \hspace{3em}\Pr(\hat{\sigma}_n(J)\neq \sigma_n(J))\overset{n\to\infty}{\longrightarrow} 0,\notag\\ &\hspace{3em}J\sim\text{Unif}([m_n]))\}
\end{align}
}
where the dependence of the matching scheme $\hat{\sigma}_n$ on the database growth rate $R$ and the column deletion index set $I_{\text{del}}$ is omitted for brevity.
\end{defn}

Note that in this setting, although the deletions are not random, the matching error is still a random variable due to the random natures of $\mathbf{X}$ and $\sigma_n$. 
In the proof of Theorem~\ref{thm:adversarialWm} below (Appendix~\ref{proof:adversarialWm}), we argue that in the adversarial setting, we can still convert deletions into erasures via the histogram-based repetition detection algorithm of Section~\ref{subsec:noiselessWm}. After the detection part, we use the following matching scheme: We first remove deleted columns from $\mathbf{X}$, and then perform exact sequence matching, as described in Algorithm~\ref{alg:adversarialmatching} which has $O(m_n^2 n)$ runtime, similar to Algorithms~\ref{alg:identicalrepetitionmatching}-\ref{alg:rowwisematchingscheme}.

\begin{algorithm}[t]
\caption{Exact Sequence Matching Scheme Under Adversarial Deletions}\label{alg:adversarialmatching}
\Input{$(\mathbf{X},\mathbf{Y})$}
\Output{$\hat{\sigma}_n$}
$(\tilde{\mathbf{X}},\tilde{\mathbf{Y}})\gets$ CollapseDatabases($\mathbf{X},\mathbf{Y}$)\Comment*[r]{Eq.~\eqref{eq:collapsedb}.}
$(\tilde{\mathbf{H}}^{(1)},\tilde{\mathbf{H}}^{(2)})\gets$ ColumnHistograms($\tilde{\mathbf{X}},\tilde{\mathbf{Y}}$)\Comment*[r]{Eq.~\eqref{eq:histogramdefn}.}
\blue{/* Histogram-based repetition detection */}\\
\For{$i=1$ \KwTo \textup{columnSize($\tilde{\mathbf{H}}^{(1)}$)}}{
count$\gets 0$\;
\For{$j=1$ \KwTo \textup{columnSize($\tilde{\mathbf{H}}^{(2)}$)}}{
\If{$\tilde{\mathbf{H}}^{(2)}[:][j]=\tilde{\mathbf{H}}^{(1)}[:][i]$}{
count$\gets$ count + 1\;
}
}
$\hat{S}[i]\gets $ count\;
}
\blue{/* Discard deleted columns */}\\
\For{$j=1$ \KwTo \textup{columnSize(}$\mathbf{X}$\textup{)}}{
\eIf{$\hat{S}[j]=0$}{
$\hat{\mathbf{X}}[:][j]\gets []$\;
}{
$\hat{\mathbf{X}}[:][j]\gets \mathbf{X}[:][j]$\;
}
}

\blue{/* Exact sequence matching (See Appendix~\ref{proof:adversarialWm}.) */}\\
\For{$i = 1$ \KwTo  \textup{rowSize($\mathbf{X}$)}}{
\textup{count}$\gets 0$\;
\For{$j = 1$ \KwTo  \textup{rowSize($\mathbf{Y}$)}}{
  \If{$\mathbf{Y}[j][:]=\hat{\mathbf{X}}[i][:]$}{
    $\hat{\sigma}_n[i] \gets$ j\;
    count$\gets$ count + 1\;
  }
  
}
\blue{/* count $>$ 1: Collision Error. */}\\
\If{\textup{count} $\neq 1$}{
$\hat{\sigma}_n[i] \gets 0$\Comment*[r]{Matching error.}
}
}

\end{algorithm}

We state our main result on the adversarial matching capacity in the following theorem:
\begin{thm}{\textbf{(Adversarial Matching Capacity)}}\label{thm:adversarialWm}
Consider a database distribution $p_X$ and an adversary with a $\delta$-\emph{deletion budget} when there is identical repetition. Then the adversarial matching capacity is
\iftoggle{singlecolumn}{
\begin{align}
    C^{\text{adv}}(\delta) &= \begin{cases}
    D(\delta\|1-\hat{q}),&\text{if } \delta\le 1-\hat{q}\\
    0, &\text{if } \delta> 1-\hat{q}
    \end{cases}
\end{align}
}{
\begin{align}
    C^{\text{adv}}(\delta) &= \begin{cases}
    D(\delta\|1-\hat{q}),&\text{if } \delta\le 1-\hat{q}\\
    0, &\text{if } \delta> 1-\hat{q}
    \end{cases}
\end{align}
}
where $\hat{q} \triangleq \sum_{x\in\mathfrak{X}} p_X(x)^2$.
\end{thm}
\begin{proof}
See Appendix~\ref{proof:adversarialWm}.
\end{proof}

\begin{figure}[t]
\centerline{\includegraphics[width=0.6\textwidth]{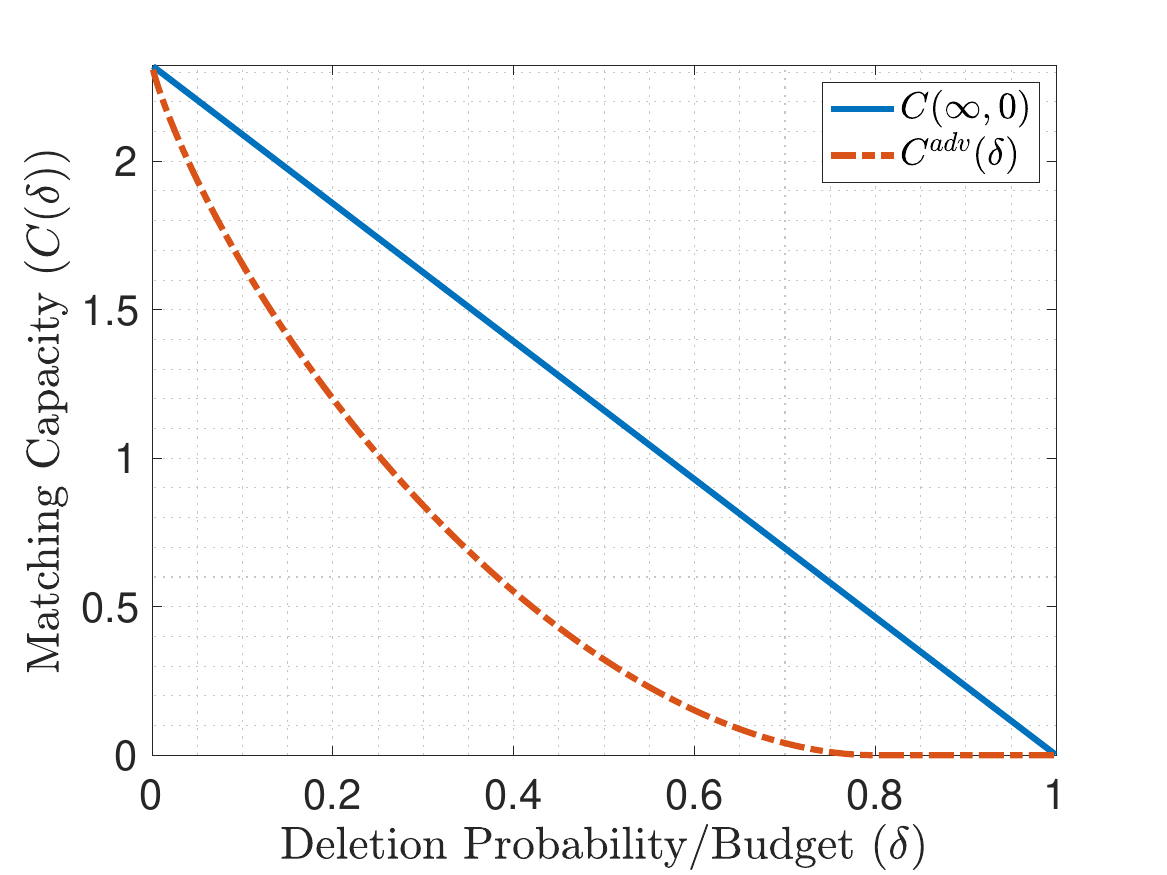}}
\caption{Matching capacities $C$ vs. deletion probability/budget ($\delta$) when $X\sim \text{Unif}(\mathfrak{X})$, $\mathfrak{X}=[5]$. Notice that in this case $\hat{q}=0.2$ and for $\delta>1-\hat{q}=0.8$ the adversarial matching capacity $C^{\text{adv}}(\delta)$ is zero, while the matching capacity with random deletions $C$ is positive.}
\label{fig:adversarialcapacity}
\end{figure}

\begin{sloppypar}
The matching capacities for random and adversarial deletions as a function of the deletion probability/budget are illustrated in Figure~\ref{fig:adversarialcapacity}. Note that for $\delta>1-\hat{q}$, we have $C^{\text{adv}}(\delta)=0$ whereas ${C=(1-\delta) H(X)>0}$. Furthermore, when $\delta\le 1-\hat{q}$ the matching capacity is significantly reduced when the column deletions are intentional rather than random.
\end{sloppypar}

\subsection{What If There Were No Seeds?}
\label{subsec:seedlessWm}
In Section~\ref{sec:matchingcapacityWm}, we assumed the availability of seeds with a seed size $\Lambda_n=\Omega(\log\log m_n)$. Now, we focus on the identical repetition scenario with no seeds. 

Note that the replica detection algorithm of Section~\ref{subsec:replicadetection} does not require any seeds. Therefore in the seedless scenario, we can still detect the replicas with a vanishing probability of error. On the other hand, in the general noisy setting, the deletion detection algorithm of Section~\ref{subsec:seededdeletiondetection} necessitates seeds. Therefore, in the case of no seeds, we cannot perform deletion detection and we need to modify the matching scheme of Section~\ref{subsec:matchingschemeWm} to obtain lower bounds on the matching capacity $C^{\text{seedless}}$. 

For tractability, we focus on the case with \emph{i.i.d.} database entries, \emph{i.e.,} $\gamma=0$. More formally, we assume  $X_i\overset{\text{iid}}{\sim}p_X$.
Under this assumption, we state a lower bound on the unseeded matching capacity with identical repetition in the following theorem via Algorithm~\ref{alg:seedlessmatching} which has $O(m_n^2 n)$ runtime, similar to Algorithms~\ref{alg:identicalrepetitionmatching}-\ref{alg:adversarialmatching}.

\begin{algorithm}[t]
\caption{Seedless Matching Scheme (Identical Repetition)}\label{alg:seedlessmatching}
\Input{$(\mathbf{X},\mathbf{Y},p_X,p_{Y|X},p_S,\epsilon)$}
\Output{$\hat{\sigma}_n$}
\textup{isReplica}$\gets$ Algorithm~\ref{alg:noisyreplicadetection}($\mathbf{Y},p_X,p_{Y|X}$)\;

$\tilde{\mathbf{Y}}\gets$ MarkerAddition($\mathbf{Y}$,isReplica)\;

\For{$i = 1$ \KwTo  \textup{rowSize($\mathbf{X}$)}}{
\textup{count}$\gets 0$\;
\For{$j = 1$ \KwTo  \textup{rowSize($\mathbf{Y}$)}}{
\blue{/* Typical subsequence check (See Appendix~\ref{proof:mainresultseedless}.) */}\\
  \If{\textup{TypicalSubsequenceCheck(}$\mathbf{X}[i][:],\tilde{\mathbf{Y}}[j][:],p_X,p_{Y|X},p_S,\epsilon$\textup{)}}{
    $\hat{\sigma}_n[i] \gets$ j\;
    count$\gets$ count + 1\;
  }
  
}
\blue{/* count = 0: $\tilde{\mathbf{Y}}[j][:]$ is not jointly typical with any subsequence of any row of $\mathbf{X}$. */}\\
\blue{/* count $>$ 1: $\tilde{\mathbf{Y}}[j][:]$ is jointly typical with a subsequence of multiple rows of $\mathbf{X}$. */}\\
\If{\textup{count} $\neq 1$}{
$\hat{\sigma}_n[i] \gets 0$\Comment*[r]{Matching error.}
}
}

\end{algorithm}

\begin{thm}{\textbf{(Seedless Matching Capacity with Identical Repetition})}\label{thm:mainresultseedless}
Consider a database distribution $p_X$, a noise distribution $p_{Y|X}$, a repetition distribution $p_S$ and an identical repetition pattern. Then, in the seedless case, the matching capacity $C^{\text{seedless}}$ satisfies
\iftoggle{singlecolumn}{
\begin{align}
    C^{\text{seedless}}&\ge \left[I(X;Y^S,S)-H_b(\delta)\right]^+\label{eq:seedlessachievable}\\
    C^{\text{seedless}} &\le I(X;Y^S,S)\label{eq:seedlessconverse}
\end{align}
}{
\begin{align}
    C^{\text{seedless}}&\ge \left[I(X;Y^S,S)-H_b(\delta)\right]^+\label{eq:seedlessachievable}\\
    C^{\text{seedless}} &\le I(X;Y^S,S)\label{eq:seedlessconverse}
\end{align}
}
where $\delta\triangleq p_S(0)$ is the deletion probability, $S\sim p_S$ and ${Y^S}$ has the following distribution conditioned on $X$ such that
\iftoggle{singlecolumn}{
\begin{align}
    \Pr(Y^{S}=y^{S}|X=x)&=\begin{cases}
    \prod\limits_{j=1}^{S} p_{Y|X}(y_j|x)&\text{if }S>0\\
    \mathbbm{1}_{[y^{s} = E]} &\text{if }S=0
    \end{cases}
\end{align}
}{
\begin{align}
    \Pr(Y^{S}=y^{S}|X=x)&=\begin{cases}
    \prod\limits_{j=1}^{S} p_{Y|X}(y_j|x)&\text{if }S>0\\
    \mathbbm{1}_{[y^{s} = E]} &\text{if }S=0
    \end{cases}
\end{align}
}
where $E$ denotes the empty string.

Furthermore, for repetition distributions with $\delta\le 1-\nicefrac{1}{|\mathfrak{X}|}$, the lower bound can be tightened as
\iftoggle{singlecolumn}{
\begin{align}
    C^{\text{seedless}}&\ge \left[I(X;Y^S,S)+\delta[H(X)-\log(|\mathfrak{X}|-1)]^+-H_b(\delta)\right]^+
\end{align}
}{
\begin{align}
    C^{\text{seedless}}&\ge [I(X;Y^S,S)-H_b(\delta)\notag\\&\hspace{0.5em}+\delta[H(X)-\log(|\mathfrak{X}|-1)]^+]^+
\end{align}
}
\end{thm}
\begin{proof}
See Appendix~\ref{proof:mainresultseedless}.
\end{proof}

We note that although the converse results of Theorems~\ref{thm:mainresultWm} and \ref{thm:mainresultseedless} match, the achievable rates differ by $H_b(\delta)$. In other words, Theorem~\ref{thm:mainresultseedless} implies that the gap between the lower and the upper bounds on the seedless matching capacity is at most $H_b(\delta)$. We note that this gap is due to our inability to detect deletions in the achievability part. Hence, we conjecture that the lower bound in Theorem~\ref{thm:mainresultseedless} is loose while the converse is tight. This is because in the noiseless setting, as discussed in Section~\ref{subsec:noiselessWm}, deletion detection can be performed without seeds and the achievability bound is indeed improved and tight.

\subsection{Zero-Rate Regime}
\label{subsec:zerorateWm}
In Section~\ref{sec:matchingcapacityWm}, we considered at the matching capacity $C$ for $\Lambda_n=\Omega(\log\log m_n)$ when the database growth rate $R$ is positive. In other words, so far, we have assumed
\iftoggle{singlecolumn}{
\begin{align}
    \lim\limits_{n\to\infty} \frac{1}{n}\log m_n &>0
\end{align}
}{
\begin{align}
    \lim\limits_{n\to\infty} \frac{1}{n}\log m_n &>0
\end{align}
}
The detection algorithms we presented in Sections~\ref{subsec:replicadetection} through \ref{subsec:noiselessWm} depended on the row size $m_n$ being large compared to the column size $n$. In this section, we further investigate these algorithms to derive the sufficient and/or necessary conditions on the relation between $m_n$ and $n$ in order for them to work in the zero-rate regime where
\iftoggle{singlecolumn}{
\begin{align}
    \lim\limits_{n\to\infty} \frac{1}{n}\log m_n &=0.
\end{align}
}{
\begin{align}
    \lim\limits_{n\to\infty} \frac{1}{n}\log m_n &=0.
\end{align}
}
Since $R=0$, we define the non-asymptotic database growth rate $R_n$ as 
\iftoggle{singlecolumn}{
\begin{align}
    R_n &\triangleq \frac{1}{n}\log m_n.
\end{align}
}{
\begin{align}
    R_n &\triangleq \frac{1}{n}\log m_n.
\end{align}
}
Here, $R=0$ trivially implies $R_n\to0$ as $n\to\infty$. Below we investigate the sufficient conditions on $R_n$ such that the results of Sections~\ref{sec:matchingcapacityWm} and \ref{sec:matchingcapacityW1} hold.

\subsubsection{Noisy Replica Detection}
We consider the replica detection algorithm discussed in Section~\ref{subsec:replicadetection}. Note that the RHS of equation~\eqref{eq:replicadetectionlast} of Appendix~\ref{proof:noisyreplicadetection} has $2K-2\le 2 n s_{\max} = O(n)$ additive terms, each decaying exponentially in $m_n$. Thus, for a given average Hamming distance threshold $\tau\in(p_1,p_0)$ which is chosen based based on $\mathbf{P}$ and $p_{Y|X}$ and in turn constant with respect to $n$
\iftoggle{singlecolumn}{
\begin{align}
    m_n &\ge \frac{\log(n s_{\max})}{\min\{D(\tau\|p_0),D(1-\tau\|1-p_1)\}}=\Theta(\log n)
\end{align}
}{
\begin{align}
    m_n &\ge \frac{\log(n s_{\max})}{\min\{D(\tau\|p_0),D(1-\tau\|1-p_1)\}}\notag\\&=\Theta(\log n)
\end{align}
}
is enough to ensure a vanishing replica detection error probability. In other words, as long as $m_n =\Omega(\log n)$ and in turn
\iftoggle{singlecolumn}{
\begin{align}
    R_n &= \Omega\left(\frac{\log \log n}{n}\right)
\end{align}
}{
\begin{align}
    R_n &= \Omega\left(\frac{\log \log n}{n}\right)
\end{align}
}
our replica detection algorithm works.

\subsubsection{Seeded Deletion Detection}
We study the seeded deletion detection algorithm discussed in Section~\ref{subsec:seededdeletiondetection}. Note that we only run the deletion detection algorithm on the seeds $(\mathbf{G}^{(1)},\mathbf{G}^{(2)})$ and not on the database pair $(\mathbf{X},\mathbf{Y})$ directly, the relationship between $m_n$ and $n$ does not affect the success of the deletion detection. Thus, as long as the seed size $\Lambda_n=\Omega(\log n)$ our deletion detection algorithm works for any database growth rate, including the zero-rate regime. This in turn implies that $m_n\ge \Lambda_n=\Omega(\log n)$ and 
\iftoggle{singlecolumn}{
\begin{align}
    R_n &= \Omega\left(\frac{\log \log n}{n}\right).
\end{align}
}{
\begin{align}
    R_n &= \Omega\left(\frac{\log \log n}{n}\right).
\end{align}
}

\subsubsection{Noiseless Joint Deletion-Replication Detection}
We investigate the histogram-based joint deletion-replication detection algorithm introduced in Section~\ref{subsec:noiselessWm} for the noiseless scenario. By Lemma~\ref{lem:histogram}, $m_n=\omega(n^4)$ is sufficient. Thus, as long as $\log m_n \ge 4 \log n$, the histogram-based detection can be performed with a performance guarantee. In turn, for any
\iftoggle{singlecolumn}{
\begin{align}
    R_n&= \Omega\left(\frac{\log n}{n}\right)
\end{align}
}{
\begin{align}
    R_n&= \Omega\left(\frac{\log n}{n}\right)
\end{align}
}
the histogram-based detection algorithm has a vanishing probability of error.

Therefore, in the noiseless setting, database growth rate $\smash{R_n=\Omega\left(\nicefrac{\log n}{n}\right)}$ provides enough granularity on the column histograms and we can perform detection with a decaying probability of error which then leads to asymptotically-zero mismatch probability.

Note that, for tractability, so far we have collapsed the databases into binary-valued ones. Further, in Lemma~\ref{lem:histogram}, we showed that for the collapsed databases $m_n=\omega(n^4)$ is enough for the asymptotic uniqueness of the column histograms.
We now tighten this order relation for the special case where $\gamma=0$ results in an \emph{i.i.d.} database distribution $X_{i,j}\overset{\text{i.i.d.}}{\sim}p_X$ with support $\mathfrak{X}$. 

\begin{lem}{\textbf{(Asymptotic Uniqueness of the Uncollapsed Histograms)}}\label{lem:histogramuncollapsed}
Consider an \emph{i.i.d.} database distribution $p_X$. Let ${H}^{(1)}_j$ denote the histogram of the $j$\textsuperscript{th} column of $\mathbf{X}$.
Then,
\iftoggle{singlecolumn}{
\begin{align}
    \Pr\left(\exists i,j\in [n],\: i\neq j,H^{(1)}_i=H^{(1)}_j\right)\to 0 \text{ as }n\to \infty
\end{align}
}{
\begin{align}
    \Pr\left(\exists i,j\in [n],\: i\neq j,H^{(1)}_i=H^{(1)}_j\right)\to 0 \text{ as }n\to \infty
\end{align}
}
if $m_n=\omega(n^\frac{4}{|\mathfrak{X}|-1})$.
\end{lem}
\begin{proof}
See Appendix~\ref{proof:histogramuncollapsed}.
\end{proof}
Note that in the binary setting the results of Lemmas~\ref{lem:histogram} and \ref{lem:histogramuncollapsed} agree. 

Lemma~\ref{lem:histogramuncollapsed} implies that we only need a row size $m_n$ polynomial in $n$ to guarantee enough granularity for the uniqueness of $H^{(1)}_i$ and that the degree of the polynomial scales inversely with the alphabet size $|\mathfrak{X}|$. Furthermore, to demonstrate the tightness of this requirement of having ${m_n=\omega(n^\frac{4}{|\mathfrak{X}|-1}})$, we consider the special case where $p_X$ is uniform over $\mathfrak{X}$. This leads to the following proposition:
\begin{prop}\label{prop:uniformhistogram}
Let ${H}^{(1)}_j$ denote the histogram of the $j$\textsuperscript{th} column of $\mathbf{X}$. If ${p_X(x)=\frac{1}{|\mathfrak{X}|},\:\forall x\in\mathfrak{X}}$, then 
\iftoggle{singlecolumn}{
\begin{align}
    \Pr\left(\exists i,j\in [n],\: i\neq j,H^{(1)}_i=H^{(1)}_j\right)&=n^2 m_n^{\frac{1-|\mathfrak{X}|}{2}} C_{|\mathfrak{X}|} (1+o_n(1))
\end{align}
}{
\begin{align}
    \Pr\Big(\exists i,j\in [n],\: &i\neq j,H^{(1)}_i=H^{(1)}_j\Big)\notag\\&=n^2 m_n^{\frac{1-|\mathfrak{X}|}{2}} C_{|\mathfrak{X}|} (1+o_n(1))
\end{align}
}
where ${C_{|\mathfrak{X}|}=\left(4\pi\right)^{\frac{1-|\mathfrak{X}|}{2}} |\mathfrak{X}|^{\frac{|\mathfrak{X}|}{2}}}$.
\end{prop}
\begin{proof}
See Appendix~\ref{proof:uniformhistogram}.
\end{proof}
Proposition~\ref{prop:uniformhistogram} states that in the setting with \emph{i.i.d.} uniform database distribution, for the asymptotic uniqueness of the column histograms $m_n=\omega(n^{\frac{4}{|\mathfrak{X}|-1}})$ is not only sufficient but also necessary.

\subsubsection{Independent Repetition Row  Matching Scheme}
In the independent repetition scenario, we have no detection algorithms which depend on the large-$m_n$ assumption. Therefore, so long as the RHS of \eqref{eq:rowwiseachievable} is positive, any $R_n=o_n(1)$ is achievable. We stress that this observation trivially applies to the identical repetition case as well since one can simply ignore any underlying structure and perform the matching scheme given in Section~\ref{subsec:achievabilityW1}.

\section{Conclusion}
\label{sec:conclusion}

In this work, we have presented a unified information-theoretic foundation for database matching under noise and synchronization errors. We have showed that when the repetition pattern is constant across rows, the running Hamming distances between the consecutive columns of the correlated repeated database can be used to detect replicas. In addition, given seeds whose size grows double-logarithmic with the number of rows, a Hamming distance-based threshold testing, after an adequate remapping of database entries, can be used to infer the locations of the deletions. Using the proposed detection algorithms, and a joint typicality-based rowwise matching scheme, we have derived an achievable database growth rate, which we prove is tight. Therefore, we have completely characterized the database matching capacity under noisy column repetitions. Furthermore, we have derived achievable database growth rates proposing a typicality-based matching scheme and a converse result for the setting where the repetition takes place entrywise, where we build analogy between database matching and synchronization channel decoding. We have also discussed some extensions, such as the adversarial column deletion setting rather then the random one. 

Other natural extensions beyond those studied in this paper include the finite column size regime, where tools from finite-blocklength information theory could be useful, and practical algorithms with theoretical guarantees. An extensive analysis of the parallels between database matching under synchronization errors and two-dimensional synchronization channels~\cite{yaakobitit,yaakobiisit} and the construction of codes tailored to correct the error patterns investigated in this paper could be an interesting line of future work. Finally, one can extend our adversarial setting into a noisy one where the privacy-preserving mechanism not only deletes columns but also introduces intentional noise on the microdata, and investigate the adversarial matching capacity through a worst-case analysis.

\bibliographystyle{IEEEtran}
\bibliography{references}
\begin{appendices}
    \section{Proof of Lemma~\ref{lem:noisyreplicadetection}}
\label{proof:noisyreplicadetection}
Observe that since the rows of $\mathbf{Y}$ are \emph{i.i.d.} conditioned on the column repetition pattern $S^n$, the Hamming distance $H_j$ between consecutive columns $C^{m_n}_j$ and $C^{m_n}_{j+1}$ follows a Binomial distribution whose success parameter depends on whether $C^{m_n}_j$ and $C^{m_n}_{j+1}$ are noisy replicas or not. More formally, if $C_j^{m_n}$ and $C_{j+1}^{m_n}$ being noisy replicas, then there exist an $i \in [m_n]$ such that
    \begin{align}
        \Pr(Y_{t,j},Y_{t,j+1}=y_1,y_2|X_{\sigma_n^{-1}(t),i}=x) &= p_{Y|X}(y_1|x) p_{Y|X}(y_2|x),\hspace{1em} \forall t\in[m_n]
    \end{align}
\begin{sloppypar}
More specifically, let $C^{m_n}_j$ and $C^{m_n}_{j+1}$ correspond to the $j_1$\textsuperscript{st} and $j_2$\textsuperscript{nd} columns of $\mathbf{X}$ and let $r\triangleq j_2-j_1-1$ denote the number of deleted columns between $C^{m_n}_j$ and $C^{m_n}_{j+1}$. Note that $r=-1$ denotes the case when $C^{m_n}_j$ and $C^{m_n}_{j+1}$ are noisy replicas. Then we have ${H_j\sim \text{Binom}(m_n,p_1)}$ if $C^{m_n}_j$ and $C^{m_n}_{j+1}$ are noisy replicas and ${H_j\sim \text{Binom}\smash{(m_n,p_0^{(r)})}}$ otherwise.
Thus, proving that $\smash{p_0^{(r)}}$ and $p_1$ are bounded away from one another for any $r\ge0$ will allow us to use the running Hamming distance based threshold test discussed in Section~\ref{subsec:replicadetection}.
\end{sloppypar}

\begingroup
\allowdisplaybreaks
Our goal is to prove that $p_0^{(r)}>p_1$ for any $r\ge0$. First, we can formally rewrite $p_0$ as
\iftoggle{singlecolumn}{
\begin{align}
    p_0^{(r)}&= \sum\limits_{x_1\in\mathfrak{X}}\sum\limits_{x_2\in\mathfrak{X}}\sum\limits_{y\in\mathfrak{X}} \Pr(X_{1,j_1}=x_1)\Pr(X_{1,j_2}=x_2|X_{1,j_1}=x_1)\notag\\
    &\hspace{10em}\Pr(Y_{\sigma_n(1),j}=y|X_{j_1}=x_1)\Pr(Y_{\sigma_n(1),j+1}\neq y|X_{1,j_2}=x_2)\\
    &=\sum\limits_{x_1\in\mathfrak{X}}\sum\limits_{x_2\in\mathfrak{X}}\sum\limits_{y\in\mathfrak{X}} u_{x_1}\Pr(X_{1,j_2}=x_2|X_{1,j_1}=x_1) p_{Y|X}(y|x_1)[1-p_{Y|X}(y|x_2)]\\
    &= \sum\limits_{i=1}^{|\mathfrak{X}|} \sum\limits_{j=1}^{|\mathfrak{X}|} \sum\limits_{k=1}^{|\mathfrak{X}|}
    u_i\: (\mathbf{P}^r)_{i,j}\: p_{Y|X}(k|i)\: [1-p_{Y|X}(k|j)]\\
    &=\sum\limits_{i=1}^{|\mathfrak{X}|} \sum\limits_{j=1}^{|\mathfrak{X}|} \sum\limits_{k=1}^{|\mathfrak{X}|}
    u_i\: [(1-\gamma^r)u_j+\gamma^r \delta_{i j}]\: p_{Y|X}(k|i) \:[1-p_{Y|X}(k|j)]\\
    &= \sum\limits_{i=1}^{|\mathfrak{X}|}  \sum\limits_{k=1}^{|\mathfrak{X}|}
    u_i\: [(1-\gamma^r)u_i+\gamma^r]\: p_{Y|X}(k|i) \:[1-p_{Y|X}(k|i)]\notag\\
    &\hspace{5em}+ \sum\limits_{i=1}^{|\mathfrak{X}|} \sum\limits_{j\neq i} \sum\limits_{k=1}^{|\mathfrak{X}|}
    u_i\: [(1-\gamma)u_j]\: p_{Y|X}(k|i) \:[1-p_{Y|X}(k|j)]\\
    &= (1-\gamma^r) \sum\limits_{i=1}^{|\mathfrak{X}|} \sum\limits_{j=1}^{|\mathfrak{X}|} \sum\limits_{k=1}^{|\mathfrak{X}|} u_i\: u_j\: p_{Y|X}(k|i) \:[1-p_{Y|X}(k|j)] \notag\\
    &\hspace{6em}+ \gamma^r \sum\limits_{i=1}^{|\mathfrak{X}|}  \sum\limits_{k=1}^{|\mathfrak{X}|}
    u_i\: p_{Y|X}(k|i) \:[1-p_{Y|X}(k|i)]\\
    &= (1-\gamma^r) p_0^{\prime} + \gamma^r p_1^{\prime}
\end{align}}
{
\begin{align}
    p_0^{(r)}&= \sum\limits_{x_1\in\mathfrak{X}}\sum\limits_{x_2\in\mathfrak{X}}\sum\limits_{y\in\mathfrak{X}} \Pr(X_{1,j_1}=x_1)\notag\\
    &\hspace{6em}\Pr(X_{1,j_2}=x_2|X_{1,j_1}=x_1)\notag\\
    &\hspace{6em}\Pr(Y_{\sigma_n(1),j}=y|X_{j_1}=x_1)\notag\\
    &\hspace{6em}\Pr(Y_{\sigma_n(1),j+1}\neq y|X_{1,j_2}=x_2)\\
    &=\sum\limits_{x_1\in\mathfrak{X}}\sum\limits_{x_2\in\mathfrak{X}}\sum\limits_{y\in\mathfrak{X}} \Pr(X_{1,j_1}=x_1)\notag\\
    &\hspace{6em}\Pr(X_{1,j_2}=x_2|X_{1,j_1}=x_1) \notag\\
    &\hspace{6em}p_{Y|X}(y|x_1)[1-p_{Y|X}(y|x_2)]\\
    &= \sum\limits_{i=1}^{|\mathfrak{X}|} \sum\limits_{j=1}^{|\mathfrak{X}|} \sum\limits_{k=1}^{|\mathfrak{X}|}
    u_i\: (\mathbf{P}^r)_{i,j}\: p_{Y|X}(k|i)\: [1-p_{Y|X}(k|j)]\\
    &=\sum\limits_{i=1}^{|\mathfrak{X}|} \sum\limits_{j=1}^{|\mathfrak{X}|} \sum\limits_{k=1}^{|\mathfrak{X}|}
    u_i\: [(1-\gamma^r)u_j+\gamma^r \delta_{i j}]\notag\\
    &\hspace{6em}p_{Y|X}(k|i) \:[1-p_{Y|X}(k|j)]\\
    &= \sum\limits_{i=1}^{|\mathfrak{X}|}  \sum\limits_{k=1}^{|\mathfrak{X}|}
    u_i\: [(1-\gamma^r)u_i+\gamma^r]\notag\\
    &\hspace{6em} p_{Y|X}(k|i) \:[1-p_{Y|X}(k|i)]\notag\\
    &\hspace{1em}+ \sum\limits_{i=1}^{|\mathfrak{X}|} \sum\limits_{j\neq i} \sum\limits_{k=1}^{|\mathfrak{X}|}
    u_i\: [(1-\gamma)u_j]\notag\\
    &\hspace{6em}p_{Y|X}(k|i) \:[1-p_{Y|X}(k|j)]\\
    &= (1-\gamma^r) \sum\limits_{i=1}^{|\mathfrak{X}|} \sum\limits_{j=1}^{|\mathfrak{X}|} \sum\limits_{k=1}^{|\mathfrak{X}|} u_i\: u_j\notag\\
    &\hspace{6em} p_{Y|X}(k|i) \:[1-p_{Y|X}(k|j)]\notag\\
    &\hspace{1em}+ \gamma^r \sum\limits_{i=1}^{|\mathfrak{X}|}  \sum\limits_{k=1}^{|\mathfrak{X}|}
    u_i\: p_{Y|X}(k|i) \:[1-p_{Y|X}(k|i)]\\
    &= (1-\gamma^r) p_0^{\prime} + \gamma^r p_1^{\prime}
\end{align}
}
\endgroup
where
\iftoggle{singlecolumn}{
\begin{align}
    p_0^{\prime}&\triangleq \sum\limits_{i=1}^{|\mathfrak{X}|} \sum\limits_{j=1}^{|\mathfrak{X}|} \sum\limits_{k=1}^{|\mathfrak{X}|} u_i\: u_j\: p_{Y|X}(k|i) \:[1-p_{Y|X}(k|j)]\\
    p_1^{\prime} &\triangleq \sum\limits_{i=1}^{|\mathfrak{X}|}  \sum\limits_{k=1}^{|\mathfrak{X}|}
    u_i\: p_{Y|X}(k|i) \:[1-p_{Y|X}(k|i)]
\end{align}
}{
\begin{align}
    p_0^{\prime}&\triangleq \sum\limits_{i=1}^{|\mathfrak{X}|} \sum\limits_{j=1}^{|\mathfrak{X}|} \sum\limits_{k=1}^{|\mathfrak{X}|} u_i\: u_j\: p_{Y|X}(k|i) \:[1-p_{Y|X}(k|j)]\\
    p_1^{\prime} &\triangleq \sum\limits_{i=1}^{|\mathfrak{X}|}  \sum\limits_{k=1}^{|\mathfrak{X}|}
    u_i\: p_{Y|X}(k|i) \:[1-p_{Y|X}(k|i)]
\end{align}
}
Similarly, we rewrite $p_1$ as
\iftoggle{singlecolumn}{
\begin{align}
    p_1 &= \sum\limits_{x\in\mathfrak{X}}\sum\limits_{y\in\mathfrak{X}} \Pr(X_{1,j_1}=x) \Pr(Y_{\sigma_n(1),j}=y|X_{1,j_1}=x) \Pr(Y_{\sigma_n(1),j+1}\neq y|X_{1,j_1}=x) \\
    &=\sum\limits_{i=1}^{|\mathfrak{X}|}  \sum\limits_{k=1}^{|\mathfrak{X}|}
    u_i\: p_{Y|X}(k|i) \:[1-p_{Y|X}(k|i)]\\
    &= p_1^{\prime}
\end{align}
}{
\begin{align}
    p_1 &= \sum\limits_{x\in\mathfrak{X}}\sum\limits_{y\in\mathfrak{X}} \Pr(X=x)\notag\\&\hspace{2em} \Pr(Y_1=y|X=x) \Pr(Y_2\neq y|X=x) \\
    &=\sum\limits_{i=1}^{|\mathfrak{X}|}  \sum\limits_{k=1}^{|\mathfrak{X}|}
    u_i\: p_{Y|X}(k|i) \:[1-p_{Y|X}(k|i)]\\
    &= p_1^{\prime}
\end{align}
}
Thus, for any $\gamma\in[0,1)$ and $r\ge 0$, we have
\iftoggle{singlecolumn}{
\begin{align}
    p_0^{(r)}>p_1 \iff p_0^{\prime}>p_1^{\prime} \label{eq:nrdequivalencetoiid}
\end{align}
}{
\begin{align}
    p_0>p_1 \iff p_0^{\prime}>p_1^{\prime}\label{eq:nrdequivalencetoiid}
\end{align}
}

Note that $p_0^{\prime}$ and $p_1^{\prime}$ would correspond to
\iftoggle{singlecolumn}{
\begin{align}
    p_0^{\prime}&= \Pr(Y_{\sigma_n(1),j}\neq Y_{\sigma_n(1),j+1}|r\ge 0)\\
    p_1^{\prime}&= \Pr(Y_{\sigma_n(1),j}\neq Y_{\sigma_n(1),j+1}|r=-1)
\end{align}
}{
\begin{align}
    p_0^{\prime}&= \Pr(Y_{\sigma_n(1),j}\neq Y_{\sigma_n(1),j+1}|r\ge 0)\\
    p_1^{\prime}&= \Pr(Y_{\sigma_n(1),j}\neq Y_{\sigma_n(1),j+1}|r=-1)
\end{align}
}
if the entries $X_{i,j}$ of $\mathbf{X}$ were drawn \emph{i.i.d.} from the stationary distribution $\pi$ of $\mathbf{P}$, instead of a Markov process. Thus, to consider the \emph{i.i.d.} database entries case, we introduce the discrete random variable $W$ with 
\iftoggle{singlecolumn}{
\begin{align}
    p_W(i) &= u_i,\:\forall i\in\mathfrak{X}\label{eq:iidW1}\\ 
    p_{Y|W}(y|w)&=p_{Y|X}(y|w),\:\forall (w,y)\in\mathfrak{X}^2 \label{eq:iidW2}
\end{align}
}{
\begin{align}
    p_W(i) &= u_i,\:\forall i\in\mathfrak{X}\label{eq:iidW1}\\ 
    p_{Y|W}(y|w)&=p_{Y|X}(y|w),\:\forall (w,y)\in\mathfrak{X}^2 \label{eq:iidW2}
\end{align}
}
We note that this equivalence induced by~\eqref{eq:nrdequivalencetoiid} is due to the specific Markov structure we adopted in Definition~\ref{defn:markovdb}.

Let $p_Y(y)\triangleq \sum\limits_{w\in \mathfrak{X}} p_{W,Y}(w,y)$ $\forall y\in\mathfrak{X}$. Then, we can rewrite $p_0^{\prime}$ and $p_1^{\prime}$ as
\iftoggle{singlecolumn}{
\begin{align}
    p_0^{\prime} &=\sum\limits_{w_1\in\mathfrak{X}}\sum\limits_{w_2\in\mathfrak{X}}\sum\limits_{y\in\mathfrak{X}} p_W(w_1) p_W(w_2) p_{Y|W}(y|w_1) \left[1-p_{Y|W}(y|w_2)\right]\\
    &=\sum\limits_{w_1\in\mathfrak{X}}\sum\limits_{y\in\mathfrak{X}} p_W(w_1)  p_{Y|W}(y|w_1)\sum\limits_{w_2\in\mathfrak{X}}p_W(w_2) \left[1-p_{Y|W}(y|w_2)\right]\\
    &=\sum\limits_{w\in\mathfrak{X}}\sum\limits_{y\in\mathfrak{X}} p_W(w)  p_{Y|W}(y|w)\left[1-p_Y(y)\right]\\
    p_1^{\prime} &=\sum\limits_{w\in\mathfrak{X}}\sum\limits_{y\in\mathfrak{X}} p_W(w) p_{Y|W}(y|w) \left[1-p_{Y|W}(y|w)\right]
\end{align}
}{
\begin{align}
    p_0^{\prime} &=\sum\limits_{w_1\in\mathfrak{X}}\sum\limits_{w_2\in\mathfrak{X}}\sum\limits_{y\in\mathfrak{X}} p_W(w_1) p_W(w_2)\notag\\&\hspace{5em} p_{Y|W}(y|w_1) \left[1-p_{Y|W}(y|w_2)\right]\\
    &=\sum\limits_{w_1\in\mathfrak{X}}\sum\limits_{y\in\mathfrak{X}} p_W(w_1)  p_{Y|W}(y|w_1)\notag\\&\hspace{4em}\sum\limits_{w_2\in\mathfrak{X}}p_W(w_2) \left[1-p_{Y|W}(y|w_2)\right]\\
    &=\sum\limits_{w\in\mathfrak{X}}\sum\limits_{y\in\mathfrak{X}} p_W(w)  p_{Y|W}(y|w)\left[1-p_Y(y)\right]\\
    p_1^{\prime} &=\sum\limits_{w\in\mathfrak{X}}\sum\limits_{y\in\mathfrak{X}} p_W(w) p_{Y|W}(y|w) \left[1-p_{Y|W}(y|w)\right]
\end{align}
}
Thus, we have 
\iftoggle{singlecolumn}{
\begin{align}
    p_0^{\prime} - p_1^{\prime} &=\sum\limits_{w\in\mathfrak{X}}\sum\limits_{y\in\mathfrak{X}} p_{W,Y}(w,y)\left[p_{Y|W}(y|w)-p_Y(y)\right].
\end{align}
}{
\begin{align}
    p_0^{\prime} - p_1^{\prime} &=\sum\limits_{w\in\mathfrak{X}}\sum\limits_{y\in\mathfrak{X}} p_{W,Y}(w,y)\left[p_{Y|W}(y|w)-p_Y(y)\right].
\end{align}
}
For every $y\in\mathfrak{X}$, let
\iftoggle{singlecolumn}{
\begin{align}
    \psi(y) &\triangleq \sum\limits_{w\in\mathfrak{X}}p_W(w)\left[p_{Y|W}(y|w)-p_Y(y)\right]^2\\
    &=\sum\limits_{w\in\mathfrak{X}}p_W(w)\left[p_{Y|W}(y|w)-\sum\limits_{z\in\mathfrak{X}}p_{Y|W}(y|z)p_W(z)\right]^2\\
    &\ge 0 \label{eq:psinonnegative}
\end{align}
}{
\begin{align}
    \psi(y) &\triangleq \sum\limits_{w\in\mathfrak{X}}p_W(w)\left[p_{Y|W}(y|w)-p_Y(y)\right]^2\\
    &=\sum\limits_{w\in\mathfrak{X}}p_W(w)\notag\\&\hspace{1.1em}\left[p_{Y|W}(y|w)-\sum\limits_{z\in\mathfrak{X}}p_{Y|W}(y|z)p_W(z)\right]^2\\
    &\ge 0 \label{eq:psinonnegative}
\end{align}
}
where \eqref{eq:psinonnegative} follows from the non-negativity of the square term in the summation. It must be noted that $\psi(y)=0$ only if $p_{Y|W}(y|w)=p_Y(y),\: \forall w\in\mathfrak{X}$ with $p_W(w)=u_w>0$.

Expanding the square term, we obtain
\iftoggle{singlecolumn}{
\begin{align}
    \psi(y)&= \sum\limits_{w\in\mathfrak{X}}p_W(w) p_{Y|W}(y|w)^2-2 p_Y(y) \sum\limits_{w\in\mathfrak{X}}p_W(w) p_{Y|W}(y|w)+\sum\limits_{w\in\mathfrak{X}}p_W(w) p_{Y}(y)^2\\
    &= \sum\limits_{w\in\mathfrak{X}}p_W(w) p_{Y|W}(y|w)^2 - 2 p_Y(y)^2 +p_Y(y)^2\\
    &= \sum\limits_{w\in\mathfrak{X}}p_W(w) p_{Y|W}(y|w)^2 - p_Y(y)^2
\end{align}
}{
\begin{align}
    \psi(y)&= \sum\limits_{w\in\mathfrak{X}}p_W(w) p_{Y|W}(y|w)^2\notag\\&\hspace{3em}-2 p_Y(y) \sum\limits_{w\in\mathfrak{X}}p_W(w) p_{Y|W}(y|w)\notag\\&\hspace{3em}+\sum\limits_{w\in\mathfrak{X}}p_W(w) p_{Y}(y)^2\\
    &= \sum\limits_{w\in\mathfrak{X}}p_W(w) p_{Y|W}(y|w)^2 - 2 p_Y(y)^2 +p_Y(y)^2\\
    &= \sum\limits_{w\in\mathfrak{X}}p_W(w) p_{Y|W}(y|w)^2 - p_Y(y)^2
\end{align}
}
Next, we rewrite $p_0^\prime-p_1^\prime$ as 
\iftoggle{singlecolumn}{
\begin{align}
    p_0^\prime-p_1^\prime&=\sum\limits_{y\in\mathfrak{X}}\sum\limits_{w\in\mathfrak{X}} p_{W,Y}(w,y)\left[p_{Y|W}(y|w)-p_Y(y)\right]\\
    &= \sum\limits_{y\in\mathfrak{X}}\left[\left(\sum\limits_{w\in\mathfrak{X}} p_W(w) p_{Y|W}(y|w)^2\right)-p_Y(y)^2\right]\\
    &= \sum\limits_{y\in\mathfrak{X}} \psi(y)\\
    &\ge 0
\end{align}
}{
\begin{align}
    p_0^\prime-p_1^\prime&=\sum\limits_{y\in\mathfrak{X}}\sum\limits_{w\in\mathfrak{X}} p_{W,Y}(w,y)\left[p_{Y|W}(y|w)-p_Y(y)\right]\\
    &= \sum\limits_{y\in\mathfrak{X}}\left[\left(\sum\limits_{w\in\mathfrak{X}} p_W(w) p_{Y|W}(y|w)^2\right)-p_Y(y)^2\right]\\
    &= \sum\limits_{y\in\mathfrak{X}} \psi(y)\\
    &\ge 0
\end{align}
}
with $p_0^\prime-p_1^\prime=0$ only when $p_{Y|W}(y|w)=p_Y(y)$, $\forall (w,y)\in\mathfrak{X}^2$. In other words, $p_0^\prime>p_1^\prime$ and in turn $p_0^{(r)}>p_1$ as long as the two databases are not independent.

We next choose any $\smash{\tau\in(p_1,p_0^{(0)})}$ bounded away from both $\smash{p_0^{(0)}}$ and $p_1$. Let $A_j$ denote the event that $C^{m_n}_{j}$ and $C^{m_n}_{j+1}$ are noisy replicas and $B_j$ denote the event that the algorithm declares $C^{m_n}_{j}$ and $C^{m_n}_{j+1}$ as replicas. Via the union bound, we can upper bound the total probability of replica detection error as
\iftoggle{singlecolumn}{
\begin{align}
    \Pr(\bigcup\limits_{j=1}^{{K_n}-1} F_j)&\le \sum\limits_{j=1}^{{K_n}-1} \Pr(A_j ^c) \Pr(B_j|A_j ^c)+ \Pr(A_j)  \Pr(B_j^c|A_j)\label{eq:replicadetectionbound}
\end{align}
}{
\begin{align}
    \Pr(\bigcup\limits_{j=1}^{{K_n}-1} F_j)&\le \sum\limits_{j=1}^{{K_n}-1} \Pr(A_j ^c) \Pr(B_j|A_j ^c) \notag\\&\hspace{5em}+\Pr(A_j)  \Pr(B_j^c|A_j)\label{eq:replicadetectionbound}
\end{align}
}
\begin{sloppypar}
Note that conditioned on $A_j^c$, we have $H_j\sim\text{Binom}\smash{(m_n,p_0^{(r)})}$ and conditioned on $A_j$, we have $H_j\sim\text{Binom}(m_n,p_1)$. Then, from the Chernoff bound~\cite[Lemma 4.7.2]{ash2012information}, we get
\iftoggle{singlecolumn}{
\begin{align}
    \Pr(B_j|A_j ^c)&\le 2^{-m_n D\left(\tau\|\smash{p_0^{(r)}}\right)}\label{eq:chernoff1}\\
  &\le 2^{-m_n D\left(\tau\|\smash{p_0^{(0)}}\right)}\label{eq:chernoff1.5}\\
    \Pr(B_j^c|A_j)&\le 2^{-m_n D\left(1-\tau \|1-p_1\right)}\label{eq:chernoff2}
\end{align}
}{
\begin{align}
    \Pr(B_j|A_j ^c)&\le 2^{-m_n D\left(\tau\|\smash{p_0^{(r)}}\right)}\label{eq:chernoff1}\\
&\le 2^{-m_n D\left(\tau\|\smash{p_0^{(0)}}\right)}\label{eq:chernoff1.5}\\
    \Pr(B_j^c|A_j)&\le 2^{-m_n D\left(1-\tau \|1-p_1\right)}\label{eq:chernoff2}
\end{align}
}
where \eqref{eq:chernoff1.5} follows from the fact that $D(\tau\|p)$ is an increasing function of $p$ for $p>\tau$.

Thus, we get
\iftoggle{singlecolumn}{
\begin{align}
    \Pr(\bigcup\limits_{j=1}^{{K_n}-1} F_j)&\le ({K_n}-1)\left[ 2^{-m_n D\left(\tau\|\smash{p_0^{(0)}}\right)}+ 2^{-m_n D\left(1-\tau\|1-p_1\right)}\right]\label{eq:replicadetectionlast}
\end{align}
}{
\begin{align}
    \Pr(\bigcup\limits_{j=1}^{{K_n}-1} E_j)&\le K_n\left[ 2^{-m_n D\left(\tau\|\smash{p_0^{(0)}}\right)}+ 2^{-m_n D\left(1-\tau\|1-p_1\right)}\right]\label{eq:replicadetectionlast}
\end{align}
}

Observe that since the RHS of \eqref{eq:replicadetectionlast} has $2{K_n}=O(n)$ terms decaying exponentially in~$m_n$, for any $m_n=\omega(\log n)$ we have 
\iftoggle{singlecolumn}{
\begin{align}
    \Pr(\bigcup\limits_{j=1}^{{K_n}-1} F_j) \to 0 \:\text{ as } n\to\infty.
\end{align}
}{
\begin{align}
    \Pr(\bigcup\limits_{j=1}^{{K_n}-1} F_j) \to 0 \:\text{ as } n\to\infty.
\end{align}
}
Finally observing that $n\sim\log m_n$ concludes the proof.\qed
\end{sloppypar}

\section{Proof of Lemma~\ref{lem:seededdeletiondetection}}
\label{proof:seededdeletiondetection}
Let ${(\tilde{X}_{i,j},\tilde{Y}_{i,j})}$ be a pair of matching entries. Since the database distribution is stationary, WLOG, we can assume $(i,j)=(1,1)$. Now, given ${(\tilde{X}_{1,1},\tilde{Y}_{1,1})}$, and the non-matching pair ${(\tilde{X}_{1,j},\tilde{Y}_{1,1})}$ with $j-1=r\neq 0$, 
we first prove the existence of such a bijective mapping $\Phi$ such that for any $r\in[n-1]$
\iftoggle{singlecolumn}{
\begin{align}
    \Pr(\Phi(\tilde{Y}_{1,1})\neq \tilde{X}_{1,1})<\Pr(\Phi(\tilde{Y}_{1,1})\neq \tilde{X}_{1,r+1}).
\end{align}
}{
\begin{align}
    \Pr(\Phi(\tilde{Y}_{1,1})\neq \tilde{X}_{1,1})<\Pr(\Phi(\tilde{Y}_{1,1})\neq \tilde{X}_{1,r+1}).
\end{align}
}
For given $\Phi$ and $r\in[n-1]$ let
\iftoggle{singlecolumn}{
\begin{align}
    q_{0,\Phi}^{(r)}&\triangleq\Pr(\Phi(\tilde{Y}_{1,1})\neq \tilde{X}_{1,r+1})\\
    q_{1,\Phi}&\triangleq\Pr(\Phi(\tilde{Y}_{1,1})\neq \tilde{X}_{1,1})
\end{align}
}{
\begin{align}
    q_{0,\Phi}^{(r)}&\triangleq\Pr(\Phi(\tilde{Y}_{1,1})\neq \tilde{X}_{1,r+1})\\
    q_{1,\Phi}&\triangleq\Pr(\Phi(\tilde{Y}_{1,1})\neq \tilde{X}_{1,1})
\end{align}
}
Here, our goal is to show that there exists at least one $\Phi$ satisfying
\iftoggle{singlecolumn}{
\begin{align}
    q_{0,\Phi}^{(r)}>q_{1,\Phi}\label{eq:q0biggerthanq1},\:\forall r\in[n-1].
\end{align}
}{
\begin{align}
    q_{0,\Phi}^{(r)}>q_{1,\Phi}\label{eq:q0biggerthanq1},\:\forall r\in[n-1].
\end{align}
}
We can rewrite $q_{0,\Phi}^{(r)}$ as
\iftoggle{singlecolumn}{
\begin{align}
    q_{0,\Phi}^{(r)}&=\sum\limits_{x_1\in\mathfrak{X}}\sum\limits_{x_2\in\mathfrak{X}} \Pr(\tilde{X}_{1,1}=x_1) \Pr(\tilde{X}_{1,r+1}=x_2|\tilde{X}_{1,1}=x_1) \Pr(\Phi(\tilde{Y}_{1,1})\neq x_2|\tilde{X}_{1,1} = x_1) \\
    &= \sum\limits_{i=1}^{|\mathfrak{X}|}\sum\limits_{j=1}^{|\mathfrak{X}|} u_i\: (\mathbf{P}^r)_{i,j}\: [1-p_{Y|X}(\Phi^{-1}(j)|i)]\\
    &= \sum\limits_{i=1}^{|\mathfrak{X}|}\sum\limits_{j=1}^{|\mathfrak{X}|} u_i\: [(1-\gamma^r)u_j+\gamma^r \delta_{i j}]\: [1-p_{Y|X}(\Phi^{-1}(j)|i)]\\
    &=(1-\gamma^r) \sum\limits_{i=1}^{|\mathfrak{X}|}\sum\limits_{j=1}^{|\mathfrak{X}|} u_i\: u_j\: [1-p_{Y|X}(\Phi^{-1}(j)|i)] + \gamma^r \sum\limits_{i=1}^{|\mathfrak{X}|} u_i\:  [1-p_{Y|X}(\Phi^{-1}(i)|i)]\\
    &= (1-\gamma^r) q_{0,\Phi}^\prime+ \gamma^r q_{1,\Phi}^\prime
\end{align}
}{
\begin{align}
    q_{0,\Phi}^{(r)}&=\sum\limits_{x_1\in\mathfrak{X}}\sum\limits_{x_2\in\mathfrak{X}} \Pr(\tilde{X}_{1,1}=x_1) \notag\\&\hspace{4em}\Pr(\tilde{X}_{1,r+1}=x_2|\tilde{X}_{1,1}=x_1)\notag\\&\hspace{4em} \Pr(\Phi(\tilde{Y}_{1,1})\neq x_2|\tilde{X}_{1,1} = x_1) \\
    &= \sum\limits_{i=1}^{|\mathfrak{X}|}\sum\limits_{j=1}^{|\mathfrak{X}|} u_i\: (\mathbf{P}^r)_{i,j}\: [1-p_{Y|X}(\Phi^{-1}(j)|i)]\\
    &= \sum\limits_{i=1}^{|\mathfrak{X}|}\sum\limits_{j=1}^{|\mathfrak{X}|} u_i\: [(1-\gamma^r)u_j+\gamma^r \delta_{i j}]\notag\\&\hspace{6em} [1-p_{Y|X}(\Phi^{-1}(j)|i)]\\
    &=(1-\gamma^r) \sum\limits_{i=1}^{|\mathfrak{X}|}\sum\limits_{j=1}^{|\mathfrak{X}|} u_i\: u_j\: [1-p_{Y|X}(\Phi^{-1}(j)|i)] \notag\\ &\hspace{3.5em}+ \gamma^r \sum\limits_{i=1}^{|\mathfrak{X}|} u_i\:  [1-p_{Y|X}(\Phi^{-1}(i)|i)]\\
    &= (1-\gamma^r) q_{0,\Phi}^\prime+ \gamma^r q_{1,\Phi}^\prime
\end{align}
}
where 
\iftoggle{singlecolumn}{
\begin{align}
    q_{0,\Phi}^\prime&\triangleq \sum\limits_{i=1}^{|\mathfrak{X}|}\sum\limits_{j=1}^{|\mathfrak{X}|} u_i\: u_j\: [1-p_{Y|X}(\Phi^{-1}(j)|i)]\\
    q_{1,\Phi}^\prime &\triangleq \sum\limits_{i=1}^{|\mathfrak{X}|} u_i\:  [1-p_{Y|X}(\Phi^{-1}(i)|i)]
\end{align}
}{
\begin{align}
    q_{0,\Phi}^\prime&\triangleq \sum\limits_{i=1}^{|\mathfrak{X}|}\sum\limits_{j=1}^{|\mathfrak{X}|} u_i\: u_j\: [1-p_{Y|X}(\Phi^{-1}(j)|i)]\\
    q_{1,\Phi}^\prime &\triangleq \sum\limits_{i=1}^{|\mathfrak{X}|} u_i\:  [1-p_{Y|X}(\Phi^{-1}(i)|i)]
\end{align}
}
Similarly, we rewrite $q_{1,\Phi}$ as
\iftoggle{singlecolumn}{
\begin{align}
    q_{1,\Phi}&=\sum\limits_{x\in\mathfrak{X}} \Pr(\tilde{X}_{1,1}=x) \Pr(\Phi(\tilde{Y}_{1,1})\neq x|\tilde{X}_{1,1} = x)  \\
    &=\sum\limits_{i=1}^{|\mathfrak{X}|} u_i [1-p_{Y|X}(\Phi^{-1}(i)|i)]\\
    &= q_{1,\Phi}^\prime
\end{align}
}{
\begin{align}
    q_{1,\Phi}&=\sum\limits_{x\in\mathfrak{X}} \Pr(\tilde{X}_{1,1}=x) \Pr(\Phi(\tilde{Y}_{1,1})\neq x|\tilde{X}_{1,1} = x)  \\
    &=\sum\limits_{i=1}^{|\mathfrak{X}|} u_i [1-p_{Y|X}(\Phi^{-1}(i)|i)]\\
    &= q_{1,\Phi}^\prime
\end{align}
}
Thus, for any $\gamma\in[0,1)$, we have
\iftoggle{singlecolumn}{
\begin{align}
    \exists \Phi,\: \forall r\in[n-1], \: q_{0,\Phi}^{(r)}>q_{1,\Phi} \iff \exists \Phi,\: q_{0,\Phi}^\prime>q_{1,\Phi}^\prime \label{eq:sddequivalencetoiid}
\end{align}
}{
\begin{align}
    \exists \Phi,\: \forall r\in[n-1], \: q_{0,\Phi}^{(r)}>q_{1,\Phi} \iff \exists \Phi,\: q_{0,\Phi}^\prime>q_{1,\Phi}^\prime \label{eq:sddequivalencetoiid}
\end{align}
}
Note that $q_{0,\Phi}^\prime$ and $q_{1,\Phi}^\prime$ correspond to
\iftoggle{singlecolumn}{
\begin{align}
    q_{0,\Phi}^\prime&=\Pr(\Phi(\tilde{Y}_{1,1})\neq \tilde{X}_{1,j}),\hspace{1em} j\neq 1\\
    q_{1,\Phi}^\prime&=\Pr(\Phi(\tilde{Y}_{1,1})\neq \tilde{X}_{1,1})
\end{align}
}{
\begin{align}
    q_{0,\Phi}^\prime&=\Pr(\Phi(\tilde{Y}_{1,1})\neq \tilde{X}_{1,j}),\hspace{1em} j\neq 1\\
    q_{1,\Phi}^\prime&=\Pr(\Phi(\tilde{Y}_{1,1})\neq \tilde{X}_{1,1})
\end{align}
}
if the entries $\tilde{X}_{i,j}$ of $\mathbf{G}^{(1)}$ were drawn \emph{i.i.d.} from the distribution $\pi=[u_1,\dots,u_{|\mathfrak{X}|}]$, instead of a Markov process. Thus, we recall the discrete random variable $W$, defined in equations~\eqref{eq:iidW1}-\eqref{eq:iidW2}, with 
\iftoggle{singlecolumn}{
\begin{align}
    p_W(i) &= u_i,\:\forall i\in\mathfrak{X}\\
    p_{Y|W}(y|w)&=p_{Y|X}(y|w),\:\forall (w,y)\in\mathfrak{X}^2
\end{align}
}{
\begin{align}
    p_W(i) &= u_i,\:\forall i\in\mathfrak{X}\\
    p_{Y|W}(y|w)&=p_{Y|X}(y|w),\:\forall (w,y)\in\mathfrak{X}^2
\end{align}
}
We note that similar to Appendix~\ref{proof:noisyreplicadetection}, this equivalence induced by~\eqref{eq:sddequivalencetoiid} is due to the specific Markov structure we adopted in Definition~\ref{defn:markovdb}.

Then, we can rewrite $q_{0,\Phi}^\prime$ and $q_{1,\Phi}^\prime$ as
\iftoggle{singlecolumn}{
\begin{align}
    q_{0,\Phi}^\prime&= \sum\limits_{w_1\in\mathfrak{X}}\sum\limits_{w_2\in\mathfrak{X}}p_W(w_1) p_W(w_2)[1-p_{Y|X}(\Phi^{-1}(w_2)|w_1)]\\
    q_{1,\Phi}^\prime
    &= \sum\limits_{w\in\mathfrak{X}}p_W(w) [1-p_{Y|W}(\Phi^{-1}(w)|w)]
\end{align}
}{
\begin{align}
    q_{0,\Phi}^\prime&= \sum\limits_{w_1\in\mathfrak{X}}\sum\limits_{w_2\in\mathfrak{X}}p_W(w_1) p_W(w_2)\notag\\ &\hspace{5em}[1-p_{Y|X}(\Phi^{-1}(w_2)|w_1)]\\
    q_{1,\Phi}^\prime
    &= \sum\limits_{w\in\mathfrak{X}}p_W(w) [1-p_{Y|W}(\Phi^{-1}(w)|w)]
\end{align}
}

We first prove the following:
\iftoggle{singlecolumn}{
\begin{align}
    \sum\limits_{\Phi} q_{0,\Phi}^\prime-q_{1,\Phi}^\prime=0\label{eq:q0q1overphizero}
\end{align}
}{
\begin{align}
    \sum\limits_{\Phi} q_{0,\Phi}^\prime-q_{1,\Phi}^\prime=0\label{eq:q0q1overphizero}
\end{align}
}
where the summation is over all permutations of $\mathfrak{X}$. For brevity, let 
\iftoggle{singlecolumn}{
\begin{align}
    Q_{i,j}\triangleq p_{Y|W}(j|i)\quad \forall i,j\in\mathfrak{X} \label{eq:definep}
\end{align}
}{
\begin{align}
    Q_{i,j}\triangleq p_{Y|W}(j|i)\quad \forall i,j\in\mathfrak{X} \label{eq:definep}
\end{align}
}

Note that from \eqref{eq:definep}, we have
\iftoggle{singlecolumn}{
\begin{align}
    \sum\limits_{j=1}^{|\mathfrak{X}|}& Q_{i,j}=1\quad \forall i\in \mathfrak{X}\label{eq:pijsumto1}\\
    \sum\limits_{i=1}^{|\mathfrak{X}|}&\sum\limits_{j=1}^{|\mathfrak{X}|} Q_{i,j}=|\mathfrak{X}|\label{eq:pijsumtoX}
\end{align}
}{
\begin{align}
    \sum\limits_{j=1}^{|\mathfrak{X}|}& Q_{i,j}=1\quad \forall i\in \mathfrak{X}\label{eq:pijsumto1}\\
    \sum\limits_{i=1}^{|\mathfrak{X}|}&\sum\limits_{j=1}^{|\mathfrak{X}|} Q_{i,j}=|\mathfrak{X}|\label{eq:pijsumtoX}
\end{align}
}
Taking the sum over all $\Phi$, we obtain
\iftoggle{singlecolumn}{
\begin{align}
    \sum\limits_{\Phi} q_{0,\Phi}^\prime-q_{1,\Phi}^\prime
    &= \sum\limits_{\Phi}\sum\limits_{i=1}^{|\mathfrak{X}|}\sum\limits_{j=1}^{|\mathfrak{X}|} p_W(i) p_W(j) Q_{i,\Phi^{-1}(j)}-\sum\limits_{\Phi}\sum\limits_{i=1}^{|\mathfrak{X}|} p_W(i) Q_{i,\Phi^{-1}(i)}\label{eq:q0q1sumoverphi}
\end{align}
}{
\begin{align}
    \sum\limits_{\Phi} q_{0,\Phi}^\prime-q_{1,\Phi}^\prime
    &= \sum\limits_{\Phi}\sum\limits_{i=1}^{|\mathfrak{X}|}\sum\limits_{j=1}^{|\mathfrak{X}|} p_W(i) p_W(j) Q_{i,\Phi^{-1}(j)}\notag\\&\hspace{3em}-\sum\limits_{\Phi}\sum\limits_{i=1}^{|\mathfrak{X}|} p_W(i) Q_{i,\Phi^{-1}(i)}\label{eq:q0q1sumoverphi}
\end{align}
}
Combining \eqref{eq:pijsumto1}-\eqref{eq:q0q1sumoverphi}, we now show that both terms on the RHS of \eqref{eq:q0q1sumoverphi} are equal to $(|\mathfrak{X}|-1)!$.
We first look at the second term on the RHS of \eqref{eq:q0q1sumoverphi}.
\iftoggle{singlecolumn}{
\begin{align}
    \sum\limits_{\Phi}\sum\limits_{i=1}^{|\mathfrak{X}|} p_W(i) Q_{i,\Phi^{-1}(i)}&=\sum\limits_{i=1}^{|\mathfrak{X}|} p_W(i) \sum\limits_{\Phi} Q_{i,\Phi^{-1}(i)}\\
    &=(|\mathfrak{X}|-1)!\sum\limits_{j=1}^{|\mathfrak{X}|}\sum\limits_{i=1}^{|\mathfrak{X}|} p_W(i) Q_{i,j}\label{eq:permlhs}\\
    &= (|\mathfrak{X}|-1)!\sum\limits_{i=1}^{|\mathfrak{X}|}\sum\limits_{j=1}^{|\mathfrak{X}|} p_{W,Y}(i,j)\\
    &= (|\mathfrak{X}|-1)!\label{eq:philhs}
\end{align}
}{
\begin{align}
    \sum\limits_{\Phi}&\sum\limits_{i=1}^{|\mathfrak{X}|} p_W(i) Q_{i,\Phi^{-1}(i)}\notag\\&=\sum\limits_{i=1}^{|\mathfrak{X}|} p_W(i) \sum\limits_{\Phi} Q_{i,\Phi^{-1}(i)}\\
    &=(|\mathfrak{X}|-1)!\sum\limits_{j=1}^{|\mathfrak{X}|}\sum\limits_{i=1}^{|\mathfrak{X}|} p_W(i) Q_{i,j}\label{eq:permlhs}\\
    &= (|\mathfrak{X}|-1)!\sum\limits_{i=1}^{|\mathfrak{X}|}\sum\limits_{j=1}^{|\mathfrak{X}|} p_{W,Y}(i,j)\\
    &= (|\mathfrak{X}|-1)!\label{eq:philhs}
\end{align}
}
where \eqref{eq:permlhs} follows from the fact that for any $j\in\mathfrak{X}$, we have exactly $(|\mathfrak{X}|-1)!$ permutations assigning $j$ to $i$ (or equivalently $\Phi^{-1}(i)=j$). 

Now we look at the first term.
\iftoggle{singlecolumn}{
\begin{align}
    \sum\limits_{\Phi}\sum\limits_{i=1}^{|\mathfrak{X}|}\sum\limits_{j=1}^{|\mathfrak{X}|} p_W(i)p_W(j) Q_{i,\Phi^{-1}(j)}&= \sum\limits_{i=1}^{|\mathfrak{X}|}\sum\limits_{j=1}^{|\mathfrak{X}|} p_W(i) p_W(j)  \sum\limits_{\Phi} Q_{i,\Phi^{-1}(j)}\\
    &=\sum\limits_{i=1}^{|\mathfrak{X}|}\sum\limits_{j=1}^{|\mathfrak{X}|} p_W(i) p_W(j) (|\mathfrak{X}|-1)!\sum\limits_{k=1}^{|\mathfrak{X}|}Q_{i,k}\label{eq:permrhs}\\
    &= (|\mathfrak{X}|-1)!\label{eq:phirhs}
\end{align}
}{
\begin{align}
    \sum\limits_{\Phi}&\sum\limits_{i=1}^{|\mathfrak{X}|}\sum\limits_{j=1}^{|\mathfrak{X}|} p_W(i)p_W(j) Q_{i,\Phi^{-1}(j)}\notag\\&= \sum\limits_{i=1}^{|\mathfrak{X}|}\sum\limits_{j=1}^{|\mathfrak{X}|} p_W(i) p_W(j)  \sum\limits_{\Phi} Q_{i,\Phi^{-1}(j)}\\
    &=\sum\limits_{i=1}^{|\mathfrak{X}|}\sum\limits_{j=1}^{|\mathfrak{X}|} p_W(i) p_W(j) (|\mathfrak{X}|-1)!\sum\limits_{k=1}^{|\mathfrak{X}|}Q_{i,k}\label{eq:permrhs}\\
    &= (|\mathfrak{X}|-1)!\label{eq:phirhs}
\end{align}
}
Again, \eqref{eq:permrhs} follows from the fact that for each $k\in\mathfrak{X}$, there are exactly $(|\mathfrak{X}|-1)!$ permutations $\Phi$ which map $k$ to $j$ (or equivalently $\Phi^{-1}(j)=k$).

Thus, we have shown that both terms on the RHS of \eqref{eq:q0q1sumoverphi} are equal to $(|\mathfrak{X}|-1)!$, proving~\eqref{eq:q0q1overphizero}. Now, we only need to show that
\iftoggle{singlecolumn}{
\begin{align}
    \exists \Phi\quad q_{0,\Phi}^\prime-q_{1,\Phi}^\prime\neq 0. \label{eq:phifinal}
\end{align}
}{
\begin{align}
    \exists \Phi\quad q_{0,\Phi}^\prime-q_{1,\Phi}^\prime\neq 0. \label{eq:phifinal}
\end{align}
}
This is because unless $q_{0,\Phi}^\prime-q_{1,\Phi}^\prime= 0$ $\forall \Phi$, due to~\eqref{eq:q0q1overphizero}, we automatically have a $\Phi$ such that this difference is strictly positive. This follows from the fact if $\exists\Phi$ $q_{0,\Phi}^\prime-q_{1,\Phi}^\prime\neq 0$, we have either
\begin{itemize}
    \item $q_{0,\Phi}^\prime-q_{1,\Phi}^\prime> 0$, which is the desired result, or
    \item $q_{0,\Phi}^\prime-q_{1,\Phi}^\prime<0$, which from~\eqref{eq:q0q1sumoverphi} requires the existence of another permutation $\tilde{\Phi}$ with $q_{0,\tilde{\Phi}}^\prime-q_{1,\tilde{\Phi}}^\prime> 0$.
\end{itemize}

We will prove~\eqref{eq:phifinal} by arguing that 
\iftoggle{singlecolumn}{
\begin{align}
    q_{0,\Phi}^\prime-q_{1,\Phi}^\prime= 0\quad \forall \Phi \iff p_{Y|W}(y|w)=p_Y(y)\quad \forall (w,y)\in\mathfrak{X}^2 \label{eq:phicondition}
\end{align}
}{
\begin{align}
    q_{0,\Phi}^\prime&-q_{1,\Phi}^\prime= 0\hspace{0.5em} \forall \Phi \notag\\&\iff p_{Y|W}(y|w)=p_Y(y)\hspace{0.5em} \forall (w,y)\in\mathfrak{X}^2 \label{eq:phicondition}
\end{align}
}
which contradicts our $p_{Y|X}\neq p_Y$ assumption.

We first prove the \say{only if} part. Suppose ${p_{Y|W}(y|w)=p_Y(y)}$, $\forall (w,y)\in\mathfrak{X}^2$. In other words, $Q_{i,k}=Q_{j,k}$, $\forall{(i,j,k)\in\mathfrak{X}^3}$. Then for any $\Phi$, we have
\iftoggle{singlecolumn}{
\begin{align}
    q_{0,\Phi}^\prime&=\sum\limits_{i=1}^{|\mathfrak{X}|}\sum\limits_{j=1}^{|\mathfrak{X}|} p_W(i) p_W(j) Q_{i,\Phi^{-1}(j)}\\&= \sum\limits_{i=1}^{|\mathfrak{X}|}\sum\limits_{j=1}^{|\mathfrak{X}|} p_W(i) p_W(j) Q_{k,\Phi^{-1}(j)}, \hspace{1em} k\neq i\\
    &=\sum\limits_{j=1}^{|\mathfrak{X}|}  p_W(j) Q_{k,\Phi^{-1}(j)}\\
    &=\sum\limits_{j=1}^{|\mathfrak{X}|}  p_W(j) Q_{j,\Phi^{-1}(j)}\\
    &= q_{1,\Phi}^\prime
\end{align}
}{
\begin{align}
    q_{0,\Phi}^\prime&=\sum\limits_{i=1}^{|\mathfrak{X}|}\sum\limits_{j=1}^{|\mathfrak{X}|} p_W(i) p_W(j) Q_{i,\Phi^{-1}(j)}\\&= \sum\limits_{i=1}^{|\mathfrak{X}|}\sum\limits_{j=1}^{|\mathfrak{X}|} p_W(i) p_W(j) Q_{k,\Phi^{-1}(j)}, \hspace{1em} k\neq i\\
    &=\sum\limits_{j=1}^{|\mathfrak{X}|}  p_W(j) Q_{k,\Phi^{-1}(j)}\\
    &=\sum\limits_{j=1}^{|\mathfrak{X}|}  p_W(j) Q_{j,\Phi^{-1}(j)}\\
    &= q_{1,\Phi}^\prime
\end{align}
}
finishing the proof of the \say{only if} part.

Now, we prove the \say{if} part. Suppose the LHS of \eqref{eq:phicondition} holds. In other words, for any $\Phi$
\iftoggle{singlecolumn}{
\begin{align}
 \sum\limits_{i=1}^{|\mathfrak{X}|} p_W(i) Q_{i,\Phi^{-1}(i)} = \sum\limits_{i=1}^{|\mathfrak{X}|}\sum\limits_{j=1}^{|\mathfrak{X}|} p_W(i) p_W(j) Q_{i,\Phi^{-1}(j)} \label{eq:lhsequivalence}
\end{align}
}{
\begin{align}
 \sum\limits_{i=1}^{|\mathfrak{X}|} p_W(i) Q_{i,\Phi^{-1}(i)} = \sum\limits_{i=1}^{|\mathfrak{X}|}\sum\limits_{j=1}^{|\mathfrak{X}|} p_W(i) p_W(j) Q_{i,\Phi^{-1}(j)} \label{eq:lhsequivalence}
\end{align}
}

First, we look at the binary case $\mathfrak{X}=\{1,2\}$. In this case, we obtain
\iftoggle{singlecolumn}{
\begin{align}
    p_W(1) Q_{1,1}+p_W(2) Q_{2,2} &= p_W(1)^2 Q_{1,1}+ p_W(1) p_W(2) Q_{1,2}\notag\\
    &\hspace{4em}+ p_W(2) p_W(1) Q_{2,1}+ p_W(2)^2 Q_{2,2}\\
    Q_{1,1}+Q_{2,2}&=Q_{1,2}+Q_{2,1}\\
    Q_{1,1}+Q_{2,2}&=1-Q_{1,1}+1-Q_{2,2}\\
    Q_{1,1}+Q_{2,2}&= 1
\end{align}
}{
\begin{align}
    p_W(1)& Q_{1,1}+p_W(2) Q_{2,2} \notag\\&= p_W(1)^2 Q_{1,1}+ p_W(1) p_W(2) Q_{1,2}\notag\\
    &+ p_W(2) p_W(1) Q_{2,1}+ p_W(2)^2 Q_{2,2}
\end{align}
\begin{align}
       Q_{1,1}+Q_{2,2}&=Q_{1,2}+Q_{2,1}\\
    Q_{1,1}+Q_{2,2}&=1-Q_{1,1}+1-Q_{2,2}\\
    Q_{1,1}+Q_{2,2}&= 1
\end{align}
}
for the identity permutation. This implies that  $Q_{1,1}=Q_{2,1}$ and $Q_{1,2}=Q_{2,2}$ and this in turn implies ${p_{Y|W}(y|w)=p_Y(y)}$ ${\forall (w,y)\in\mathfrak{X}^2}$, concluding the proof for the binary case.

Now, we investigate the larger alphabet sizes ($|\mathfrak{X}|\ge 3$). Since the equality holds for all $\Phi$, we now carefully select some one-cycle permutations $\Phi$ to construct a system of linear equations. 

Let $\Phi_{\text{id}}$ be the identity permutation and $\Phi_{i-j},\Phi_{i-k},\Phi_{i-j-k}$ denote the one-cycle permutations with the respective cycles $(i\: j)$, $(i\: k)$ and $(i\: j\: k)$ for some distinct $(i,j,k)$ triplet. For the rest of this proof, we will jointly solve the system of equations put forward by these permutations.

Recall that $p_W(l)=u_l$, $\forall l\in\mathfrak{X}$. Then, $\Phi_{\text{id}}$ leads to
\iftoggle{singlecolumn}{
\begin{align}
u_i Q_{i,i}+u_j Q_{j,j} 
+u_k Q_{k,k}+\sum\limits_{l\neq i,j,k}u_l Q_{l,l}  
&=u_i \sum\limits_{t=1}^{|\mathfrak{X}|} u_t Q_{t,i}+ u_j \sum\limits_{t=1}^{|\mathfrak{X}|} u_t Q_{t,j}  + u_k \sum\limits_{t=1}^{|\mathfrak{X}|} u_t Q_{t,k} \notag\\ &\hspace{2em}+ \sum\limits_{l\neq i,j,k} u_l \sum\limits_{t=1}^{|\mathfrak{X}|} u_t Q_{t,l}\label{eq:phiid}
\end{align}
}{
\begin{align}
u_i Q_{i,i}&+u_j Q_{j,j} 
+u_k Q_{k,k}+\sum\limits_{l\neq i,j,k}u_l Q_{l,l}  \notag\\
&=u_i \sum\limits_{t=1}^{|\mathfrak{X}|} u_t Q_{t,i}+ u_j \sum\limits_{t=1}^{|\mathfrak{X}|} u_t Q_{t,j}  \notag\\&\hspace{2em}+ u_k \sum\limits_{t=1}^{|\mathfrak{X}|} u_t Q_{t,k} + \sum\limits_{l\neq i,j,k} u_l \sum\limits_{t=1}^{|\mathfrak{X}|} u_t Q_{t,l}\label{eq:phiid}
\end{align}
}

Similarly, $\Phi_{i-j}$ leads to
\iftoggle{singlecolumn}{
\begin{align}
    u_i Q_{i,j}+u_j Q_{j,i}+u_k Q_{k,k}+\sum\limits_{l\neq i,j,k}u_l Q_{l,l}
    &= u_i \sum\limits_{t=1}^{|\mathfrak{X}|} u_t Q_{t,j}+ u_j \sum\limits_{t=1}^{|\mathfrak{X}|} u_t Q_{t,i}
    + u_k \sum\limits_{t=1}^{|\mathfrak{X}|} u_t Q_{t,k} \notag\\ &\hspace{2em}+ \sum\limits_{l\neq i,j,k} u_l \sum\limits_{t=1}^{|\mathfrak{X}|} u_t Q_{t,l}\label{eq:phiij}
\end{align}
}{
\begin{align}
    u_i Q_{i,j}&+u_j Q_{j,i}+u_k Q_{k,k}+\sum\limits_{l\neq i,j,k}u_l Q_{l,l}\notag\\
    &= u_i \sum\limits_{t=1}^{|\mathfrak{X}|} u_t Q_{t,j}+ u_j \sum\limits_{t=1}^{|\mathfrak{X}|} u_t Q_{t,i}
   \notag\\&\hspace{2em} + u_k \sum\limits_{t=1}^{|\mathfrak{X}|} u_t Q_{t,k} + \sum\limits_{l\neq i,j,k} u_l \sum\limits_{t=1}^{|\mathfrak{X}|} u_t Q_{t,l}\label{eq:phiij}
\end{align}
}
When we subtract \eqref{eq:phiij} from \eqref{eq:phiid}, we obtain
\iftoggle{singlecolumn}{
\begin{align}
    u_i (Q_{i,i}-Q_{i,j})-u_j (Q_{j,i}-Q_{j,j}) &= (u_i-u_j) \sum\limits_{t=1}^{|\mathfrak{X}|} u_t (Q_{t,i}-Q_{t,j})
\end{align}
}{
\begin{align}
    u_i (Q_{i,i}-Q_{i,j})&-u_j (Q_{j,i}-Q_{j,j}) \notag\\&= (u_i-u_j) \sum\limits_{t=1}^{|\mathfrak{X}|} u_t (Q_{t,i}-Q_{t,j})
\end{align}
}
Equivalently, we have
\iftoggle{singlecolumn}{
\begin{align}
    p_{W,Y}(i,i)-p_{W,Y}(i,j)-p_{W,Y}(j,i)+p_{W,Y}(j,j)&=p_W(i) p_Y(i)-p_W(i) p_Y(j)\notag\\ &\hspace{2em}-p_W(j) p_Y(i)+p_W(j) p_Y(j)
\end{align}
}{
\begin{align}
    p_{W,Y}(i,i)&-p_{W,Y}(i,j)-p_{W,Y}(j,i)+p_{W,Y}(j,j)\notag\\ &=p_W(i) p_Y(i)-p_W(i) p_Y(j)\notag\\&\hspace{2em}-p_W(j) p_Y(i)+p_W(j) p_Y(j)
\end{align}
}

Following the same steps, from $\Phi_{i-k}$ we get
\iftoggle{singlecolumn}{
\begin{align}
    p_{W,Y}(i,i)-p_{W,Y}(i,k)-p_{W,Y}(k,i)+p_{W,Y}(k,k)&=p_W(i) p_Y(i)-p_W(i) p_Y(k)\notag\\ &\hspace{2em}-p_W(k) p_Y(i)+p_W(k) p_Y(k)\label{eq:phiikrearrange}
\end{align}
}{
\begin{align}
    p_{W,Y}(i,i)&-p_{W,Y}(i,k)-p_{W,Y}(k,i)+p_{W,Y}(k,k)\notag\\ &=p_W(i) p_Y(i)-p_W(i) p_Y(k)\notag\\&\hspace{2em}-p_W(k) p_Y(i)+p_W(k) p_Y(k)\label{eq:phiikrearrange}
\end{align}
}

We can rearrange the terms in \eqref{eq:phiikrearrange} to obtain
\iftoggle{singlecolumn}{
\begin{align}
    p_{W,Y}(i,k)&= p_{W,Y}(i,i)-p_{W,Y}(k,i)+p_{W,Y}(k,k)-p_W(i) p_Y(i)\notag\\ &\hspace{2em}+p_W(i) p_Y(k)+p_W(k) p_Y(i)-p_W(k) p_Y(k) \label{eq:pxyik}
\end{align}
}{
\begin{align}
    p_{W,Y}(i,k)&= p_{W,Y}(i,i)-p_{W,Y}(k,i)+p_{W,Y}(k,k)\notag\\&\hspace{2em}-p_W(i) p_Y(i)+p_W(i) p_Y(k)\notag\\&\hspace{2em}+p_W(k) p_Y(i)-p_W(k) p_Y(k) \label{eq:pxyik}
\end{align}
}

Furthermore, $\Phi_{i-j-k}$ gives us
\iftoggle{singlecolumn}{
\begin{align}
    u_i Q_{i,k}+u_j Q_{j,i}+u_k Q_{k,j}+\sum\limits_{l\neq i,j,k}u_l Q_{l,l}&= u_i \sum\limits_{t=1}^{|\mathfrak{X}|} u_t Q_{t,k}+ u_j \sum\limits_{t=1}^{|\mathfrak{X}|} u_t Q_{t,i}\notag\\&\hspace{2em}+ u_k \sum\limits_{t=1}^{|\mathfrak{X}|} u_t Q_{t,j} + \sum\limits_{l\neq i,j,k} u_l \sum\limits_{t=1}^{|\mathfrak{X}|} u_t Q_{t,l}\label{eq:phiijk}
\end{align}
}{
\begin{align}
    u_i Q_{i,k}&+u_j Q_{j,i}+u_k Q_{k,j}+\sum\limits_{l\neq i,j,k}u_l Q_{l,l}\notag\\&= u_i \sum\limits_{t=1}^{|\mathfrak{X}|} u_t Q_{t,k}+ u_j \sum\limits_{t=1}^{|\mathfrak{X}|} u_t Q_{t,i}
    \notag\\&\hspace{2em}+ u_k \sum\limits_{t=1}^{|\mathfrak{X}|} u_t Q_{t,j} + \sum\limits_{l\neq i,j,k} u_l \sum\limits_{t=1}^{|\mathfrak{X}|} u_t Q_{t,l}\label{eq:phiijk}
\end{align}
}

Subtracting \eqref{eq:phiijk} from \eqref{eq:phiij} yields
\iftoggle{singlecolumn}{
\begin{align}
    u_i(Q_{i,j}-Q_{i,k})+u_k(Q_{k,k}-Q_{k,j})&= (u_i-u_k)\left[ \sum\limits_{t=1}^{|\mathfrak{X}|} u_t Q_{t,j} - \sum\limits_{t=1}^{|\mathfrak{X}|} u_t Q_{t,k}\right]
\end{align}
}{
\begin{align}
    u_i(Q_{i,j}&-Q_{i,k})+u_k(Q_{k,k}-Q_{k,j})\notag\\&= (u_i-u_k)\left[ \sum\limits_{t=1}^{|\mathfrak{X}|} u_t Q_{t,j} - \sum\limits_{t=1}^{|\mathfrak{X}|} u_t Q_{t,k}\right]
\end{align}
}

Equivalently,
\iftoggle{singlecolumn}{
\begin{align}
    p_{W,Y}(i,j)-p_{W,Y}(i,k)-p_{W,Y}(k,j)+p_{W,Y}(k,k) &=p_W(i) p_Y(j)-p_W(i) p_Y(k)\notag\\ &\hspace{2em}-p_W(k) p_Y(j)+p_W(k) p_Y(k)\label{eq:pxyiksub}
\end{align}
}{
\begin{align}
    p_{W,Y}(i,j)&-p_{W,Y}(i,k)-p_{W,Y}(k,j)+p_{W,Y}(k,k)\notag\\ &=p_W(i) p_Y(j)-p_W(i) p_Y(k)\notag\\&\hspace{2em}-p_W(k) p_Y(j)+p_W(k) p_Y(k)\label{eq:pxyiksub}
\end{align}
}

Plugging $p_{W,Y}(i,k)$ from \eqref{eq:pxyik} into \eqref{eq:pxyiksub} yields
\iftoggle{singlecolumn}{
\begin{align}
    p_{W,Y}(i,j)-p_{W,Y}(i,i)-p_{W,Y}(k,j)+p_{W,Y}(k,i) &=p_W(i) p_Y(j)-p_W(i) p_Y(i)\notag\\ &\hspace{2em}-p_W(k) p_Y(j)+p_W(k) p_Y(i) \label{eq:phisumk}
\end{align}
}{
\begin{align}
    p_{W,Y}(i,j)&-p_{W,Y}(i,i)-p_{W,Y}(k,j)+p_{W,Y}(k,i)\notag\\ &=p_W(i) p_Y(j)-p_W(i) p_Y(i)\notag\\&\hspace{2em}-p_W(k) p_Y(j)+p_W(k) p_Y(i) \label{eq:phisumk}
\end{align}
}

Taking a summation over $k$ in \eqref{eq:phisumk} gives us
\iftoggle{singlecolumn}{
\begin{align}
    |\mathfrak{X}| p_{W,Y}(i,j)-|\mathfrak{X}| p_{W,Y}(i,i)-p_{Y}(j)+p_{Y}(i) &=|\mathfrak{X}| p_W(i) p_Y(j)-|\mathfrak{X}| p_W(i) p_Y(i)\notag\\ &\hspace{2em}- p_Y(j)+ p_Y(i)\\
    p_{W,Y}(i,j)- p_{W,Y}(i,i)&=p_W(i) p_Y(j)- p_W(i) p_Y(i)\label{eq:pxyiisub}
\end{align}
}{
\begin{align}
    |\mathfrak{X}| p_{W,Y}(i,j)&-|\mathfrak{X}| p_{W,Y}(i,i)-p_{Y}(j)+p_{Y}(i)\notag\\ &=|\mathfrak{X}| p_W(i) p_Y(j)-|\mathfrak{X}| p_W(i) p_Y(i)\notag\\&\hspace{2em}- p_Y(j)+ p_Y(i)\\
    p_{W,Y}(i,j)&- p_{W,Y}(i,i)\notag\\&=p_W(i) p_Y(j)- p_W(i) p_Y(i)\label{eq:pxyiisub}
\end{align}
}

Similarly, taking a summation over $j$ in \eqref{eq:pxyiisub} yields
\iftoggle{singlecolumn}{
\begin{align}
    p_{W}(i)-|\mathfrak{X}| p_{W,Y}(i,i)&=p_W(i) -|\mathfrak{X}| p_W(i) p_Y(i)\\
    p_{W,Y}(i,i) &= p_W(i) p_Y(i)\label{eq:pxyii}
\end{align}
}{
\begin{align}
    p_{W}(i)-|\mathfrak{X}| p_{W,Y}(i,i)&=p_W(i) -|\mathfrak{X}| p_W(i) p_Y(i)\\
    p_{W,Y}(i,i) &= p_W(i) p_Y(i)\label{eq:pxyii}
\end{align}
}

Plugging \eqref{eq:pxyii} into \eqref{eq:pxyiisub} yields
\iftoggle{singlecolumn}{
\begin{align}
    p_{W,Y}(i,j)- p_{W,Y}(i,i)&=p_W(i) p_Y(j)- p_W(i) p_Y(i)\\
    p_{W,Y}(i,j)&=p_W(i) p_Y(j)
\end{align}
}{
\begin{align}
    p_{W,Y}(i,j)- p_{W,Y}(i,i)&=p_W(i) p_Y(j)- p_W(i) p_Y(i)\\
    p_{W,Y}(i,j)&=p_W(i) p_Y(j)
\end{align}
}

\begin{sloppypar}
    Note that $i$ and $j$ are chosen arbitrarily. Therefore the condition given in \eqref{eq:lhsequivalence} implies that ${p_{Y|W}(y|w)=p_Y(y)}$,  ${\forall(w,y)\in\mathfrak{X}^2}$, concluding the proof of the \say{if} part. 
\end{sloppypar}

Hence, we have proved~\eqref{eq:phifinal}. Thus, there exists a deterministic bijective mapping $\Phi$ satisfying
$q_{0,\Phi}^\prime>q_{1,\Phi}^\prime$
and in turn $q_{0,\Phi}^{(r)}>q_{1,\Phi}^\prime$, $\forall r\in [n-1]$.

Now choose such a mapping $\Phi$ and note that for any $\gamma\in[0,1)$
\iftoggle{singlecolumn}{
\begin{align}
    q_{0,\Phi}^{(r)}-q_{1,\Phi}^\prime&= (1-\gamma^r) [q_0^\prime(\Phi)-q_1^\prime(\Phi)]\\
    &\ge (1-\gamma) [q_0^\prime(\Phi)-q_1^\prime(\Phi)],\: \forall r\in [n-1]\label{eq:reducetoiid}\\
    &> 0 ,\: \forall r\in [n-1]
\end{align}
}{
\begin{align}
    q_{0,\Phi}^{(r)}-q_{1,\Phi}^\prime&= (1-\gamma^r) [q_0^\prime(\Phi)-q_1^\prime(\Phi)]\\
    &\ge (1-\gamma) [q_0^\prime(\Phi)-q_1^\prime(\Phi)],\: \forall r\in [n-1]\label{eq:reducetoiid}\\
    &> 0 ,\: \forall r\in [n-1]
\end{align}
}
Next, define
\iftoggle{singlecolumn}{
\begin{align}
    q_{0,\Phi}^{\text{min}}&\triangleq (1-\gamma) q_{0,\Phi}^{\prime}+\gamma q_{1,\Phi}^{\prime}
\end{align}
}{
\begin{align}
    q_{0,\Phi}^{\text{min}}&\triangleq (1-\gamma) q_{0,\Phi}^{\prime}+\gamma q_{1,\Phi}^{\prime}
\end{align}
}
and choose a $\bar{\tau}\in\left(q_{1,\Phi}^{\prime},q_{0,\Phi}^{\text{min}}\right)$ bounded away from both ends of the interval.

Let $\hat{K}_n\triangleq n-\sum_{j=1}^n I_j$ and $L_j$ denote the $j$\textsuperscript{th} $0$ in $I^n$, $j=1,\dots,\hat{K}_n$. In other words, $L_j$ holds the index of the $j$\textsuperscript{th} retained column $C_j^{(2)}(\Phi)$ of $\tilde{\mathbf{G}}_{\Phi}^{(2)}$ in $\mathbf{G}^{(1)}$. Similarly, for $i$ with $I_i=0$, let $R_i\triangleq i- \sum_{l=1}^i I_l$ store the index of $C_i^{(1)}$ in $\tilde{\mathbf{G}}_{\Phi}^{(2)}$. 

\begin{sloppypar}
    Now note that when we have $I_i=1$, ${d_H(C_i^{(1)},C_j^{(2)}(\Phi))\sim\text{Binom}(\Lambda_n,q_{0,\Phi}^{(|i-L_j|)})}$ and when $I_i=0$, ${d_H(C_i^{(1)},C_{R_i}^{(2)}(\Phi))\sim\text{Binom}(\Lambda_n,q_{1,\Phi}^{\prime})}$. 
\end{sloppypar}

Next, we write the misdetection probability $P_{e,i}$ of $C^{(1)}_i$ as
\iftoggle{singlecolumn}{
\begin{align}
    P_{e,i} &= \Pr\left(\exists j\in[\hat{K}_n]: \Delta_{i,j}(\Phi) \le \Lambda_n \bar{\tau}, I_i=1\right)\notag\\
    &\hspace{6em}+\Pr\left(\forall j\in[\hat{K}_n]: \Delta_{i,j}(\Phi) > \Lambda_n \bar{\tau}, I_i=0\right)\\
    &\le \Pr\left(\exists j\in[\hat{K}_n]: \Delta_{i,j}(\Phi) \le \Lambda_n \bar{\tau}, I_i=1\right)+\Pr\left( \Delta_{i,R_i}(\Phi) > \Lambda_n \bar{\tau}, I_i=0\right)
\end{align}
}{
\begin{align}
    P_{e,i} &= \Pr\left(\exists j\in[\hat{K}_n]: \Delta_{i,j}(\Phi) \le \Lambda_n \bar{\tau}, I_i=1\right)\notag\\
    &\hspace{0.6em}+\Pr\left(\forall j\in[\hat{K}_n]: \Delta_{i,j}(\Phi) > \Lambda_n \bar{\tau}, I_i=0\right)\\
    &\le \Pr\left(\exists j\in[\hat{K}_n]: \Delta_{i,j}(\Phi) \le \Lambda_n \bar{\tau}, I_i=1\right)\notag\\
    &\hspace{0.6em}+\Pr\left( \Delta_{i,R_i}(\Phi) > \Lambda_n \bar{\tau}, I_i=0\right)
\end{align}
}
where
\iftoggle{singlecolumn}{
\begin{align}
    \Delta_{i,j}(\Phi) &\triangleq d_H(C_i^{(1)},C_j^{(2)}(\Phi)).
\end{align}
}{
\begin{align}
    \Delta_{i,j}(\Phi) &\triangleq d_H(C_i^{(1)},C_j^{(2)}(\Phi)).
\end{align}
}

From the union bound and Chernoff bound~\cite[Lemma 4.7.2]{ash2012information}, we obtain
\iftoggle{singlecolumn}{
\begin{align}
    P_{e,i}&\le \sum\limits_{j=1}^{\hat{K}_n} \Pr\left(\Delta_{i,j}(\Phi) \le \Lambda_n \bar{\tau}, I_i=1\right)+\Pr\left( \Delta_{i,R_i}(\Phi) > \Lambda_n \bar{\tau}, I_i=0\right)\\
    &\le \sum\limits_{j=1}^{\hat{K}_n}  2^{-\Lambda_n D(\bar{\tau}\|q_{0,\Phi}^{(|i-L_j|)})}+  2^{-\Lambda_n D(1-\bar{\tau}\|1-q_{1,\Phi}^{\prime})}
\end{align}
}{
\begin{align}
    P_{e,i}&\le \sum\limits_{j=1}^{\hat{K}_n} \Pr\left(\Delta_{i,j}(\Phi) \le \Lambda_n \bar{\tau}, I_i=1\right)\notag\\
    &\hspace{3em}+\Pr\left( \Delta_{i,R_i}(\Phi) > \Lambda_n \bar{\tau}, I_i=0\right)\\
    &\le \sum\limits_{j=1}^{\hat{K}_n}  2^{-\Lambda_n D(\bar{\tau}\|q_{0,\Phi}^{(|i-L_j|)})}+  2^{-\Lambda_n D(1-\bar{\tau}\|1-q_{1,\Phi}^{\prime})}
\end{align}
}
It is straightforward to show that $D(\bar{\tau}\|p)$ is an increasing function of $p$ for $p>\bar{\tau}$. Thus $\forall i\in[n], j\in[\hat{K}_n]$, we have
\iftoggle{singlecolumn}{
\begin{align}
    q_{0,\Phi}^{(|i-L_j|)} &\ge q_{0,\Phi}^{\prime}\\
    D(\bar{\tau}\|q_{0,\Phi}^{(|i-L_j|)}) &\ge D(\bar{\tau}\|q_{0,\Phi}^{\text{min}})\\
    2^{-\Lambda_n D(\bar{\tau}\|q_{0,\Phi}^{(|i-L_j|)})} &\le 2^{-\Lambda_n D(\bar{\tau}\|q_{0,\Phi}^{\text{min}})}
\end{align}
}{
\begin{align}
    q_{0,\Phi}^{(|i-L_j|)} &\ge q_{0,\Phi}^{\prime}\\
    D(\bar{\tau}\|q_{0,\Phi}^{(|i-L_j|)}) &\ge D(\bar{\tau}\|q_{0,\Phi}^{\text{min}})\\
    2^{-\Lambda_n D(\bar{\tau}\|q_{0,\Phi}^{(|i-L_j|)})} &\le 2^{-\Lambda_n D(\bar{\tau}\|q_{0,\Phi}^{\text{min}})}
\end{align}
}
Thus, we have
\iftoggle{singlecolumn}{
\begin{align}
    P_{e,i} &\le \sum\limits_{j=1}^{\hat{K}_n} 2^{-\Lambda_n D(\bar{\tau}\|q_{0,\Phi}^{(|i-L_j|)})} + 2^{-\Lambda_n D(1-\bar{\tau}\|1-q_{1,\Phi}^{\prime})}\\
    &\le \sum\limits_{j=1}^{\hat{K}_n} 2^{-\Lambda_n D(\tau\|q_{0,\Phi}^{\text{min}})} +  2^{-\Lambda_n D(1-\tau\|1-q_{1,\Phi}^{\prime})}\\
    &=  \hat{K}_n 2^{-\Lambda_n D(\tau\|q_{0,\Phi}^{\text{min}})} +  2^{-\Lambda_n D(1-\tau\|1-q_{1,\Phi}^{\prime})}
\end{align}
}{
\begin{align}
    P_{e,i} &\le \sum\limits_{j=1}^{\hat{K}_n} 2^{-\Lambda_n D(\bar{\tau}\|q_{0,\Phi}^{(|i-L_j|)})} + 2^{-\Lambda_n D(1-\bar{\tau}\|1-q_{1,\Phi}^{\prime})}\\
    &\le \sum\limits_{j=1}^{\hat{K}_n} 2^{-\Lambda_n D(\tau\|q_{0,\Phi}^{\text{min}})} +  2^{-\Lambda_n D(1-\tau\|1-q_{1,\Phi}^{\prime})}\\
    &=  \hat{K}_n 2^{-\Lambda_n D(\tau\|q_{0,\Phi}^{\text{min}})} +  2^{-\Lambda_n D(1-\tau\|1-q_{1,\Phi}^{\prime})}
\end{align}
}
Thus, by simple union bound the total misdetection probability $P_{e,total}$ can be bounded as
\iftoggle{singlecolumn}{
\begin{align}
    P_{e,total} &\le \sum\limits_{i=1}^n P_{e,i}\\
    &\le \sum\limits_{i=1}^n \hat{K}_n 2^{-\Lambda_n D(\bar{\tau}\|q_{0,\Phi}^{\text{min}})} +  2^{-\Lambda_n D(1-\bar{\tau}\|1-q_{1,\Phi}^{\prime})}\\
    &= n \hat{K}_n 2^{-\Lambda_n D(\bar{\tau}\|q_{0,\Phi}^{\text{min}})} + n  2^{-\Lambda_n D(1-\bar{\tau}\|1-q_{1,\Phi}^{\prime})}\\
    &\le  n^2 2^{-\Lambda_n D(\bar{\tau}\|q_{0,\Phi}^{\text{min}})} + n  2^{-\Lambda_n D(1-\bar{\tau}\|1-q_{1,\Phi}^{\prime})}
\end{align}
}{
\begin{align}
    P_{e,total} &\le \sum\limits_{i=1}^n P_{e,i}\\
    &\le \sum\limits_{i=1}^n \hat{K}_n 2^{-\Lambda_n D(\bar{\tau}\|q_{0,\Phi}^{\text{min}})} +  2^{-\Lambda_n D(1-\bar{\tau}\|1-q_{1,\Phi}^{\prime})}\\
    &= n \hat{K}_n 2^{-\Lambda_n D(\bar{\tau}\|q_{0,\Phi}^{\text{min}})} + n  2^{-\Lambda_n D(1-\bar{\tau}\|1-q_{1,\Phi}^{\prime})}\\
    &\le  n^2 2^{-\Lambda_n D(\bar{\tau}\|q_{0,\Phi}^{\text{min}})} + n  2^{-\Lambda_n D(1-\bar{\tau}\|1-q_{1,\Phi}^{\prime})}
\end{align}
}
Hence, $P_{e,total}\to 0$ as $n\to\infty$ if the seed size $\Lambda_n$ satisfies 
\iftoggle{singlecolumn}{
\begin{align}
    \Lambda_n D(\bar{\tau}\|q_{0,\Phi}^{\text{min}})- 2\log n &>0\\
    \Lambda_n D(1-\bar{\tau}\|1-q_{1,\Phi}^{\prime}) - \log n&>0
\end{align}
}{
\begin{align}
    \Lambda_n D(\bar{\tau}\|q_{0,\Phi}^{\text{min}})- 2\log n &>0\\
    \Lambda_n D(1-\bar{\tau}\|1-q_{1,\Phi}^{\prime}) - \log n&>0
\end{align}
}
Thus any seed size $\Lambda_n$ satisfying 
\iftoggle{singlecolumn}{
\begin{align}
    \Lambda_n &>\frac{\log n}{\min \left\{\frac{1}{2}D(\bar{\tau}\|q_{0,\Phi}^{\text{min}}), D(1-\bar{\tau}\|1-q_{1,\Phi}^{\prime})\right\}}
\end{align}
}{
\begin{align}
    \Lambda_n &>\frac{\log n}{\min \left\{\frac{1}{2}D(\bar{\tau}\|q_{0,\Phi}^{\text{min}}), D(1-\bar{\tau}\|1-q_{1,\Phi}^{\prime})\right\}}
\end{align}
}
is sufficient to drive $P_{e,total}$ to 0. Thus a seed size $\Lambda_n=\Omega(\log n)$ is enough for successful deletion detection.
\qed

\section{Proof of Lemma~\ref{lem:histogram}}
\label{proof:histogram}

First, observe that from~\cite[Theorem 3]{burke1958markovian} and the specific Markov structure given in Definition~\ref{defn:markovdb}, the rows of the collapsed database $\tilde{\mathbf{X}}$ become \emph{i.i.d.} first-order stationary binary Markov chains, with the following probability transition matrix and stationary distribution:
\iftoggle{singlecolumn}{
\begin{align}
    \tilde{\mathbf{P}}&=\begin{bmatrix}\gamma+(1-\gamma) u_1 & (1-\gamma)(1-u_1)\\
    (1-\gamma)u_1 & 1-(1-\gamma)u_1
    \end{bmatrix}\\
    \tilde{\pi}&=\begin{bmatrix}
    u_1 & 1-u_1
    \end{bmatrix}
\end{align}
}{
\begin{align}
    \tilde{\mathbf{P}}&=\begin{bmatrix}\gamma+(1-\gamma) u_1 & (1-\gamma)(1-u_1)\\
    (1-\gamma)u_1 & 1-(1-\gamma)u_1
    \end{bmatrix}\\
    \tilde{\pi}&=\begin{bmatrix}
    u_1 & 1-u_1
    \end{bmatrix}
\end{align}
}

For brevity, we let ${\mu_n\triangleq \Pr(\exists i,j\in [n],\: i\neq j,\tilde{{H}}^{(1)}_{i}=\tilde{{H}}^{(1)}_j)}$. Next, from the union bound, we obtain
\iftoggle{singlecolumn}{
\begin{align}
    \mu_n&\le \sum\limits_{{(i,j)\in[n]^2:i<j}} \Pr(\tilde{{H}}^{(1)}_{i}=\tilde{{H}}^{(1)}_{j})\\
    &\le n^2 \max\limits_{{(i,j)\in[n]^2:i<j}} \Pr(\tilde{{H}}^{(1)}_{i}=\tilde{{H}}^{(1)}_{j})
\end{align}
}{
\begin{align}
    \mu_n&\le \sum\limits_{{(i,j)\in[n]^2:i<j}} \Pr(\tilde{{H}}^{(1)}_{i}=\tilde{{H}}^{(1)}_{j})\\
    &\le n^2 \max\limits_{{(i,j)\in[n]^2:i<j}} \Pr(\tilde{{H}}^{(1)}_{i}=\tilde{{H}}^{(1)}_{j})
\end{align}
}
Due to stationarity of $\tilde{\mathbf{P}}$, this maximum is equal to $\Pr(\tilde{{H}}^{(1)}_1=\tilde{{H}}^{(1)}_{s+1})$ for some $s$. For brevity, let $\mathbf{Q}\triangleq\tilde{\mathbf{P}}^s$ and $q\triangleq \Pr(\tilde{{H}}^{(1)}_1=\tilde{{H}}^{(1)}_{s+1})$. Observe that 
$\tilde{{H}}^{(1)}_1$ and $\tilde{{H}}^{(1)}_{s+1}$ are correlated Binom($m_n,1-u_1$) random variables and for any $s$, $\mathbf{Q}$ has positive values, \emph{i.e.,} the collapsed Markov chain is irreducible for any $s$. Now, we have
\iftoggle{singlecolumn}{
\vspace{-1em}
\begin{adjustwidth}{-0.25cm}{0pt}
\begin{align}
     q&= \sum\limits_{r=0}^{m_n} \Pr(\tilde{{H}}^{(1)}_1=r) \Pr(\tilde{{H}}^{(1)}_{s+1}=r|\tilde{{H}}^{(1)}_1=r)\\
    &= \sum\limits_{r=0}^{m_n} \binom{m}{r}(1-u_1)^r u_1^{m_n-r} \Pr(\tilde{{H}}^{(1)}_{s+1}=r|\tilde{{H}}^{(1)}_1=r)
\end{align}
\end{adjustwidth}
}{
\vspace{-1em}
\begin{adjustwidth}{-0.25cm}{0pt}
\begin{align}
     q&= \sum\limits_{r=0}^{m_n} \Pr(\tilde{{H}}^{(1)}_1=r) \Pr(\tilde{{H}}^{(1)}_{s+1}=r|\tilde{{H}}^{(1)}_1=r)\\
    &= \sum\limits_{r=0}^{m_n} \binom{m}{r}(1-u_1)^r u_1^{m_n-r} \Pr(\tilde{{H}}^{(1)}_{s+1}=r|\tilde{{H}}^{(1)}_1=r)
\end{align}
\end{adjustwidth}
}
Note that since the rows of $\tilde{\mathbf{X}}$ are \emph{i.i.d.}, we have 
\iftoggle{singlecolumn}{
\begin{align}
    \Pr(\tilde{{H}}^{(1)}_{s+1}=r|\tilde{{H}}^{(1)}_1=r) = \Pr(M+N=r)
\end{align}
}{
\begin{align}
    \Pr(\tilde{{H}}^{(1)}_{s+1}=r|\tilde{{H}}^{(1)}_1=r) = \Pr(M+N=r)
\end{align}
}
where $M\sim\text{Binom}(r,Q_{2,2})$ and $N\sim\text{Binom}(m_n-r,Q_{1,2})$ are independent. Note that there are two ways leading to state 2 in the collapsed column after $s$ steps. The first one is the state $2$ staying in the same state after $s$ steps, and the second one is state $1$ being converted to state $2$ after $s$ steps. Here the Binomial random variables $E$ and $F$ keep counts of the former and the latter ways, respectively.

Then, from Stirling's approximation~\cite[Chapter 3.2]{cormen2022introduction} on the factorial terms in the Binomial coefficient and~\cite[Theorem 11.1.2]{cover2006elements}, we get
\iftoggle{singlecolumn}{
\begin{align}
    q
    &= \sum\limits_{r=0}^{m_n} \binom{{m_n}}{r}(1-u_1)^r u_1^{{m_n}-r} \Pr(M+N=r)\\
    &\le \frac{e}{\sqrt{2\pi}} {m_n}^{-1/2} \sum\limits_{r=0}^{m_n} \Pi_r^{-1} 2^{-{m_n} D(\frac{r}{{m_n}}\|(1-u_1))}\Pr(M+N=r) 
\end{align}
}{
\begin{align}
    q
    &= \sum\limits_{r=0}^{m_n} \binom{{m_n}}{r}(1-u_1)^r u_1^{{m_n}-r} \Pr(M+N=r)\\
    &\le \frac{e}{\sqrt{2\pi}} {m_n}^{-1/2} \sum\limits_{r=0}^{m_n} \Pi_r^{-1} 2^{-{m_n} D(\frac{r}{{m_n}}\|(1-u_1))}\notag\\&\hspace{11em}\Pr(M+N=r) 
\end{align}
}
where $\Pi_r=\frac{r}{{m_n}}(1-\frac{r}{{m_n}})$. Let
\iftoggle{singlecolumn}{
\begin{align}
    T &= \sum\limits_{r=0}^{m_n} \Pi_r^{-1} 2^{-{m_n} D(\frac{r}{{m_n}}\|(1-u_1))}\Pr(M+N=r)\\ &= T_1+T_2
\end{align}
}{
\begin{align}
    T &= \sum\limits_{r=0}^{m_n} \Pi_r^{-1} 2^{-{m_n} D(\frac{r}{{m_n}}\|(1-u_1))}\Pr(M+N=r)\\ &= T_1+T_2
\end{align}
}
where
\iftoggle{singlecolumn}{
\begin{align}
    T_1 &= \sum_{\mathclap{\hspace{4em} r:D(\frac{r}{{m_n}}\|1-u_1)> \frac{\epsilon_n^2}{2\log_e 2}}} \hspace{2em}\Pi_r^{-1} 2^{-{m_n} D(\frac{r}{{m_n}}\|(1-u_1))}\Pr(M+N=r)\label{eq:T1}\\
    T_2 &= \sum_{\mathclap{\hspace{4em} r:D(\frac{r}{{m_n}}\|1-u_1)\le \frac{\epsilon_n^2}{2\log_e 2}}} \hspace{2em} \Pi_r^{-1} 2^{-{m_n} D(\frac{r}{{m_n}}\|(1-u_1))}\Pr(M+N=r),\label{eq:T2}
\end{align}
}{
\begin{align}
    T_1 &= \sum_{\mathclap{\hspace{4em} r:D(\frac{r}{{m_n}}\|1-u_1)> \frac{\epsilon_n^2}{2\log_e 2}}} \hspace{2em}\Pi_r^{-1} 2^{-{m_n} D(\frac{r}{{m_n}}\|(1-u_1))}\Pr(M+N=r)\label{eq:T1}\\
    T_2 &= \sum_{\mathclap{\hspace{4em} r:D(\frac{r}{{m_n}}\|1-u_1)\le \frac{\epsilon_n^2}{2\log_e 2}}} \hspace{2em} \Pi_r^{-1} 2^{-{m_n} D(\frac{r}{{m_n}}\|(1-u_1))}\Pr(M+N=r),\label{eq:T2}
\end{align}
}
$\epsilon_n>0$, which is described below in more detail, is such that $\epsilon_n\to0$ as $n\to\infty$.

First, we look at $T_1$. Note that for any $r\in\mathbb{N}$, we have ${\Pi_r\le {m_n}^2}$, suggesting the multiplicative term in the summation in~\eqref{eq:T1} is polynomial with ${m_n}$. Note that we can simply separate the cases $r=0$, $r={m_n}$ whose probabilities vanish exponentially in ${m_n}$. Therefore, as long as ${{m_n} \epsilon_n^2\to\infty}$, $T_1$ has a polynomial number of elements that decay exponentially with ${m_n}$. Thus
\iftoggle{singlecolumn}{
\begin{align}
    T_1\to0\text{ as }n\to\infty\label{eq:t1}
\end{align}
}{
\begin{align}
    T_1\to0\text{ as }n\to\infty\label{eq:t1}
\end{align}
}
as long as ${{m_n} \epsilon_n^2\to\infty}$.

Now, we focus on $T_2$. From Pinsker's inequality~\cite[Lemma 11.6.1]{cover2006elements}, we have
\iftoggle{singlecolumn}{
\begin{align}
    D\left(\frac{r}{{m_n}}\Big\|1-u_1\right)\le \frac{\epsilon_n^2}{2\log_e 2}\Rightarrow \mathbb{V}\left(\frac{r}{{m_n}}, 1-u_1\right)\le \epsilon_n
\end{align}
}{
\begin{align}
    D\Big(\frac{r}{{m_n}}\Big\|1-u_1\Big)&\le \frac{\epsilon_n^2}{2\log_e 2}\notag\\&\Longrightarrow \mathbb{V}\left(\frac{r}{{m_n}}, 1-u_1\right)\le \epsilon_n
\end{align}
}
where $\mathbb{V}$ denotes the (unnormalized) total variation distance between the Bernoulli distributions with given parameters. Therefore
\iftoggle{singlecolumn}{
\begin{align}
    \Big|\{r:D\Big(\frac{r}{{m_n}}\Big\|1-u_1\Big)\le \frac{\epsilon_n^2}{2\log_e 2}\}\Big|
    &\le \Big|\{r:\mathbb{V}\Big(\frac{r}{{m_n}}, 1-u_1\Big)\le \epsilon_n\}\Big| \\
    &= O({m_n}\epsilon_n)
\end{align}
}{
\begin{align}
    \Big|\{r:D\Big(\frac{r}{{m_n}}&\Big\|1-u_1\Big)\le \frac{\epsilon_n^2}{2\log_e 2}\}\Big|\notag\\
    &\le \Big|\{r:\mathbb{V}\Big(\frac{r}{{m_n}}, 1-u_1\Big)\le \epsilon_n\}\Big| \\
    &= O({m_n}\epsilon_n)
\end{align}
}
for small $\epsilon_n$. Furthermore, if $\mathbb{V}\left(\frac{r}{{m_n}}, 1-u_1\right)\le \epsilon_n$, we have 
\iftoggle{singlecolumn}{
\begin{align}
    \Pi_{r}^{-1} &\le \frac{1}{(1-u_1) u_1}
\end{align}
}{
\begin{align}
    \Pi_{r}^{-1} &\le \frac{1}{(1-u_1) u_1}
\end{align}
}
Now, we investigate $\Pr(M+N=r)$ for the values of $r$ in the interval $m_n(1-u_1)\pm m_n\epsilon_n$.
\iftoggle{singlecolumn}{
\begin{align}
    \Pr(M+N=r)&=\sum\limits_{i=1}^r \Pr(M=r-i) \Pr(N=i)+ \Pr(M=r)\Pr(N=0)\\
    &= \sum\limits_{i=1}^r \binom{r}{i} Q_{2,2}^{r-i} (1-Q_{2,2})^i \binom{{m_n}-r}{i} Q_{1,2}^i (1-Q_{1,2})^{{m_n}-r-i}\notag\\&\hspace{2em}+Q_{2,2}^r Q_{1,1}^{{m_n}-r}\label{eq:binom2}
\end{align}
}{
\begin{align}
    \Pr(M+N=r)&=\sum\limits_{i=1}^r \Pr(M=r-i) \Pr(N=i)\notag\\&\hspace{2em}+ \Pr(M=r)\Pr(N=0)\\
    &=Q_{2,2}^r Q_{1,1}^{{m_n}-r}\notag\\&\hspace{0.5em}+ \sum\limits_{i=1}^r \binom{r}{i} Q_{2,2}^{r-i} (1-Q_{2,2})^i \notag\\&\hspace{2em}\binom{{m_n}-r}{i} Q_{1,2}^i (1-Q_{1,2})^{{m_n}-r-i}\label{eq:binom2}
\end{align}
}

Again, from Stirling's approximation~\cite[Chapter 3.2]{cormen2022introduction} on the factorial terms in the Binomial coefficient in \eqref{eq:binom2} and from~\cite[Theorem 11.1.2]{cover2006elements}, we have
\iftoggle{singlecolumn}{
\begin{align}
    \Pr(M+N=r)&\le Q_{2,2}^r Q_{1,1}^{{m_n}-r} + \frac{e^2}{2\pi}[r({m_n}-r)]^{-1/2} U
\end{align}
}{
\begin{align}
    \Pr(M+N=r)&\le Q_{2,2}^r Q_{1,1}^{{m_n}-r} + \frac{e^2}{2\pi}[r({m_n}-r)]^{-1/2} U
\end{align}
}
where 
\iftoggle{singlecolumn}{
\begin{align}
    U = \sum\limits_{i=1}^r \Pi^{-1}_{i/r} \Pi^{-1}_{i/{m_n}-r} 2^{-r D(1-\frac{i}{r}\|Q_{2,2}) -({m_n}-r) D(\frac{i}{{m_n}-r}\|Q_{1,2})}
\end{align}
}{
\begin{align}
    U = \sum\limits_{i=1}^r \Pi^{-1}_{i/r} \Pi^{-1}_{i/{m_n}-r} 2^{-r D(1-\frac{i}{r}\|Q_{2,2}) -({m_n}-r) D(\frac{i}{{m_n}-r}\|Q_{1,2})}
\end{align}
}
Then, from $r\in[{{m_n}(1-u_1-\epsilon_n)},{{m_n}(1-u_1+\epsilon_n)}]$ we obtain
\iftoggle{singlecolumn}{
\begin{align}
    \Pr(M+N=r)&\le Q_{2,2}^r Q_{1,1}^{{m_n}-r} + \frac{e^2}{2\pi}\frac{{m_n}^{-1}}{\sqrt{(1-u_1-\epsilon_n)(u_1-\epsilon_n)}}  U
\end{align}
}{
\begin{align}
    \Pr(M+N=r)&\le Q_{2,2}^r Q_{1,1}^{{m_n}-r} \notag\\&\hspace{1em}+ \frac{e^2}{2\pi}\frac{{m_n}^{-1}}{\sqrt{(1-u_1-\epsilon_n)(u_1-\epsilon_n)}}  U
\end{align}
}
and
\iftoggle{singlecolumn}{
\begin{align}
    U&\le \sum\limits_{i=1}^r \Pi^{-1}_{i/r} \Pi^{-1}_{i/{m_n}-r} 2^{-{m_n} (1-u_1-\epsilon_n) D(1-\frac{i}{r}\|Q_{2,2})}2^{-{m_n} (u_1-\epsilon_n) D(\frac{i}{{m_n}-r}\|Q_{1,2})}\\
    &= \sum\limits_{i\notin \mathcal{R}(\epsilon_n)} \Pi^{-1}_{i/r} \Pi^{-1}_{i/{m_n}-r} 2^{-{m_n} (1-u_1-\epsilon_n) D(1-\frac{i}{r}\|Q_{2,2})} 2^{-{m_n} (u_1-\epsilon_n) D(\frac{i}{{m_n}-r}\|Q_{1,2})}\notag\\
    &\hspace{2em}+ \sum\limits_{i\in \mathcal{R}(\epsilon_n)} \Pi^{-1}_{i/r} \Pi^{-1}_{i/{m_n}-r} 2^{-{m_n} (1-u_1-\epsilon_n) D(1-\frac{i}{r}\|Q_{2,2})} 2^{-{m_n} (u_1-\epsilon_n) D(\frac{i}{{m_n}-r}\|Q_{1,2})}\label{eq:U}
\end{align}
}{
\begin{align}
    U&\le \sum\limits_{i=1}^r \Pi^{-1}_{i/r} \Pi^{-1}_{i/{m_n}-r} 2^{-{m_n} (1-u_1-\epsilon_n) D(1-\frac{i}{r}\|Q_{2,2})}\notag\\&\hspace{8em}2^{-{m_n} (u_1-\epsilon_n) D(\frac{i}{{m_n}-r}\|Q_{1,2})}\\
    &= \sum\limits_{i\notin \mathcal{R}(\epsilon_n)} \Pi^{-1}_{i/r} \Pi^{-1}_{i/{m_n}-r} 2^{-{m_n} (1-u_1-\epsilon_n) D(1-\frac{i}{r}\|Q_{2,2})} \notag\\&\hspace{8em}2^{-{m_n} (u_1-\epsilon_n) D(\frac{i}{{m_n}-r}\|Q_{1,2})}\notag\\
    &\hspace{1em}+ \sum\limits_{i\in \mathcal{R}(\epsilon_n)} \Pi^{-1}_{i/r} \Pi^{-1}_{i/{m_n}-r} 2^{-{m_n} (1-u_1-\epsilon_n) D(1-\frac{i}{r}\|Q_{2,2})} \notag\\&\hspace{8em}2^{-{m_n} (u_1-\epsilon_n) D(\frac{i}{{m_n}-r}\|Q_{1,2})}\label{eq:U}
\end{align}
}
where we define the set $\mathcal{R}(\epsilon_n)$ as
\iftoggle{singlecolumn}{
\begin{align}
    \mathcal{R}(\epsilon_n) &\triangleq\Big\{i\in[r]:D\Big(1-\frac{i}{r}\Big\|Q_{2,2}\Big),D\Big(\frac{i}{{m_n}-r}\Big\|Q_{1,2}\Big)\le \frac{\epsilon_n^2}{2\log_e 2}\Big\}
\end{align}
}{
\begin{align}
    \mathcal{R}(\epsilon_n) &\triangleq\Big\{i\in[r]:D\Big(1-\frac{i}{r}\Big\|Q_{2,2}\Big)\le \frac{\epsilon_n^2}{2\log_e 2},\notag\\
    &\hspace{3em}D\Big(\frac{i}{{m_n}-r}\Big\|Q_{1,2}\Big)\le \frac{\epsilon_n^2}{2\log_e 2}\Big\}
\end{align}
}
Note that similar to $T_1$, the first summation in \eqref{eq:U} vanishes exponentially in ${m_n}$ whenever ${m_n}\epsilon_n^2\to\infty$, and using Pinsker's inequality once more, the second term can be upper bounded by
\iftoggle{singlecolumn}{
\begin{align}
    O(|\mathcal{R}(\epsilon_n)|)=O({m_n}\epsilon_n)
\end{align}
}{
\begin{align}
    O(|\mathcal{R}(\epsilon_n)|)=O({m_n}\epsilon_n)
\end{align}
}
Now, we choose ${\epsilon_n={m_n}^{-\frac{1}{2}} V_n}$ for some $V_n$ satisfying ${V_n=\omega(1)}$ and ${V_n=o(m_n^{1/2})}$. Thus, $T_1$ vanishes exponentially fast since ${m_n\epsilon_n^2=V_n^2\to\infty}$ and 
\iftoggle{singlecolumn}{
\begin{gather}
\Pr(M+N=r) = O(\epsilon_n)\\
    T  = O({m_n} \epsilon_n^2)=O(V_n^2)\\
    \mu_n = O(n^2 {m_n}^{-1/2} V_n^2)
\end{gather}
}{
\begin{gather}
\Pr(M+N=r) = O(\epsilon_n)\\
    T  = O({m_n} \epsilon_n^2)=O(V_n^2)\\
    \mu_n = O(n^2 {m_n}^{-1/2} V_n^2)
\end{gather}
}
By the assumption ${m_n=\omega(n^4)}$, we have ${m_n=n^4 Z_n}$ for some $Z_n$ satisfying  ${\lim\limits_{n\to\infty} Z_n=\infty}$. Now, taking ${V_n=o(Z_n^{1/4})}$ (e.g. ${V_n=Z_n^{1/6}}$), we get
\iftoggle{singlecolumn}{
\begin{align}
    \mu_n&\le O( Z_n^{-1/2} V_n^2)
    = o(1)
\end{align}
}{
\begin{align}
    \mu_n&\le O( Z_n^{-1/2} V_n^2)
    = o(1)
\end{align}
}
Thus $m_n=\omega(n^4)$ is sufficient to have $\mu_n\to0$ as $n\to\infty$, concluding the proof. \qed


\section{Proof of Achievability of Theorem~\ref{thm:mainresultW1}}
\label{proof:achievabilityW1}

The proof of the achievability part follows from successive union bounds exploiting the following: 
\begin{itemize}
    \item For any typical row $Y^{K_n}$ of $\mathbf{Y}$, there are approximately $2^{{K_n} H(X|Y)}$ jointly typical sequences with respect to $p_{X,Y}$.
    \item If the output of the synchronization channel has length ${K_n}$ then there are at least ${k_{\min}=\left\lceil \frac{{K_n}}{s_{\max}}\right\rceil}$ retained (not deleted) elements.
    \item For the number of columns $n$, the number of deletion patterns with $k_{\min}$ retained elements is 
    \iftoggle{singlecolumn}{
    \begin{align}
        \binom{n}{k_{\min}}&\le 2^{n H_b(k_{\min}/n)}
    \end{align}
    }{
    \begin{align}
        \binom{n}{k_{\min}}&\le 2^{n H_b(k_{\min}/n)}
    \end{align}
    }
    \item Any stretched row has the same probability as the original row.
    \item If the original length-$n$ sequence and the retained length-$k_{\min}$ sequence after the deletion channel are $\epsilon$-typical with respect to $p_X$, then the complementary length-$(n-k_{\min})$ subsequence is $\tilde{\epsilon}$-typical with respect to $p_X$, where $\tilde{\epsilon}=\frac{n+k_{\min}}{n-k_{\min}}$.
    \item The cardinality of the set of $\tilde{\epsilon}$-typical sequences of length $n-k_{\min}$ with respect to $p_X$ is approximately $2^{(n-k_{\min}) H(X)}$.
\end{itemize}
We need to show that for a given pair of matching rows, WLOG, $X^n_1$ of $\mathbf{X}$ and $Y^{K_n}_t$ of $\mathbf{Y}$ with $\sigma_n(1)=t$, the probability of error ${P_e\triangleq\Pr(\hat{\sigma}_n (1)\neq t)}$ of the following matching scheme can be made arbitrarily small asymptotically where ${K_n}=\sum_{j=1}^n S_{1,j}$ is the random variable corresponding to the length of $Y^{K_n}_t$. The matching scheme we propose follows these steps:
\begin{enumerate}[label=\textbf{\arabic*)},leftmargin=1.3\parindent]
\item For all $j\in[n]$, discard the $j$\textsuperscript{th} column of $\mathbf{X}$ if $A_j=1$ to obtain $\bar{\mathbf{X}}$ whose column size is $n-A$ where $A=\sum_{j=1}^n A_j$.
\item Stretch each row $\bar{X}^{n-A}_i=\bar{X}_{i,1},\dots,\bar{X}_{i,n-A}$ of $\bar{\mathbf{X}}$ into $\tilde{X}_{i}^{(n-A) s_{\max}}$, by repeating each element of $\bar{X}^{n-A}_i$ $s_{\max}$ times as follows
\iftoggle{singlecolumn}{
\begin{align}
    \tilde{X}_{i}^{(n-A) s_{\max}}&= 1^{s_{\max}}\otimes \bar{X}_{i,1},\dots,1^{s_{\max}}\otimes \bar{X}_{i,n-A}
\end{align}
}{
\begin{align}
    \tilde{X}_{i}^{(n-A) s_{\max}}&= 1^{s_{\max}}\otimes \bar{X}_{i,1},\dots,1^{s_{\max}}\otimes \bar{X}_{i,n-A}
\end{align}
}
where $1^{s_{\max}}$ is an all-one row vector of length $s_{\max}$ and $\otimes$ denotes the Kronecker product.
\item Fix $\epsilon>0$. If $K_n<k\triangleq n(\mathbb{E}[S]-\epsilon)$ declare error, whose probability is denoted by $\kappa_n$ where $k$ is assumed to be an integer for computational simplicity. Otherwise, proceed with the next step.
\item If $A<a=n(\alpha\delta-\epsilon)$ declare error, whose probability is denoted by $\mu_n$. Otherwise, proceed with the next step.
\item Match the $t$\textsuperscript{th} row $Y^{K_n}_t$ of $\mathbf{Y}$ $X^n_{1}$ of $\mathbf{X}$, assigning $\hat{\sigma}_n(1)=t$, if $i=1$ is the only index in $[m_n]$ such that \emph{i)} $\bar{X}^{n-A}_i$ is  $\epsilon$-typical and \emph{ii)} $\tilde{X}_i^{(n-A) s_{\max}}$ contains a subsequence jointly $\epsilon$-typical with $Y^{K_n}_t$ with respect to $p_{X,Y}$. Otherwise, declare a \emph{collision} error.
\end{enumerate}

Since additional columns in $\mathbf{Y}$ and additional detected deleted columns in $\mathbf{X}$ would decrease the collision probability, we have
\iftoggle{singlecolumn}{
\begin{align}
    \Pr(\text{collision between 1 and }i|K_n\ge k,A\ge a)\le \Pr(\text{collision between 1 and }i||K_n=k,A=a)
\end{align}
}{
\begin{align}
    \Pr&(\text{collision between 1 and }i|K_n\ge k,A\ge a)\notag\\&\le \Pr(\text{collision between 1 and }i||K_n=k,A=a)
\end{align}
}
for any $i\in[m_n]\setminus\{1\}$. Thus, we can focus on the case $K_n=k$, $A=a$, as it yields an upper bound on the error probability of our matching scheme.

Let $A_\epsilon^{(n-a)}(X)$ denote the set of $\epsilon$-typical (with respect to $p_X$) sequences of length $n-a$ and $A_\epsilon(X^k|Y^k_t)$ denote the set of sequences of length $k$ jointly $\epsilon$-typical (with respect to $p_{X,Y}$) with $Y^k_t$. For the matching rows $X^n_1$, $Y^k_t$ of $\mathbf{X}$ and $\mathbf{Y}$, define the pairwise collision probability between $X^n_1$ and $X^n_i$ for any $i\in[m_n]\setminus\{1\}$ as
\iftoggle{singlecolumn}{
\begin{align}
    P_{\text{col,i}}\triangleq \Pr(\exists {z}^k:  {z}^k\in A_\epsilon({X}^k|Y^k_t) \text{ and }{z}^k \text{ is a subsequence of } \tilde{X}_{i}^{(n-a) s_{\max}}.).
\end{align}
}{
\begin{align}
    P_{\text{col,i}}&\triangleq \Pr(\exists {z}^k:  {z}^k\in A_\epsilon({X}^k|Y^k_t) \text{ and }{z}^k \text{ is a} \notag\\&\hspace{4em} \text{subsequence of } \tilde{X}_{i}^{(n-a) s_{\max}}.).
\end{align}
}

Therefore given the correct labeling for $Y^k_t\in\mathbf{Y}$ is $X^n_1\in\mathbf{X}$, the probability of error $P_e$ can be bounded as
\iftoggle{singlecolumn}{
\begin{align}
    P_e 
    &\le \Pr(\nexists {z}^k:  {z}^k\in A_\epsilon({X}^k|Y^k_t) \text{ and }{z}^k \text{ is a subsequence of } \tilde{X}_{1}^{(n-a) s_{\max}}.)\notag\\
    &\qquad +\Pr(X^n_1\notin A_\epsilon^{(n)}(X))+\sum\limits_{i=2}^{2^{n R}} P_{\text{col,i}}+ \kappa_n+\mu_n\\
    &\le 2\epsilon+\sum\limits_{i=2}^{2^{n R}} P_{\text{col,i}}+\kappa_n+\mu_n\\
    &\le 2\epsilon+ 2^{n R} P_{\text{col,2}}+\kappa_n+\mu_n \label{eq:Perowwise}
\end{align}
}{
\begin{align}
    P_e 
    &\le \Pr(\nexists {z}^k:  {z}^k\in A_\epsilon({X}^k|Y^k_t) \text{ and }{z}^k \text{ is a} \notag\\&\hspace{4em}\text{subsequence of } \tilde{X}_{1}^{(n-a) s_{\max}}.)\notag\\
    &\hspace{2em} +\Pr(X^n_1\notin A_\epsilon^{(n)}(X))\notag\\
    &\hspace{2em}+\sum\limits_{i=2}^{2^{n R}} P_{\text{col,i}}+ \kappa_n+\mu_n\\
    &\le 2\epsilon+\sum\limits_{i=2}^{2^{n R}} P_{\text{col,i}}+\kappa_n+\mu_n\\
    &\le 2\epsilon+ 2^{n R} P_{\text{col,2}}+\kappa_n+\mu_n \label{eq:Perowwise}
\end{align}
}
where \eqref{eq:Perowwise} follows from the fact the the rows are \emph{i.i.d.} and thus $P_{\text{col,i}}=P_{\text{col,2}}$, $\forall i\in[m_n]\setminus\{1\}$.

We now upper bound $P_{\text{col,2}}$. 
First, we investigate repetition distributions with $\frac{1}{s_{\max}}\mathbb{E}[S]\ge \frac{1-\alpha\delta}{|\mathfrak{X}|}$. Let $F(n,k,|\mathfrak{X}|)$ denote the number of $|\mathfrak{X}|$-ary sequences of length $n$, which contain a fixed $|\mathfrak{X}|$-ary sequence of length $k$. We note that this $F(n,k,|\mathfrak{X}|)$ is constant for any $|\mathfrak{X}|$-ary sequence of length $k$ \cite[Lemma 1]{chvatal1975longest}. Now we define $G_{z^k}(n s_{\max},k,|\mathfrak{X}|)$ as the number of $s_{\max}$ times stretched sequences of length $n s_{\max}$, containing a $|\mathfrak{X}|$-ary sequence $z^k$ of length $k$. We stress that this counting function $G_{z^k}$ will not be independent of $z^k$ as is the case for the counting function $F$. For example, let $s_{\max}=2$, $\mathfrak{X}=\{0,1\}$, $n=2$, $k=2$, $z^k_1=01$ and $z^k_2=00$. Then we have $G_{z^k_1}(n s_{\max},k,|\mathfrak{X}|)=1$ since only $0011$ contains $z^k_1=01$, whereas $G_{z^k_2}(n s_{\max},k,|\mathfrak{X}|)=3$ since $0000$, $0011$ and $1100$ all contain $z^k_2=00$.

Observe that the maximum value of $G_{z^k}(n s_{\max},k,|\mathfrak{X}|)$ is attained when $z^k$ consists only of one symbol repeated $k$ times, as this grouping of elements in $z^k$ yields the maximum number of possible elementwise replicated sequences. WLOG, let $z^k=00\dots0$. Then, to count $G_{z^k}(n s_{\max},k,|\mathfrak{X}|)$, we group the consecutive $s_{\max}$ 0's in $z^k$ together, allowing the last group to have possibly fewer than $s_{\max}$ elements. It is clear that there are $\left\lceil \frac{k}{s_{\max}}\right\rceil$ of such groups of 0's. Since we put a stretching constraint on the sequences of length $n s_{\max}$ when we count $G_{z^k}(n s_{\max},k,|\mathfrak{X}|)$, we are looking for sequences of length $n$, containing a subsequence of length $\left\lceil \frac{k}{s_{\max}}\right\rceil$. Thus, counting this number will be the same as counting $F\left(n,\left\lceil \frac{k}{s_{\max}}\right\rceil,|\mathfrak{X}|\right)$. Thus we have 
\iftoggle{singlecolumn}{
\begin{align}
    G_{z^k}(n s_{\max},k,|\mathfrak{X}|)\le F\left(n,\left\lceil \nicefrac{k}{s_{\max}}\right\rceil,|\mathfrak{X}|\right),\quad \forall z^k\in \mathfrak{X}^k\label{eqn:ineqFG}
\end{align}
}{
\begin{align}
    G_{z^k}(n s_{\max},k,|\mathfrak{X}|)\le F\left(n,\left\lceil \nicefrac{k}{s_{\max}}\right\rceil,|\mathfrak{X}|\right),\quad \forall z^k\in \mathfrak{X}^k\label{eqn:ineqFG}
\end{align}
}
We note that the inequality given in \eqref{eqn:ineqFG} is the tightest upper bound independent of $z^k$, equality being achieved when $z^k$ is a constant (\textit{e.g.,} all-zeros) sequence.

Now, let 
\iftoggle{singlecolumn}{
\begin{align}
    T(z^k,A^n)&\triangleq \{x^n\in\mathfrak{X}^n:\bar{x}^{(n-a)}\in A_\epsilon^{(n-a)}(X) \text{ and }\tilde{x}^{(n-a) s_{\max}}  \text{ contains } z^k. \}
\end{align}
}{
\begin{align}
    T(z^k,A^n)&\triangleq \{x^n\in\mathfrak{X}^n:\bar{x}^{(n-a)}\in A_\epsilon^{(n-a)}(X)\notag\\&\hspace{2em} \text{ and }\tilde{x}^{(n-a) s_{\max}}  \text{ contains } z^k.\}
\end{align}
}
Then, we obtain
\iftoggle{singlecolumn}{
\begin{align}
    |T(z^k,A^n)|&\le G_{z^k}((n-a) s_{\max},k,|\mathfrak{X}|)\\
    &\le F\left(n-a,\left\lceil\nicefrac{k}{s_{\max}}\right\rceil,|\mathfrak{X}|\right)\label{eq:tset}
\end{align}
}{
\begin{align}
    |T(z^k,A^n)|&\le G_{z^k}((n-a) s_{\max},k,|\mathfrak{X}|)\\
    &\le F\left(n-a,\left\lceil\nicefrac{k}{s_{\max}}\right\rceil,|\mathfrak{X}|\right)\label{eq:tset}
\end{align}
}
For the sake of computational simplicity, suppose $\frac{k}{s_{\max}}$ is an integer. Since ${\frac{1}{s_{\max}}\mathbb{E}[S]\ge \frac{1-\alpha\delta}{|\mathfrak{X}|}}$, from~\cite{chvatal1975longest} and~\cite[Chapter 11]{cover2006elements} we have the following upper bound:
\iftoggle{singlecolumn}{
\begin{align}
    F\left(n-a,\nicefrac{k}{s_{\max}},|\mathfrak{X}|\right)&\le (n-a) 2^{(n-a) H_b\left(\frac{k}{s_{\max}(n-a)}\right)} (|\mathfrak{X}|-1)^{\left( n-a-\frac{k}{s_{max}}\right)}
\end{align}
}{
\begin{align}
    F\left(n-a,\nicefrac{k}{s_{\max}},|\mathfrak{X}|\right)&\le (n-a) 2^{(n-a) H_b\left(\frac{k}{s_{\max}(n-a)}\right)} \notag\\&\hspace{1em}(|\mathfrak{X}|-1)^{\left( n-a-\frac{k}{s_{max}}\right)}
\end{align}
}
Furthermore, for any $x^n\in T(z^k,A^n)$, since $T(z^k,A^n)\subseteq A_\epsilon^{(n-a)}(X)$, we have
\iftoggle{singlecolumn}{
\begin{align}
    p_{X^n}(x^n)\le 2^{-(n-a) (H(X)-\epsilon)}
\end{align}
}{
\begin{align}
    p_{X^n}(x^n)\le 2^{-(n-a) (H(X)-\epsilon)}
\end{align}
}
and since the rows $X^n_i$ of $\mathbf{X}$ are \emph{i.i.d.}, we have 
\iftoggle{singlecolumn}{
\begin{align}
    \Pr(X^n_2 \in T(z^k,A^n)|X^n_1 \in T(z^k,A^n))=\Pr(X^n_2 \in T(z^k,A^n))
\end{align}
}{
\begin{align}
    \Pr(X^n_2 \in T(z^k,A^n)|X^n_1 &\in T(z^k,A^n))\notag\\&=\Pr(X^n_2 \in T(z^k,A^n))
\end{align}
}
Finally, we have
\iftoggle{singlecolumn}{
\begin{align}
    |A_\epsilon(X^k|Y^k_t)| &\le 2^{k(H(X|Y)+\epsilon)}\label{eq:condtypicalset}
\end{align}
}{
\begin{align}
    |A_\epsilon(X^k|Y^k_t)| &\le 2^{k(H(X|Y)+\epsilon)}\label{eq:condtypicalset}
\end{align}
}

Combining \eqref{eq:tset}-\eqref{eq:condtypicalset}, we can upper bound $P_{\text{col,2}}$ as
\iftoggle{singlecolumn}{
\begin{align}
    P_{\text{col,2}} &\le \sum\limits_{z^k\in A_\epsilon(X^k|Y^k_t)} \Pr(X^n_2 \in T(z^k,A^n))\\
    &=\sum\limits_{z^k\in A_\epsilon(X^k|Y^k_t)} \sum\limits_{x^n\in T(z^k,A^n)} p_{X^n}(x^n)\\
    &\le \sum\limits_{z^k\in A_\epsilon(X^k|Y^k_t)} \sum\limits_{x^n\in T(z^k,A^n)} 2^{-(n-a)(H(X)-\epsilon)}\\
    &= \sum\limits_{z^k\in A_\epsilon(X^k|Y^k_t)} |T(z^k,A^n)| 2^{-(n-a)(H(X)-\epsilon)}\label{eq:tsetseparation}\\
    &\le  \sum\limits_{z^k\in A_\epsilon(X^k|Y^k_t)} 2^{-(n-a)(H(X)-\epsilon)} F\left(n-a,\nicefrac{k}{s_{\max}},|\mathfrak{X}|\right)\\
    &= |A_\epsilon(X^k|Y^k_t)| 2^{-(n-a)(H(X)-\epsilon)} F\left(n-a,\nicefrac{k}{s_{\max}},|\mathfrak{X}|\right)\\
    &\le |A_\epsilon(X^k|Y^k_t)| (n-a) 2^{-(n-a)\left[H(X)-\epsilon-H_b(\frac{k}{s_{\max}(n-a)})\right]} (|\mathfrak{X}|-1)^{(n-a-\frac{k}{s_{\max}})}\\
    &\le 2^{k(H(X|Y)+\epsilon)} (n-a) 2^{-(n-a) \left[H(X)-\epsilon-H_b(\frac{k}{s_{\max}(n-a)})\right]} (|\mathfrak{X}|-1)^{(n-a-\frac{k}{s_{max}})}
\end{align}
}{
\begin{align}
    P_{\text{col,2}} &\le \sum\limits_{z^k\in A_\epsilon(X^k|Y^k_t)} \Pr(X^n_2 \in T(z^k,A^n))\\
    &=\sum\limits_{z^k\in A_\epsilon(X^k|Y^k_t)} \sum\limits_{x^n\in T(z^k,A^n)} p_{X^n}(x^n)\\
    &\le \sum\limits_{z^k\in A_\epsilon(X^k|Y^k_t)} \sum\limits_{x^n\in T(z^k,A^n)} 2^{-(n-a)(H(X)-\epsilon)}\\
    &= \sum\limits_{z^k\in A_\epsilon(X^k|Y^k_t)} |T(z^k,A^n)| 2^{-(n-a)(H(X)-\epsilon)}\label{eq:tsetseparation}\\
    &\le  \sum\limits_{z^k\in A_\epsilon(X^k|Y^k_t)} 2^{-(n-a)(H(X)-\epsilon)}\notag\\&\hspace{4em} F\left(n-a,\nicefrac{k}{s_{\max}},|\mathfrak{X}|\right)\\
    &= |A_\epsilon(X^k|Y^k_t)| 2^{-(n-a)(H(X)-\epsilon)} \notag\\&\hspace{4em}F\left(n-a,\nicefrac{k}{s_{\max}},|\mathfrak{X}|\right)\\
    &\le |A_\epsilon(X^k|Y^k_t)| (n-a) (|\mathfrak{X}|-1)^{(n-a-\frac{k}{s_{\max}})}\notag\\&\hspace{4em}2^{-(n-a)\left[H(X)-\epsilon-H_b(\frac{k}{s_{\max}(n-a)})\right]}\\
    &\le 2^{k(H(X|Y)+\epsilon)} (n-a)  (|\mathfrak{X}|-1)^{(n-a-\frac{k}{s_{max}})}\notag\\&\hspace{4em}2^{-(n-a) \left[H(X)-\epsilon-H_b(\frac{k}{s_{\max}(n-a)})\right]}
\end{align}
}
Thus, we have the following upper bound on the error probability
\iftoggle{singlecolumn}{
\begin{align}
    P_e &\le 2\epsilon+ 2^{n R} 2^{k(H(X|Y)+\epsilon)} (n-a) 2^{-(n-a) \left[H(X)-\epsilon-H_b(\frac{k}{s_{\max}(n-a)})\right]} (|\mathfrak{X}|-1)^{(n-a-\frac{k}{s_{max}})} +\kappa_n+\mu_n
\end{align}
}{
\begin{align}
    P_e &\le 2\epsilon+\kappa_n+\mu_n\notag\\
    &\hspace{1em}+ 2^{n R} 2^{k(H(X|Y)+\epsilon)} (n-a)(|\mathfrak{X}|-1)^{(n-a-\frac{k}{s_{max}})} \notag\\&\hspace{4em} 2^{-(n-a) \left[H(X)-\epsilon-H_b(\frac{k}{s_{\max}(n-a)})\right]} 
\end{align}
}
By LLN, we have $\kappa_n\to0$ and $\mu_n\to0$ as $n\to\infty$. Hence, we can argue that any database growth rate $R$ satisfying 
\iftoggle{singlecolumn}{
\begin{align}
    R&<\Big[(1-\alpha\delta)\left(H(X)-H_b\left( \frac{\mathbb{E}[S]}{(1-\alpha\delta)s_{\max}}\right)\right) \notag\\
    &\hspace{6em}-\left(1-\alpha\delta-\frac{\mathbb{E}[S]}{s_{\max}} \right)\log\left(|\mathfrak{X}|-1\right)-\mathbb{E}[S]H(X|Y)\Big]^+ \label{eq:rowwiseachievablerate}
\end{align}
}{
\begin{align}
    R&<\Big[(1-\alpha\delta)\left(H(X)-H_b\left( \frac{\mathbb{E}[S]}{(1-\alpha\delta)s_{\max}}\right)\right) \notag\\&-\left(1-\alpha\delta-\frac{\mathbb{E}[S]}{s_{\max}} \right)\log\left(|\mathfrak{X}|-1\right)-\mathbb{E}[S]H(X|Y)\Big]^+ \label{eq:rowwiseachievablerate}
\end{align}
}
is achievable, by taking $\epsilon$ small enough.

Now, we focus on general repetition distributions. For any subsequence $z^k$ of $s_{\max}$-times stretched sequence of length $(n-a)s_{\max}$, let $r(z^k)$ be the number of runs in $z^k$ with at most $s_{\max}$ elements and note that $r(z^k)\le n-a$. Then, let $\tilde{z}^{r(z^k)}$ be the sequence storing the values of each run in $z^k$. Observe that for any $z^{k}\in A_\epsilon(X^{k}|Y^k_t)$, we have $\tilde{z}^{r(z^k)}\in A_\epsilon^{(r(z^k))}(X)$.

\begin{sloppypar}
For any such grouping of $r(z^k)$ runs, the $\epsilon$-typicality of $x^n=(x_1,\dots,x_n)\in T(z^k,A^n)$ and $\tilde{z}^{r(z^k)}$ with respect to $p_X$ implies the $\tilde{\epsilon}$-typicality of the remaining sequence of length ${n-a-r(z^k)}$ obtained after discarding $\tilde{z}^{r(z^k)}$ from $\bar{x}^{n-a}$, where $\tilde{\epsilon}=\frac{n-a+r(z^k)}{n-a-r(z^k)}\epsilon$. Furthermore, by a similar argument made above, we stress that $T(z^k,A^n)$ attains its maximum value when $r(z^k)$ is the minimum, which is $k_{\min}\triangleq\lceil \frac{k}{s_{\max}}\rceil$, attained when $z^k$ is a $s_{\max}$ times stretched sequence itself. Therefore for any $z^{k}\in A_\epsilon(X^{k}|Y^k_t)$, taking the union bound over all possible groupings with $r(z^k)$ runs, the cardinality of $T(z^{k},A^n)$ can be upper bounded as
\iftoggle{singlecolumn}{
\begin{align}
    |T(z^{k},A^n)| &\le \binom{n-a}{k_{\min}} |A_{\tilde{\epsilon}}^{(n-a-k_{\min})}(X)|\\
    &\le 2^{(n-a) H_b\left(\frac{k_{\min}}{n-a}\right)} |A_{\tilde{\epsilon}}^{(n-a-\hat{k})}(X)|\\
    &\le 2^{(n-a) H_b\left(\frac{k_{\min}}{n-a}\right)} 2^{(n-a-k_{\min})(H(X)+\tilde{\epsilon})}\\
    &=  2^{n \left[(1-\frac{a}{n})H_b\left(\frac{k_{\min}}{n-a}\right)+(1-\frac{a}{n}-\frac{k_{\min}}{n})(H(X)+\tilde{\epsilon})\right]}\label{eq:tsettypicalbound2}
\end{align}
}{
\begin{align}
    |T(z^{k},A^n)| &\le \binom{n-a}{k_{\min}} |A_{\tilde{\epsilon}}^{(n-a-k_{\min})}(X)|\\
    &\le 2^{(n-a) H_b\left(\frac{k_{\min}}{n-a}\right)} |A_{\tilde{\epsilon}}^{(n-a-\hat{k})}(X)|\\
    &\le 2^{(n-a) H_b\left(\frac{k_{\min}}{n-a}\right)} 2^{(n-a-k_{\min})(H(X)+\tilde{\epsilon})}\\
    &=  2^{n \left[(1-\frac{a}{n})H_b\left(\frac{k_{\min}}{n-a}\right)+(1-\frac{a}{n}-\frac{k_{\min}}{n})(H(X)+\tilde{\epsilon})\right]}\label{eq:tsettypicalbound2}
\end{align}
}
\end{sloppypar}
\begin{sloppypar}
Plugging \eqref{eq:tsettypicalbound2} into \eqref{eq:tsetseparation} and following the same steps, one can show that any rate $R$ satisfying 
\iftoggle{singlecolumn}{
\begin{align}
    R&<\left[\frac{\mathbb{E}[S]}{s_{\max}} H(X)-(1-\alpha\delta)H_b\left( \frac{\mathbb{E}[S]}{(1-\alpha\delta)s_{\max}}\right)-\mathbb{E}[S]H(X|Y)\right]^+\label{eq:rowwiseachievablerate2}
\end{align}
}{
\begin{align}
    R&<\Big[\frac{\mathbb{E}[S]}{s_{\max}} H(X)-\mathbb{E}[S]H(X|Y)\notag\\&\hspace{3em}-(1-\alpha\delta)H_b\left( \frac{\mathbb{E}[S]}{(1-\alpha\delta)s_{\max}}\right)\Big]^+\label{eq:rowwiseachievablerate2}
\end{align}
}
is achievable. Simply taking the maximum of the two proven achievable rates (\eqref{eq:rowwiseachievablerate} and \eqref{eq:rowwiseachievablerate2}) when ${\frac{1}{s_{\max}}\mathbb{E}[S]\ge \frac{1-\alpha\delta}{|\mathfrak{X}|}}$ yields \eqref{eq:rowwiseachievable2}. This concludes the proof. \qed
\end{sloppypar}

\section{Proof of Corollary~\ref{cor:W1conversenoiselessub}}
\label{proof:rowwiseconversenoiselessub}
Let $E$ denote the empty string and $\Tilde{X}$ denote the sequence obtained after discarding the detected deleted entries from $X^2$. The dependence of $\Tilde{X}$ on $X^2$ and $A^2$ and that of $Y$ on $X^2$ and $S^2$ are omitted for brevity. 

We start with the fact that since the entries of $X^2$ are independent, the deleted entries do not offer any information. Thus, we can discard them without any information loss. Thus, we have
\iftoggle{singlecolumn}{
\begin{align}
    I(X^2;Y,A^2)&= I(\Tilde{X};Y|A^2)\\
    &=H(\Tilde{X}|A^2)-H(\Tilde{X}|Y,{A}^2)
\end{align}
}{
\begin{align}
    I(X^2;Y,A^2)&= I(\Tilde{X};Y|A^2)\\
    &=H(\Tilde{X}|A^2)-H(\Tilde{X}|Y,{A}^2)
\end{align}
}

\iftoggle{singlecolumn}{
We have
\begin{align}
    H(\Tilde{X}|{A}^2)&=\sum\limits_{{a}^2\in\{0,1\}^2}\Pr({A}^2={a}^2) H(\Tilde{{X}}|{A}^2={a}^2)\\
    &=\Pr({A}^2=00)H(\Tilde{{X}}|{A}^2=00)+\Pr({A}^2=01)H(\Tilde{{X}}|{A}^2=01)\notag\\
    &\hspace{3em}+ \Pr({A}^2=10)H(\Tilde{{X}}|{A}^2=10)+\Pr({A}^2=11)H(\Tilde{{X}}|{A}^2=11)\\
    &= (1-\alpha\delta)^2 2H(X)+\alpha\delta(1-\alpha\delta)H(X) + (1-\alpha\delta)\alpha\delta H(X)+0\\
    &=2(1-\alpha\delta)H(X)\label{eq:HXtildeA}
\end{align}
}{
\begingroup
\allowdisplaybreaks
We have
\begin{align}
    H(\Tilde{X}|{A}^2)&=\sum\limits_{{a}^2\in\{0,1\}^2}\Pr({A}^2={a}^2) H(\Tilde{{X}}|{A}^2={a}^2)\\
    &=\Pr({A}^2=00)H(\Tilde{{X}}|{A}^2=00) \notag\\&\hspace{1em} +\Pr({A}^2=01)H(\Tilde{{X}}|{A}^2=01)\notag\\
    &\hspace{1em}+ \Pr({A}^2=10)H(\Tilde{{X}}|{A}^2=10) \notag\\&\hspace{1em} +\Pr({A}^2=11)H(\Tilde{{X}}|{A}^2=11)\\
    &= (1-\alpha\delta)^2 2H(X) \notag\\&\hspace{1em} +\alpha\delta(1-\alpha\delta)H(X)  \notag\\&\hspace{1em} + (1-\alpha\delta)\alpha\delta H(X) \notag\\&\hspace{1em} +0\\
    &=2(1-\alpha\delta)H(X)\label{eq:HXtildeA}
\end{align}
\endgroup
}
Furthermore, we have
\iftoggle{singlecolumn}{
\begin{align}
    H(\Tilde{{X}}|{Y},{A}^2)&=\sum\limits_{{y},{a}^2} \Pr({Y}={y},{A}^2={a}^2) H(\Tilde{{X}}|{Y}={y},{A}^2={a}^2)\\
    &= \Pr({Y}=E,{A}^2=00)H(\Tilde{{X}}|{Y}=E,{A}^2=00)\notag\\
    &\hspace{2em} +\Pr({Y}=E,{A}^2=01)H(\Tilde{{X}}|{Y}=E,{A}^2=01)\notag\\
    &\hspace{2em} +\Pr({Y}=E,{A}^2=10)H(\Tilde{{X}}|{Y}=E,{A}^2=10)\notag\\
    &\hspace{2em} +\sum\limits_{x\in\mathfrak{X}} \Pr({Y}=x,{A}^2=00)H(\Tilde{{X}}|{Y}=x,{A}^2=00)\notag\\
    &\hspace{2em} +\sum\limits_{x\in\mathfrak{X}} \Pr({Y}=x,{A}^2=01)H(\Tilde{{X}}|{Y}=x,{A}^2=01)\notag\\
    &\hspace{2em} +\sum\limits_{x\in\mathfrak{X}} \Pr({Y}=x,{A}^2=10)H(\Tilde{{X}}|{Y}=x,{A}^2=10)\label{eq:HXtildeYA}
\end{align}
}{
\begin{align}
    H&(\Tilde{{X}}|{Y},{A}^2)=\sum\limits_{{y},{a}^2} \Pr({Y}={y},{A}^2={a}^2) \notag\\&\hspace{7em}  H(\Tilde{{X}}|{Y}={y},{A}^2={a}^2)\\
    &= \Pr({Y}=E,{A}^2=00)H(\Tilde{{X}}|{Y}=E,{A}^2=00)\notag\\
    &\hspace{1em} +\Pr({Y}=E,{A}^2=01)H(\Tilde{{X}}|{Y}=E,{A}^2=01)\notag\\
    &\hspace{1em} +\Pr({Y}=E,{A}^2=10)H(\Tilde{{X}}|{Y}=E,{A}^2=10)\notag\\
    &\hspace{1em} +\sum\limits_{x\in\mathfrak{X}} \Pr({Y}=x,{A}^2=00)H(\Tilde{{X}}|{Y}=x,{A}^2=00)\notag\\
    &\hspace{1em} +\sum\limits_{x\in\mathfrak{X}} \Pr({Y}=x,{A}^2=01)H(\Tilde{{X}}|{Y}=x,{A}^2=01)\notag\\
    &\hspace{1em} +\sum\limits_{x\in\mathfrak{X}} \Pr({Y}=x,{A}^2=10)H(\Tilde{{X}}|{Y}=x,{A}^2=10)\label{eq:HXtildeYA}
\end{align}
}
Note that in \eqref{eq:HXtildeYA}, we discarded the terms with ${A}^2=11$ for $|{Y}|\ge 1$, since in that case we have ${\Pr(|{Y}|\ge 1,{A}^2=11)=0}$. We can further discard the terms with $|{Y}|=n=2$, since in that case we have no deletion and ${Y}={Y}^2={X}^2$. Finally, we can also discard the last two terms in \eqref{eq:HXtildeYA} since for any $x\in\mathfrak{X}$ we have 
\iftoggle{singlecolumn}{
\begin{align}
    H(\Tilde{{X}}|{Y}=x,{A}^2=01)=H(\Tilde{{X}}|{Y}=x,{A}^2=10)=0
\end{align}
}{
\begin{align}
    H(\Tilde{{X}}|{Y}=x,{A}^2=01)=H(\Tilde{{X}}|{Y}=x,{A}^2=10)=0
\end{align}
}

\iftoggle{singlecolumn}{
Thus, we have
\begin{align}
    H(\Tilde{{X}}|{Y},{A}^2)&=\delta^2(1-\alpha)^2 2H(X)+\delta^2(1-\alpha)\alpha H(X)+\delta^2\alpha(1-\alpha)H(X)\notag\\&\hspace{2em}+\sum\limits_{x\in\mathfrak{X}} \Pr({Y}=x,{A}^2=00)H(\Tilde{{X}}|{Y}=x,{A}^2=00)\\
    &= 2\delta^2(1-\alpha)H(X)+\sum\limits_{x\in\mathfrak{X}} \Pr({Y}=x,{A}^2=00)H(\Tilde{{X}}|{Y}=x,{A}^2=00)\label{eq:HXtilde1}
\end{align}
}{
\begingroup
Thus, we have
\begin{align}
    H(\Tilde{{X}}|{Y},{A}^2)&=\delta^2(1-\alpha)^2 2H(X) \notag\\&\hspace{1em} +\delta^2(1-\alpha)\alpha H(X) \notag\\&\hspace{1em} +\delta^2\alpha(1-\alpha)H(X)\notag\\&\hspace{1em}+\sum\limits_{x\in\mathfrak{X}} \Pr({Y}=x,{A}^2=00) \notag\\&\hspace{4em} H(\Tilde{{X}}|{Y}=x,{A}^2=00)\\
    &= 2\delta^2(1-\alpha)H(X) \notag\\&\hspace{1em} +\sum\limits_{x\in\mathfrak{X}} \Pr({Y}=x,{A}^2=00) \notag\\&\hspace{4em} H(\Tilde{{X}}|{Y}=x,{A}^2=00)\label{eq:HXtilde1}
\end{align}
\endgroup
}

We first compute $\Pr({Y}=x,{A}^2=00)$. For any $x\in\mathfrak{X}$, we have
\iftoggle{singlecolumn}{
\begin{align}
    \Pr({Y}=x,{A}^2=00) &= \sum\limits_{{x}^2\in\mathfrak{X}^2}\Pr({Y}=x,{A}^2=00,{X}^2={x}^2)\\
    &= \Pr({Y}=x,{A}^2=00,{X}^2=xx) \notag\\
    &\hspace{2em}+ 2\sum\limits_{y\neq x}\Pr({Y}=x,{A}^2=00,{X}^2=xy)
    \\
    &= p_X(x)^2 2\delta(1-\delta)(1-\alpha) + 2\sum\limits_{y\neq x} p_X(x) p_X(y) \delta(1-\delta)(1-\alpha)
    \\
    &= 2\delta(1-\delta)(1-\alpha)p_X(x)\sum\limits_{y\in\mathfrak{X}} p_X(y)\\
    &= 2\delta(1-\delta)(1-\alpha)p_X(x)\label{eq:HXtilde2}
\end{align}
}{
\begin{align}
    \Pr&({Y}=x,{A}^2=00)\notag \\&= \sum\limits_{{x}^2\in\mathfrak{X}^2}\Pr({Y}=x,{A}^2=00,{X}^2={x}^2)\\
    &= \Pr({Y}=x,{A}^2=00,{X}^2=xx) \notag\\&\hspace{1em}  + 2\sum\limits_{y\neq x}\Pr({Y}=x,{A}^2=00,{X}^2=xy)
    \\
    &= p_X(x)^2 2\delta(1-\delta)(1-\alpha)  \notag\\&\hspace{1em} + 2\sum\limits_{y\neq x} p_X(x) p_X(y) \delta(1-\delta)(1-\alpha)
    \\
    &= 2\delta(1-\delta)(1-\alpha)p_X(x)\sum\limits_{y\in\mathfrak{X}} p_X(y)\\
    &= 2\delta(1-\delta)(1-\alpha)p_X(x)\label{eq:HXtilde2}
\end{align}
}

Now, we compute $H(\Tilde{{X}}|{Y}=x,{A}^2=00)$. For any $x\in\mathfrak{X}$ we have $2|\mathfrak{X}|-1$ possible patterns for $\Tilde{{X}}$, given that ${Y}=x$. $2|\mathfrak{X}|-2$ of these patterns have probabilities proportional to $p_X(x) p_X(y)$ $y\in \mathfrak{X}\setminus\{x\}$ and the remaining pattern has probability proportional to $2 p_X(x)^2$. Thus we have
\iftoggle{singlecolumn}{
\begin{align}
    H(\Tilde{{X}}|{Y}=x,{A}^2=00)= H\Big(\frac{p_X(1) p_X(x)}{c},&\frac{p_X(x) p_X(1)}{c},\dots,\frac{2 (p_X(x))^2}{c},\notag\\&\dots,\frac{p_X(|\mathfrak{X}|) p_X(x)}{c},\frac{p_X(x) p_X(|\mathfrak{X}|)}{c}\Big)
\end{align}
}{
\begin{align}
    H&(\Tilde{{X}}|{Y}=x,{A}^2=00)\notag\\&= H\Big(\frac{p_X(1) p_X(x)}{c},\frac{p_X(x) p_X(1)}{c},\notag\\&\hspace{3em}\dots,\frac{2 p_X(x)^2}{c},\dots,\notag\\&\hspace{4em}\frac{p_X(|\mathfrak{X}|) p_X(x)}{c},\frac{p_X(x) p_X(|\mathfrak{X}|)}{c}\Big)
\end{align}
}
where the normalization constant $c$ is  $c=2p_X(x)$. Thus,
\iftoggle{singlecolumn}{
\begin{align}
    H(\Tilde{{X}}|{Y}=x,{A}^2=00)&= H\Big(\frac{p_X(1) }{2},\frac{ p_X(1)}{2},\dots, p_X(x),\dots,\frac{p_X(|\mathfrak{X}|)}{2},\frac{ p_X(|\mathfrak{X}|)}{2}\Big)\\
    &= H(X)+1-p_X(x)\label{eq:HXtilde3}
\end{align}
}{
\begin{align}
    H(\Tilde{{X}}|{Y}=x,&{A}^2=00)\notag\\&= H\Big(\frac{p_X(1) }{2},\frac{ p_X(1)}{2},\notag\\&\hspace{3em}\dots, p_X(x),\dots,\notag\\&\hspace{4em}\frac{p_X(|\mathfrak{X}|)}{2},\frac{ p_X(|\mathfrak{X}|)}{2}\Big)\\
    &= H(X)+1-p_X(x)\label{eq:HXtilde3}
\end{align}
}

\iftoggle{singlecolumn}{
Combining \eqref{eq:HXtilde1}-\eqref{eq:HXtilde3}, we can compute $H(\Tilde{{X}}|{Y},{A}^2)$ as 
\begin{align}
    H(\Tilde{{X}}|{Y},{A}^2)&=2\delta^2(1-\alpha)H(X)+\sum\limits_{x\in\mathfrak{X}} 2\delta(1-\delta)(1-\alpha)p_X(x) [H(X)+1-p_X(x)]\\
    &= 2\delta^2(1-\alpha)H(X)+2\delta(1-\delta)(1-\alpha) \left(H(X)+1-\hat{q}\right)\\
    &= 2\delta(1-\alpha)H(X)+2\delta(1-\delta)(1-\alpha)\left(1-\hat{q}\right)\label{eq:HXtildeYAfinal}
\end{align}
}{
\begingroup
\allowdisplaybreaks
Combining \eqref{eq:HXtilde1}-\eqref{eq:HXtilde3}, we can compute $H(\Tilde{{X}}|{Y},{A}^2)$ as 
\begin{align}
    H(\Tilde{{X}}|{Y},{A}^2)&=2\delta^2(1-\alpha)H(X)\notag\\&\hspace{1em}+\sum\limits_{x\in\mathfrak{X}} 2\delta(1-\delta)(1-\alpha)p_X(x)\notag\\&\hspace{5em} [H(X)+1-p_X(x)]\\
    &= 2\delta^2(1-\alpha)H(X)\notag\\&\hspace{1em}+2\delta(1-\delta)(1-\alpha) \notag\\&\hspace{3em}\left(H(X)+1-\hat{q}\right)\\
    &= 2\delta(1-\alpha)H(X)\notag\\&\hspace{1em}+2\delta(1-\delta)(1-\alpha)(1-\hat{q})\label{eq:HXtildeYAfinal}
\end{align}
\endgroup
}
Finally, combining \eqref{eq:HXtildeA} and \eqref{eq:HXtildeYAfinal}, we obtain
\iftoggle{singlecolumn}{
\begin{align}
    I(\Tilde{{X}};{Y}^K|{A}^2)
    &=H(\Tilde{{X}}|{A}^2)-H(\Tilde{{X}}|{Y}({X}^2),{A}^2)\\
    &= 2(1-\alpha\delta)H(X)-2\delta(1-\alpha)H(X)-2\delta(1-\delta)(1-\alpha)\left(1-\hat{q}\right)\\
    &=2(1-\delta)H(X)-2\delta(1-\delta)(1-\alpha)\left(1-\hat{q}\right)\label{eqn:MI2}
\end{align}
}{
\begin{align}
    I(\Tilde{{X}};{Y}^K|{A}^2)
    &=H(\Tilde{{X}}|{A}^2)-H(\Tilde{{X}}|{Y}({X}^2),{A}^2)\\
    &= 2(1-\alpha\delta)H(X)-2\delta(1-\alpha)H(X)\notag\\
    &\hspace{1em}-2\delta(1-\delta)(1-\alpha)\left(1-\hat{q}\right)\\
    &=2(1-\delta)H(X)\notag\\&\hspace{1em}-2\delta(1-\delta)(1-\alpha)\left(1-\hat{q}\right)\label{eqn:MI2}
\end{align}
}
Thus, we have
\iftoggle{singlecolumn}{
\begin{align}
    \frac{1}{2}I({X}^2;{Y}^K,{A}^2)
    &=(1-\delta)H(X)-\delta(1-\delta)(1-\alpha)\left(1-\hat{q}\right)
\end{align}
}{
\begin{align}
    \frac{1}{2}I({X}^2;{Y}^K,{A}^2)
    &=(1-\delta)H(X)\notag\\&\hspace{1em}-\delta(1-\delta)(1-\alpha)\left(1-\hat{q}\right)
\end{align}
}
concluding the proof.
\qed


\section{Proof of Corollary~\ref{cor:W1conversenoisybinaryub}}
\label{proof:W1conversenoisybinaryub}
We start by observing that
\iftoggle{singlecolumn}{
\begin{align}
    I({X}^2;{Y},{A}^2)&= I({X}^2;{Y},|{Y}|,{A}^2)\\
    &= H({X}^2)-H({X}^2|{Y},|{Y}|,{A}^2)\\
    &= 2H(X)-H({X}^2|{Y},|{Y}|,{A}^2)\label{eq:MIXYAnoisy1}
\end{align}
}{
\begin{align}
    I({X}^2;{Y},{A}^2)&= I({X}^2;{Y},|{Y}|,{A}^2)\\
    &= H({X}^2)-H({X}^2|{Y},|{Y}|,{A}^2)\\
    &= 2H(X)-H({X}^2|{Y},|{Y}|,{A}^2)\label{eq:MIXYAnoisy1}
\end{align}
}
\iftoggle{singlecolumn}{
Furthermore, we have
\begin{align}
    H({X}^2|{Y},|{Y}|,{A}^2)&= \sum\limits_{i=0}^2 \Pr(|{Y}|=i) H({X}^2|{Y},|{Y}|=i,{A}^2)\\
    &=\delta^2 H({X}^2|{Y},|{Y}|=0,{A}^2) \notag\\
    &\hspace{2em}+ 2\delta(1-\delta) H({X}^2|{Y},|{Y}|=1,{A}^2)\notag\\
    &\hspace{2em}+ (1-\delta)^2 H({X}^2|{Y},|{Y}|=2,{A}^2)\\
    &= \delta^2 2 H(X) \notag\\
    &\hspace{2em}+ 2\delta(1-\delta) H({X}^2|{Y},|{Y}|=1,{A}^2)\notag\\
    &\hspace{2em}+ (1-\delta)^2 2 H(X|Y)\\
    &= \delta^2 2 H(X) \notag\\
    &\hspace{2em}+ 2\delta(1-\delta)\alpha [H(X)+H(X|Y)]\notag\\
    &\hspace{2em}+ 2\delta(1-\delta)(1-\alpha) H({X}^2|{Y},|{Y}|=1,{A}^2=00)\notag\\
    &\hspace{2em}+ (1-\delta)^2 2 H(X|Y)
\end{align}
}{
\begingroup

Furthermore, we have
\begin{align}
    H&({X}^2|{Y},|{Y}|,{A}^2)\notag\\&= \sum\limits_{i=0}^2 \Pr(|{Y}|=i) H({X}^2|{Y},|{Y}|=i,{A}^2)\\
    &=\delta^2 H({X}^2|{Y},|{Y}|=0,{A}^2) \notag\\
    &\hspace{1em}+ 2\delta(1-\delta) H({X}^2|{Y},|{Y}|=1,{A}^2)\notag\\
    &\hspace{1em}+ (1-\delta)^2 H({X}^2|{Y},|{Y}|=2,{A}^2)\\
    &= \delta^2 2 H(X) \notag\\
    &\hspace{1em}+ 2\delta(1-\delta) H({X}^2|{Y},|{Y}|=1,{A}^2)\notag\\
    &\hspace{1em}+ (1-\delta)^2 2 H(X|Y)\\
    &= \delta^2 2 H(X) \notag\\
    &\hspace{1em}+ 2\delta(1-\delta)\alpha [H(X)+H(X|Y)]\notag\\
    &\hspace{1em}+ 2\delta(1-\delta)(1-\alpha) H({X}^2|{Y},|{Y}|=1,{A}^2=00)\notag\\
    &\hspace{1em}+ (1-\delta)^2 2 H(X|Y)
\end{align}
\endgroup
}
Note that we can rewrite $H({X}^2|{Y},|{Y}|=1,{A}^2=00)$ as
\iftoggle{singlecolumn}{
\begin{align}
    H({X}^2|{Y},|{Y}|=1,{A}^2=00)&=H({X}^2|{Y},|{Y}|=1)\\
    &= 2H(X)- I({X}^2;{Y}||{Y}|=1)\\
    &= 2H(X) - [H(Y)-H({Y}|{X}^2,|{Y}|=1)]
\end{align}
}{
\begin{align}
    H&({X}^2|{Y},|{Y}|=1,{A}^2=00)\notag\\&\hspace{1em}=H({X}^2|{Y},|{Y}|=1)\\
    &\hspace{1em}= 2H(X)- I({X}^2;{Y}||{Y}|=1)\\
    &\hspace{1em}= 2H(X) - [H(Y)-H({Y}|{X}^2,|{Y}|=1)]
\end{align}
}
where we have
\iftoggle{singlecolumn}{
\begin{align}
    H({Y}|{X}^2,|{Y}|=1)&= \sum\limits_{{x}^2\in\mathfrak{X}^2} \Pr({X}^2={x}^2) H({Y}|{X}^2={x}^2,|{Y}|=1)\label{eq:HYgivenXlengthY1}
\end{align}
}{
\begin{align}
    H&({Y}|{X}^2,|{Y}|=1)\notag\\
    \hspace{1em}&= \sum\limits_{{x}^2\in\mathfrak{X}^2} \Pr({X}^2={x}^2)\notag\\[-1em]&\hspace{6em} H({Y}|{X}^2={x}^2,|{Y}|=1)\label{eq:HYgivenXlengthY1}
\end{align}
}
Writing the sum in \eqref{eq:HYgivenXlengthY1} explicitly, we obtain
\iftoggle{singlecolumn}{
\begin{align}
    H({Y}|{X}^2,|{Y}|=1)&= (1-p)^2 H({Y}|{X}^2=00,|{Y}|=1) + p^2 H({Y}|{X}^2=11,|{Y}|=1)\notag\\
    &\hspace{2em} +p(1-p) H({Y}|{X}^2=01,|{Y}|=1) \notag\\
    &\hspace{2em}+ p(1-p) H({Y}|{X}^2=10,|{Y}|=1)
\end{align}
}{
\begin{align}
    H({Y}|&{X}^2,|{Y}|=1)\notag\\
    &= (1-p)^2 H({Y}|{X}^2=00,|{Y}|=1) \notag\\&\hspace{1em}+ p^2 H({Y}|{X}^2=11,|{Y}|=1)\notag\\
    &\hspace{1em} +p(1-p) H({Y}|{X}^2=01,|{Y}|=1) \notag\\&\hspace{1em}+ p(1-p) H({Y}|{X}^2=10,|{Y}|=1)
\end{align}
}
Observing the following,
\iftoggle{singlecolumn}{
\begin{align}
    H({Y}|{X}^2=00,|{Y}|=1)&= H(Y|X=0)\\
    H({Y}|{X}^2=11,|{Y}|=1)&=H(Y|X=1)\\
    H({Y}|{X}^2=01,|{Y}|=1)&= H(V)\\
    H({Y}|{X}^2=10,|{Y}|=1)&= H(V)\\
    H(Y|X=0) + H(Y|X=1)&= 2\left[\frac{1}{2} H(Y|X=0) + \frac{1}{2} H(Y|X=1)\right] \\
    &= 2 H(V|U)
\end{align}
}{
\begin{align}
    H({Y}|{X}^2=00,|{Y}|=1)&= H(Y|X=0)\\
    H({Y}|{X}^2=11,|{Y}|=1)&=H(Y|X=1)\\
    H({Y}|{X}^2=01,|{Y}|=1)&= H(V)\\
    H({Y}|{X}^2=10,|{Y}|=1)&= H(V)
\end{align}
\begin{align}
    H(&Y|X=0) + H(Y|X=1)\notag\\
    &= 2\left[\frac{1}{2} H(Y|X=0) + \frac{1}{2} H(Y|X=1)\right] \\
    &= 2 H(V|U)
\end{align}
}
we obtain
\iftoggle{singlecolumn}{
\begin{align}
    H({Y}|{X}^2,|{Y}|=1)&= (1-p) H(Y|X=0)-p(1-p)H(Y|X=0)\notag\\
    &\hspace{1em}+ p H(Y|X=1)-p(1-p)H(Y|X=1)\notag\\
    &\hspace{1em} +2p(1-p) H(V)\\
    &= H(Y|X)+2p(1-p) I(U;V)
\end{align}
}{
\begin{align}
    H({Y}|{X}^2,|{Y}|=1)&= (1-p) H(Y|X=0)\notag\\&\hspace{1em}-p(1-p)H(Y|X=0)\notag\\
    &\hspace{1em}+ p H(Y|X=1)\notag\\&\hspace{1em}-p(1-p)H(Y|X=1)\notag\\
    &\hspace{1em} +2p(1-p) H(V)\\
    &= H(Y|X)+2p(1-p) I(U;V)
\end{align}
}
Hence, we have 
\iftoggle{singlecolumn}{
\begin{align}
    H({X}^2|{Y},|{Y}|=1,{A}^2=00)&=2H(X)-I(X;Y) + 2p(1-p) I(U;V) \label{eq:MIXYAnoisyn}
\end{align}
}{
\begin{align}
    H({X}^2|{Y},|{Y}|=1,{A}^2=00)&=2H(X)-I(X;Y) \notag\\&\hspace{1em}+ 2p(1-p) I(U;V) \label{eq:MIXYAnoisyn}
\end{align}
}
Combining \eqref{eq:MIXYAnoisy1}-\eqref{eq:MIXYAnoisyn}, we have
\iftoggle{singlecolumn}{
\begin{align}
    \frac{1}{2} I({X}^2;{Y}^{K},{A}^2) = (1-\delta) I(X;Y)-2 \delta(1-\delta) (1-\alpha)p(1-p) I(U;V)
\end{align}
}{
\begin{align}
    \frac{1}{2} I({X}^2;{Y}^{K},{A}^2) &= (1-\delta) I(X;Y)\notag\\&\hspace{0.4em}-2 \delta(1-\delta) (1-\alpha)p(1-p) I(U;V)
\end{align}
}
concluding the proof.\qed

\section{Proof of Theorem~\ref{thm:adversarialWm}}
\label{proof:adversarialWm}

First, we focus on $\delta\le 1-\hat{q}$ and prove the achievability part. For a given pair of matching rows, WLOG, $X_1^n$ of $\mathbf{X}$ and $Y_t^{K_n}$ of $\mathbf{Y}$ with $\sigma_n(1)=t$, let $P_e\triangleq \Pr(\hat{\sigma}_n(1)\neq t)$ be the probability of error of the following matching scheme:
\begin{enumerate}[label=\textbf{\arabic*)},leftmargin=1.3\parindent]
\item Construct the collapsed histogram vectors $\tilde{{H}}_j^{(1),n}$ and $\tilde{{H}}_j^{(2),K_n}$ as in~\eqref{eq:histogramdefn}-\eqref{eq:histogramdefn2}.

\item Check the uniqueness of the entries $\tilde{H}^{(1)}_j$ $j\in[n]$ of $\tilde{{H}}^{(1),n}$. If there are at least two that are identical, declare a \emph{detection error} whose probability is denoted by $\mu_n$. Otherwise, proceed with Step~3.
\item $\forall i\in[n]$ if $ \nexists j\in[K_n]$, $\tilde{H}_i^{(1)}=\tilde{H}_j^{(2)}$, declare the $i$\textsuperscript{th} column of $\mathbf{X}$ deleted, assigning $i\in \hat{I}_{\text{del}}$. Note that conditioned on Step~2, this step is error-free.
\item Match the $t$\textsuperscript{th} row $Y^{K_n}_t$ of $\mathbf{Y}$ with the $1$\textsuperscript{st} row $X^n_1$ of $\mathbf{X}$, assigning $\hat{\sigma}_n(1)=t$ if the $1$\textsuperscript{st} row $\hat{X}_1^{K_n}(\hat{I}_{\text{del}})$ of $\hat{\mathbf{X}}$ is the only row of $\hat{\mathbf{X}}$ equal to $Y^{K_n}_t$ where $\hat{X}_i^{K_n}(\hat{I}_{\text{del}})$ is obtained by discarding the elements of $X_i^n$ whose indices lie in $\hat{I}_{\text{del}}$. Otherwise, declare a \emph{collision error}.
\end{enumerate}
Let $I(\delta)$ be the set of all deletion patterns with up to $n\delta$ deletions. For the matching rows $X^n_1$, $Y^k_t$ of $\mathbf{X}$ and $\mathbf{Y}$, define the pairwise adversarial collision probability between $X^n_1$ and $X^n_i$ for any $i\in[m_n]\setminus\{1\}$ as
\iftoggle{singlecolumn}{
\begin{align}
    P_{\text{col,i}}&\triangleq \Pr(\exists \hat{I}_{\text{del}} \in I(\delta):\: \hat{X}_i^{K_n}(\hat{I}_{\text{del}})=Y_t^{K_n})\\
    &=\Pr(\exists \hat{I}_{\text{del}} \in I(\delta):\: \hat{X}_i^{K_n}(\hat{I}_{\text{del}})=\hat{X}_1^{K_n}(\hat{I}_{\text{del}})).
\end{align}
}{
\begin{align}
    P_{\text{col,i}}&\triangleq \Pr(\exists \hat{I}_{\text{del}} \in I(\delta):\: \hat{X}_i^{K_n}(\hat{I}_{\text{del}})=Y_t^{K_n})\\
    &=\Pr(\exists \hat{I}_{\text{del}} \in I(\delta):\: \hat{X}_i^{K_n}(\hat{I}_{\text{del}})=\hat{X}_1^{K_n}(\hat{I}_{\text{del}})).
\end{align}
}
Note that the statement $\exists \hat{I}_{\text{del}} \in I(\delta):\: \hat{X}_i^{K_n}(\hat{I}_{\text{del}})=\hat{X}_1^{K_n}(\hat{I}_{\text{del}})$ is equivalent to the case when the Hamming distance between $X^n_i$ and $X^n_1$ being upper bounded by $n\delta$. In other words,
\iftoggle{singlecolumn}{
\begin{align}
    P_{\text{col,i}} &= \Pr(d_H(X_1^n,X_i^n)\le n\delta)
\end{align}
}{
\begin{align}
    P_{\text{col,i}} &= \Pr(d_H(X_1^n,X_i^n)\le n\delta)
\end{align}
}
where 
\iftoggle{singlecolumn}{
\begin{align}
    d_H(X_1^n,X_i^n) &= \sum\limits_{j=1}^n \mathbbm{1}_{[X_{1,j}\neq X_{i,j}]}
\end{align}
}{
\begin{align}
    d_H(X_1^n,X_i^n) &= \sum\limits_{j=1}^n \mathbbm{1}_{[X_{1,j}\neq X_{i,j}]}
\end{align}
}
Note that due to the \emph{i.i.d.} nature of the database elements, $d_H(X_1^n,X_i^n)\sim \text{Binom}(n,1-\hat{q})$. Thus, for any $\delta\le 1-\hat{q}$, using Chernoff bound~\cite[Lemma 4.7.2]{ash2012information}, we have
\iftoggle{singlecolumn}{
\begin{align}
    P_{\text{col,i}} &= \Pr(d_H(X_1^n,X_i^n)\le n\delta)\\
    &\le 2^{-n D(\delta\|1-\hat{q})}\label{eq:advchernoff}
\end{align}
}{
\begin{align}
    P_{\text{col,i}} &= \Pr(d_H(X_1^n,X_i^n)\le n\delta)\\
    &\le 2^{-n D(\delta\|1-\hat{q})}\label{eq:advchernoff}
\end{align}
}

Therefore given the correct labeling for $Y^k_t\in\mathbf{Y}$ is $X^n_1\in\mathbf{X}$, the probability of error $P_e$ can be bounded as
\iftoggle{singlecolumn}{
\begin{align}
    P_e  &\le \Pr(\exists i\in[m_n]\setminus\{1\}: \hat{X}_i^{K_n}=\hat{X}_1^{K_n})\\
    &\le \sum\limits_{i=2}^{2^{n R}} P_{col,i}+ \kappa_n\\
    &\le  2^{n R} P_{\text{col,2}}+\kappa_n\label{eq:Perowwiseadv}
\end{align}
}{
\begin{align}
    P_e  &\le \Pr(\exists i\in[m_n]\setminus\{1\}: \hat{X}_i^{K_n}=\hat{X}_1^{K_n})\\
    &\le \sum\limits_{i=2}^{2^{n R}} P_{col,i}+ \kappa_n\\
    &\le  2^{n R} P_{\text{col,2}}+\kappa_n\label{eq:Perowwiseadv}
\end{align}
}
where \eqref{eq:Perowwiseadv} follows from the fact the the rows are \emph{i.i.d.} and thus $P_{\text{col,i}}=P_{\text{col,2}},\:\forall i\in[m_n]\setminus\{1\}$. Combining \eqref{eq:advchernoff}-\eqref{eq:Perowwiseadv}, we get
\iftoggle{singlecolumn}{
\begin{align}
    P_e &\le  2^{n R} \Pr(d_H(X_1^n,X_i^n)\le n\delta) + \kappa_n\\
    &\le 2^{n R} 2^{-n D(\delta\|1-\hat{q})}+ \kappa_n\\
    &= 2^{-n \left[D(\delta\|1-\hat{q})-R\right]}+\kappa_n
\end{align}
}{
\begin{align}
    P_e &\le  2^{n R} \Pr(d_H(X_1^n,X_i^n)\le n\delta) + \kappa_n\\
    &\le 2^{n R} 2^{-n D(\delta\|1-\hat{q})}+ \kappa_n\\
    &= 2^{-n \left[D(\delta\|1-\hat{q})-R\right]}+\kappa_n
\end{align}
}
By Lemma~\ref{lem:histogram}, $\kappa_n\to0$ as $n\to \infty$. Thus, we argue that any rate $R$ satisfying 
\iftoggle{singlecolumn}{
\begin{align}
    R&<D(\delta\|1-\hat{q})
\end{align}
}{
\begin{align}
    R&<D(\delta\|1-\hat{q})
\end{align}
}
is achievable. 

Now we prove the converse part. Suppose $P_e\to0$. Then, we have
\iftoggle{singlecolumn}{
\begin{align}
    P_e &= \Pr(\exists i\in[m_n]\setminus\{1\}: d_H(X_1^n,X_i^n)\le n\delta)\\
    &= 1 - \Pr(\forall i\in[m_n]\setminus\{1\}: d_H(X_1^n,X_i^n)> n\delta)\label{eq:advconverse1}\\
    &= 1 -\prod\limits_{i=2}^{m_n} \Pr(d_H(X_1^n,X_i^n)> n\delta)\\
    &=  1 -\prod\limits_{i=2}^{m_n} [1-\Pr(d_H(X_1^n,X_i^n)\le n\delta)]\\
    &= 1-[1-\Pr(d_H(X_1^n,X_2^n)\le n\delta)]^{m_n-1} \label{eq:advconverse2}
\end{align}
}{
\begin{align}
    P_e &= \Pr(\exists i\in[m_n]\setminus\{1\}: d_H(X_1^n,X_i^n)\le n\delta)\\
    &= 1 - \Pr(\forall i\in[m_n]\setminus\{1\}: d_H(X_1^n,X_i^n)> n\delta)\label{eq:advconverse1}\\
    &= 1 -\prod\limits_{i=2}^{m_n} \Pr(d_H(X_1^n,X_i^n)> n\delta)\\
    &=  1 -\prod\limits_{i=2}^{m_n} [1-\Pr(d_H(X_1^n,X_i^n)\le n\delta)]\\
    &= 1-[1-\Pr(d_H(X_1^n,X_2^n)\le n\delta)]^{m_n-1} \label{eq:advconverse2}
\end{align}
}
where \eqref{eq:advconverse1}-\eqref{eq:advconverse2} follow from the \emph{i.i.d.}ness of the rows of $\mathbf{X}$. Since $D_{n,2}\sim\text{Binom}(n,1-\hat{q})$, for ${\delta\le 1-\hat{q}}$, from~\cite[Lemma 4.7.2]{ash2012information}, we obtain
\iftoggle{singlecolumn}{
\begin{align}
    \Pr(D_{n,2}\le n\delta) &\ge \frac{2^{-n D(\delta\|1-\hat{q})}}{\sqrt{2n}}\label{eq:advChernoffLB}
\end{align}
}{
\begin{align}
    \Pr(D_{n,2}\le n\delta) &\ge \frac{2^{-n D(\delta\|1-\hat{q})}}{\sqrt{2n}}\label{eq:advChernoffLB}
\end{align}
}
Plugging \eqref{eq:advChernoffLB} into \eqref{eq:advconverse2}, we get
\iftoggle{singlecolumn}{
\begin{align}
    P_e&\ge 1-\left[1-\frac{2^{-n D(\delta\|1-\hat{q})}}{\sqrt{2n}}\right]^{m_n-1}
\end{align}
}{
\begin{align}
    P_e&\ge 1-\left[1-\frac{2^{-n D(\delta\|1-\hat{q})}}{\sqrt{2n}}\right]^{m_n-1}
\end{align}
}

Now let $y=-\frac{2^{-n D(\delta\|1-\hat{q})}}{\sqrt{2n}}\in(-1,0)$. Then, we get
\iftoggle{singlecolumn}{
\begin{align}
    P_e&\ge 1-(1+y)^{m_n-1}
\end{align}
}{
\begin{align}
    P_e&\ge 1-(1+y)^{m_n-1}
\end{align}
}
Since $y\ge -1$, and $m_n\in\mathbb{N}$, we have
\iftoggle{singlecolumn}{
\begin{align}
    1+y(m_n-1)&\le (1+y)^{m_n-1}\le e^{y (m_n-1)}\label{eq:bernoulli}
\end{align}
}{
\begin{align}
    1+y(m_n-1)&\le (1+y)^{m_n-1}\le e^{y (m_n-1)}\label{eq:bernoulli}
\end{align}
}
where the LHS of \eqref{eq:bernoulli} follows from Bernoulli's inequality~\cite[Theorem 1]{brannan2006first} and the RHS of \eqref{eq:bernoulli} follows from the fact that
\iftoggle{singlecolumn}{
\begin{align}
    \forall x\in\mathbb{R},\hspace{1em} \forall r\in\mathbb{R}_{\ge0} \hspace{1em} (1+x)^r &\le e^{x r}
\end{align}
}{
\begin{align}
    \forall x\in\mathbb{R},\hspace{1em} \forall r\in\mathbb{R}_{\ge0} \hspace{1em} (1+x)^r &\le e^{x r}
\end{align}
}
Thus, we get
\begin{align}
    P_e&\ge 1-(1+y)^{m_n-1}\\
    &\ge 1-e^{y (m_n-1)}\\
    &\ge 0
\end{align}
since $y<0$, $m_n-1>0$. Note that since $P_e\to 0$, by the Squeeze Theorem~\cite[Theorem 2]{brannan2006first}, we have
\iftoggle{singlecolumn}{
\begin{align}
    \lim\limits_{n\to\infty} 1-e^{y (m_n-1)}&\to 0.
\end{align}
}{
\begin{align}
    \lim\limits_{n\to\infty} 1-e^{y (m_n-1)}&\to 0.
\end{align}
}
This, in turn, implies $y m_n\to0$ since the exponential function is continuous everywhere. In other words,
\iftoggle{singlecolumn}{
\begin{align}
    \lim\limits_{n\to\infty}& -\frac{2^{-n D(\delta\|1-\hat{q})}}{\sqrt{2n}} m_n 
    \to0.
\end{align}
}{
\begin{align}
    \lim\limits_{n\to\infty}& -\frac{2^{-n D(\delta\|1-\hat{q})}}{\sqrt{2n}} m_n 
    \to0.
\end{align}
}
Equivalently, from the continuity of the logarithm function, we get
\iftoggle{singlecolumn}{
\begin{align}
    \lim\limits_{n\to\infty}&-n D(\delta\|1-\hat{q})+\log m_n - \frac{1}{2}\log (2 n)\to - \infty\\
    \lim\limits_{n\to\infty}&-n \left[D(\delta\|1-\hat{q})-\frac{1}{n}\log m_n+\frac{\log (2 n)}{2n}\right]\to-\infty\\
    \lim\limits_{n\to\infty}& \left[D(\delta\|1-\hat{q})-\frac{1}{n}\log m_n+\frac{\log (2 n)}{2n}\right]\ge 0
\end{align}
}{
\begin{align}
    \lim\limits_{n\to\infty}&-n D(\delta\|1-\hat{q})+\log m_n - \frac{1}{2}\log (2 n)\to - \infty\\
    \lim\limits_{n\to\infty}&-n \left[D(\delta\|1-\hat{q})-\frac{1}{n}\log m_n+\frac{\log (2 n)}{2n}\right]\to-\infty\\
    \lim\limits_{n\to\infty}& \left[D(\delta\|1-\hat{q})-\frac{1}{n}\log m_n+\frac{\log (2 n)}{2n}\right]\ge 0
\end{align}
}
This implies
\iftoggle{singlecolumn}{
\begin{align}
        D(\delta\|1-\hat{q})&\ge \lim\limits_{n\to\infty} \frac{1}{n}\log m_n\\
        &= R
\end{align}
}{
\begin{align}
        D(\delta\|1-\hat{q})&\ge \lim\limits_{n\to\infty} \frac{1}{n}\log m_n\\
        &= R
\end{align}
}
finishing the proof for $\delta\le 1-\hat{q}$. Thus, we have shown that
\iftoggle{singlecolumn}{
\begin{align}
    C^{\text{adv}}(\delta)&=D(\delta\|1-\hat{q})
\end{align}
}{
\begin{align}
    C^{\text{adv}}(\delta)&=D(\delta\|1-\hat{q})
\end{align}
}
for $\delta\le 1-\hat{q}$. 

We argue that for $\delta>1-\hat{q}$, the adversarial matching capacity is zero, by using two facts: \emph{i)} Since the adversarial deletion budget is an upper bound on deletions, the adversarial matching capacity satisfies
\iftoggle{singlecolumn}{
\begin{align}
    C^{\text{adv}}(\delta)&\le C^{\text{adv}}(\delta^\prime),\hspace{1em}\forall \delta^\prime\le \delta
\end{align}
}{
\begin{align}
    C^{\text{adv}}(\delta)&\le C^{\text{adv}}(\delta^\prime),\hspace{1em}\forall \delta^\prime\le \delta
\end{align}
}
and \emph{ii)} $C^{\text{adv}}(1-\hat{q})=0$. Thus, $\forall \delta>1-\hat{q}$, ${C^{\text{adv}}(\delta)=0}$. This concludes the proof. \qed
\section{Proof of Theorem~\ref{thm:mainresultseedless}}
\label{proof:mainresultseedless}
First, note that the converse part of Theorem~\ref{thm:mainresultseedless} (equation~\eqref{eq:seedlessconverse})  is trivially true since $C(0)$ is a non-decreasing function of the seed size $\Lambda_n$. Hence it is sufficient to prove the achievability part of Theorem~\ref{thm:mainresultseedless} (equation~\eqref{eq:seedlessachievable}). 

For the achievability, we use a matching scheme which \emph{i)} utilizes replica detection and marker addition as done in Section~\ref{subsec:matchingschemeWm} and \emph{ii)} checks the existence of jointly typical subsequences as done in Section~\ref{subsec:achievabilityW1}. The matching scheme we propose is as follows:
\begin{enumerate}[label=\textbf{ \arabic*)},leftmargin=1.3\parindent]
\item Perform replica detection as in Section~\ref{subsec:replicadetection}. The probability of error of this step is denoted by $\rho_n$.
\item Based on the replica detection step, place markers between the noisy replica runs of different columns to obtain $\tilde{\mathbf{Y}}$. Note that at this step we cannot detect runs of length 0 as done in Section~\ref{subsec:matchingschemeWm}. Therefore conditioned on the success of the replica detection we have $\tilde{K}_n=\sum_{j=1}^n \mathbbm{1}_{[S_j\neq0]}$ runs separated with markers.
\item Fix $\epsilon>0$. If $K_n<k\triangleq n(\mathbb{E}[S]-\epsilon)$ or $\hat{K}_n<\hat{k}\triangleq n(1-\delta-\epsilon)$ declare error, whose probability is denoted by $\kappa_n$ where $k$ and $\hat{k}$ are assumed to be integers for computational simplicity. Otherwise, proceed with the next step.
\item Match the $t$\textsuperscript{th} row $Y^{K_n}_t$ of $\mathbf{Y}$ $X^n_{i}$ of $\mathbf{X}$, assigning $\hat{\sigma}_n(i)=t$, if $i$ is the only index in $[m_n]$ such that \emph{i)} $X^n_i$ is  $\epsilon$-typical with respect to $p_X$ and \emph{ii)} $X_i^{n}$ contains a subsequence of length $\tilde{K}_n$, jointly $\epsilon$-typical with $\tilde{Y}^K_t$ with respect to $p_{X,Y,\hat{S}}$ where $\hat{S}\sim p_{\hat{S}}$ with
\iftoggle{singlecolumn}{
\begin{align}
    p_{\hat{S}}(s) &=\begin{cases}
    \frac{p_S(s)}{1-\delta}&\text{if }s\in \{1,\dots,s_{\max}\}\\
    0 & \text{otherwise}
    \end{cases}
\end{align}
}{
\begin{align}
    p_{\hat{S}}(s) &=\begin{cases}
    \frac{p_S(s)}{1-\delta}&\text{if }s\in \{1,\dots,s_{\max}\}\\
    0 & \text{otherwise}
    \end{cases}
\end{align}
}
and
\iftoggle{singlecolumn}{
\begin{align}
    \Pr(Y^S=y^S|X=x,\hat{S}=s)&= \prod\limits_{j=1}^{s} p_{Y|X}(y_j|x).
\end{align}
}{
\begin{align}
    \Pr(Y^S=y^S|X=x,\hat{S}=s)&= \prod\limits_{j=1}^{s} p_{Y|X}(y_j|x).
\end{align}
}
Otherwise, declare a \emph{collision} error.
\end{enumerate}
Since additional runs in $\mathbf{Y}$ and additional columns in each run would decrease the collision probability, we have
\iftoggle{singlecolumn}{
\begin{align}
    \Pr(\text{collision between 1 and }i|K_n&\ge k, \tilde{K}_n\ge \tilde{k})\notag\\&\le \Pr(\text{collision between 1 and }i||K_n=k,\tilde{K}_n= \tilde{k})
\end{align}
}{
\begin{align}
    \Pr&(\text{collision between 1 and }i|K_n\ge k, \tilde{K}_n\ge \tilde{k})\notag\\&\le \Pr(\text{collision between 1 and }i||K_n=k,\tilde{K}_n= \tilde{k})
\end{align}
}
for any $i\in[m_n]\setminus\{1\}$. Thus, for the sake of simplicity, we can focus on the case $K=k$ as it yields an upper bound on the error probability of our matching scheme.

Let $A_\epsilon^{(n)}(X)$ denote the set of $\epsilon$-typical (with respect to $p_X$) sequences of length $n$ and $A_\epsilon(X^{\hat{k}}|Y^k_t,\hat{S}^{\hat{k}})$ denote the set of sequences of length $\hat{k}$ jointly $\epsilon$-typical (with respect to $p_{X,Y,\hat{S}}$) with $Y^k_t$ conditioned on $\hat{S}^n$. For the matching rows $X^n_1$, $Y^k_t$ of $\mathbf{X}$ and $\mathbf{Y}$, define the pairwise collision probability between $X^n_1$ and $X^n_i$ where $i\neq 1$ as
\iftoggle{singlecolumn}{
\begin{align}
    P_{\text{col,i}}\triangleq \Pr(X_{i}^{n}\in A_{\epsilon}^{(n)}(X) \text{ and }\exists {z}^{\hat{k}}\in A_\epsilon(X^{\hat{k}}|Y^k_t,\hat{S}^{\hat{k}}) \text{ which is a subsequence of } X_{i}^{n}.).
\end{align}
}{
\begin{align}
    P_{\text{col,i}}&\triangleq \Pr(X_{i}^{n}\in A_{\epsilon}^{(n)}(X) \text{ and }\exists {z}^{\hat{k}}\in A_\epsilon(X^{\hat{k}}|Y^k_t,\hat{S}^{\hat{k}})\notag\\&\hspace{2.3em} \text{ which is a subsequence of } X_{i}^{n}.).
\end{align}
}
Therefore given the correct labeling for $Y^k_t\in\mathbf{Y}$ is $X^n_1\in\mathbf{X}$, the probability of error $P_e$ can be bounded as
\iftoggle{singlecolumn}{
\begin{align}
    P_e 
    &\le \Pr(\nexists {z}^{\hat{k}}:  {z}^{\hat{k}}\in A_\epsilon(X^{\hat{k}}|Y^k_t,\hat{S}^{\hat{k}}) \text{ and }{z}^{\hat{k}} \text{ is a subsequence of } X_{1}^{n}.)\notag\\
    &\qquad +\Pr(X^n_1\notin A_\epsilon^{(n)}(X))+\sum\limits_{i=2}^{2^{n R}} P_{\text{col,i}}+ \kappa_n+\rho_n\\
    &\le 2\epsilon+\sum\limits_{i=2}^{2^{n R}} P_{\text{col,i}}+\kappa_n+\rho_n\\
    &\le 2\epsilon+ 2^{n R} P_{\text{col,2}}+\kappa_n+\rho_n \label{eq:Perowwiseseedless}
\end{align}
}{
\begin{align}
    P_e 
    &\le \Pr(\nexists {z}^{\hat{k}}:  {z}^{\hat{k}}\in A_\epsilon(X^{\hat{k}}|Y^k_t,\hat{S}^{\hat{k}}) \text{ and }{z}^{\hat{k}}\notag\\&\hspace{7em} \text{ is a subsequence of } X_{1}^{n}.)\notag\\
    &\hspace{0.5em} +\Pr(X^n_1\notin A_\epsilon^{(n)}(X))+\sum\limits_{i=2}^{2^{n R}} P_{\text{col,i}}+ \kappa_n+\rho_n\\
    &\le 2\epsilon+\sum\limits_{i=2}^{2^{n R}} P_{\text{col,i}}+\kappa_n+\rho_n\\
    &\le 2\epsilon+ 2^{n R} P_{\text{col,2}}+\kappa_n+\rho_n \label{eq:Perowwiseseedless}
\end{align}
}
where \eqref{eq:Perowwiseseedless} follows from the fact the the rows are \emph{i.i.d.} and thus $P_{\text{col,i}}=P_{\text{col,2}}$, $\forall i\in[m_n]\setminus\{1\}$.

We now upper bound $P_{\text{col,2}}$. For any $z^{\hat{k}}$ define
\iftoggle{singlecolumn}{
\begin{align}
    T(z^{\hat{k}})&\triangleq \{x^n\in\mathfrak{X}^n:x^{n}\in A_\epsilon^{(n)}(X) \text{, }x^{n}  \text{ contains } z^{\hat{k}}.\}.
\end{align}
}{
\begin{align}
    T(z^{\hat{k}})&\triangleq \{x^n\in\mathfrak{X}^n:x^{n}\in A_\epsilon^{(n)}(X) \text{, }x^{n}  \text{ contains } z^{\hat{k}}.\}.
\end{align}
}
Observe that for any $z^{\hat{k}}\in A_\epsilon(X^{\hat{k}}|Y^k_t,\hat{S}^{\hat{k}})$, we have $z^{\hat{k}}\in A_\epsilon^{(\hat{k})}(X)$. Furthermore, for a given deletion pattern with $n-\hat{k}=\Theta(n)$ deletions, WLOG $(\hat{k}+1,\dots,n)$, the $\epsilon$-typicality of $x^n=(x_1,\dots,x_n)$ and $z^{\hat{k}}=(x_1,\dots,x_{\hat{k}})$ with respect to $p_X$ implies the $\tilde{\epsilon}$-typicality of $(x_{\hat{k}+1},\dots,x_n)$, where $\tilde{\epsilon}=\frac{2-\delta-\epsilon}{\delta+\epsilon}\epsilon$. Therefore for any $z^{\hat{k}}\in A_\epsilon(X^{\hat{k}}|Y^k_t,\hat{S}^{\hat{k}})$, taking the union bound over all possible deletion patterns with $n-\hat{k}$ deletions, the cardinality of $T(z^{\hat{k}})$ can be upper bounded as
\iftoggle{singlecolumn}{
\begin{align}
    |T(z^{\hat{k}})| &\le \binom{n}{\hat{k}} |A_{\tilde{\epsilon}}^{(n-\hat{k})}(X)|\\
    &\le 2^{n H_b(\frac{\hat{k}}{n})} |A_{\tilde{\epsilon}}^{(n-\hat{k})}(X)|\\
    &\le 2^{n H_b(\frac{\hat{k}}{n})} 2^{(n-\hat{k})(H(X)+\tilde{\epsilon})}\\
    &=  2^{n \left[H_b(\frac{\hat{k}}{n})+(1-\frac{\hat{k}}{n})(H(X)+\tilde{\epsilon})\right]}\label{eq:tsettypicalbound}
\end{align}
}{
\begin{align}
    |T(z^{\hat{k}})| &\le \binom{n}{\hat{k}} |A_{\tilde{\epsilon}}^{(n-\hat{k})}(X)|\\
    &\le 2^{n H_b(\frac{\hat{k}}{n})} |A_{\tilde{\epsilon}}^{(n-\hat{k})}(X)|\\
    &\le 2^{n H_b(\frac{\hat{k}}{n})} 2^{(n-\hat{k})(H(X)+\tilde{\epsilon})}\\
    &=  2^{n \left[H_b(\frac{\hat{k}}{n})+(1-\frac{\hat{k}}{n})(H(X)+\tilde{\epsilon})\right]}\label{eq:tsettypicalbound}
\end{align}
}

Furthermore, for any $x^n\in T(z^{\hat{k}})$, since ${T(z^{\hat{k}})\subseteq A_\epsilon^{(n)}(X)}$, we have
\iftoggle{singlecolumn}{
\begin{align}
    p_{X^n}(x^n)\le 2^{-n(H(X)-\epsilon)}
\end{align}
}{
\begin{align}
    p_{X^n}(x^n)\le 2^{-n(H(X)-\epsilon)}
\end{align}
}
and since the rows $X^n_i$ of $\mathbf{X}$ are \emph{i.i.d.}, we have 
\iftoggle{singlecolumn}{
\begin{align}
    \Pr(X^n_2 \in T(z^{\hat{k}})|X^n_1 \in T(z^{\hat{k}}))=\Pr(X^n_2 \in T(z^{\hat{k}})).
\end{align}
}{
\begin{align}
    \Pr(X^n_2 \in T(z^{\hat{k}})|X^n_1 \in T(z^{\hat{k}}))=\Pr(X^n_2 \in T(z^{\hat{k}})).
\end{align}
}
Finally, we note that
\iftoggle{singlecolumn}{
\begin{align}
    |A_\epsilon(X^{\hat{k}}|Y^k_t,\hat{S}^{\hat{k}})| &\le 2^{\hat{k}(H(X|Y^{\hat{S}},\hat{S})+\epsilon)}
\end{align}
}{
\begin{align}
    |A_\epsilon(X^{\hat{k}}|Y^k_t,\hat{S}^{\hat{k}})| &\le 2^{\hat{k}(H(X|Y^{\hat{S}},\hat{S})+\epsilon)}
\end{align}
}
and 
\iftoggle{singlecolumn}{
\begin{align}
    H(X|Y^{\hat{S}},\hat{S})&=\sum\limits_{s=1}^{s_{\max}} p_{\hat{S}}(s) H(X|Y^{\hat{S}},\hat{S}=s)\\
    &=\frac{1}{1-\delta} \sum\limits_{s=1}^{s_{\max}} p_{S}(s) H(X|Y^{\hat{S}},\hat{S}=s)\\
    &=\frac{1}{1-\delta}\left[\sum\limits_{s=0}^{s_{\max}} p_{S}(s) H(X|Y^{\hat{S}},\hat{S}=s)-\delta H(X|Y^{S},S=0)\right]\\
    &=\frac{1}{1-\delta} \left[H(X|Y^S,S)-\delta H(X)\right] \\
    &=\frac{1}{1-\delta} \left[(1-\delta) H(X)-I(X;Y^S,S)\right]\\
    &= H(X)-\frac{I(X;Y^S,S)}{1-\delta}
\end{align}
}{
\begin{align}
    H(X|Y^{\hat{S}},\hat{S})&=\sum\limits_{s=1}^{s_{\max}} p_{\hat{S}}(s) H(X|Y^{\hat{S}},\hat{S}=s)\\
    &=\frac{1}{1-\delta} \sum\limits_{s=1}^{s_{\max}} p_{S}(s) H(X|Y^{\hat{S}},\hat{S}=s)\\
    &=\frac{1}{1-\delta}\Big[\sum\limits_{s=0}^{s_{\max}} p_{S}(s) H(X|Y^{\hat{S}},\hat{S}=s)\notag\\&\hspace{4em}-\delta H(X|Y^{S},S=0)\Big]\\
    &=\frac{1}{1-\delta} \left[H(X|Y^S,S)-\delta H(X)\right] \\
    &=\frac{1}{1-\delta} \left[(1-\delta) H(X)-I(X;Y^S,S)\right]\\
    &= H(X)-\frac{I(X;Y^S,S)}{1-\delta}
\end{align}
}
Thus, we get
\iftoggle{singlecolumn}{
\begin{align}
    |A_\epsilon(X^{\hat{k}}|Y^k_t,\hat{S}^{\hat{k}})|&\le 2^{\hat{k}\left[H(X)-\frac{I(X;Y^S,S)}{1-\delta} +\epsilon\right]}.\label{eq:condtypicalset2}
\end{align}
}{
\begin{align}
    |A_\epsilon(X^{\hat{k}}|Y^k_t,\hat{S}^{\hat{k}})|&\le 2^{\hat{k}\left[H(X)-\frac{I(X;Y^S,S)}{1-\delta} +\epsilon\right]}.\label{eq:condtypicalset2}
\end{align}
}
\begingroup
\allowdisplaybreaks
Combining \eqref{eq:tsettypicalbound}-\eqref{eq:condtypicalset2}, we can upper bound $P_{\text{col,2}}$ as
\iftoggle{singlecolumn}{
\begin{align}
    P_{\text{col,2}} &\le \sum\limits_{z^{\hat{k}}\in A_\epsilon(X^{\hat{k}}|Y^k_t,\hat{S}^{\hat{k}})} \Pr(X^n_2 \in T(z^{\hat{k}}))\\
    &=\sum\limits_{z^{\hat{k}}\in A_\epsilon(X^{\hat{k}}|Y^k_t,\hat{S}^{\hat{k}})} \sum\limits_{x^n\in T(z^{\hat{k}})} p_{X^n}(x^n)\\
    &\le \sum\limits_{z^{\hat{k}}\in A_\epsilon(X^{\hat{k}}|Y^k_t,\hat{S}^{\hat{k}})} \sum\limits_{x^n\in T(z^{\hat{k}})} 2^{-n(H(X)-\epsilon)}\\
    &= \sum\limits_{z^{\hat{k}}\in A_\epsilon(X^{\hat{k}}|Y^k_t,\hat{S}^{\hat{k}})} |T(z^{\hat{k}})| 2^{-n(H(X)-\epsilon)}\label{eq:tsetseparation2}\\
    &\le  \sum\limits_{z^{\hat{k}}\in A_\epsilon(X^{\hat{k}}|Y^k_t,\hat{S}^{\hat{k}})} 2^{-n(H(X)-\epsilon)} 2^{n \left[H_b(\frac{\hat{k}}{n})+(1-\frac{\hat{k}}{n})(H(X)+\tilde{\epsilon})\right]}\\
    &= |A_\epsilon(X^{\hat{k}}|Y^k_t,\hat{S}^{\hat{k}})| 2^{-\left[\hat{k}H(X)- n \epsilon-H_b(\frac{\hat{k}}{n})-(n-\hat{k})\tilde{\epsilon}\right]} \\
    &\le 2^{\hat{k}\left[H(X)-\frac{I(X;Y^S,S)}{1-\delta} +\epsilon\right]} 2^{-\left[\hat{k}H(X)- n \epsilon-n H_b(\frac{\hat{k}}{n})-(n-\hat{k})\tilde{\epsilon}\right]}\\
    &= 2^{-n\left[\frac{1-\delta-\epsilon}{1-\delta} I(X;Y^S,S)-H_b(\delta+\epsilon)-(\delta+\epsilon)(\epsilon+\tilde{\epsilon}) \right]}\\
    &=2^{-n\left[\frac{1-\delta-\epsilon}{1-\delta} I(X;Y^S,S)-H_b(\delta+\epsilon)-2\epsilon \right]}
\end{align}
}{
\begin{align}
    P_{\text{col,2}} &\le \sum\limits_{z^{\hat{k}}\in A_\epsilon(Z^{\hat{k}}|Y^k_t,\hat{S}^{\hat{k}})} \Pr(X^n_2 \in T(z^{\hat{k}}))\\
    &=\sum\limits_{z^{\hat{k}}\in A_\epsilon(Z^{\hat{k}}|Y^k_t,\hat{S}^{\hat{k}})} \sum\limits_{x^n\in T(z^{\hat{k}})} p_{X^n}(x^n)\\
    &\le \sum\limits_{z^{\hat{k}}\in A_\epsilon(Z^{\hat{k}}|Y^k_t,\hat{S}^{\hat{k}})} \sum\limits_{x^n\in T(z^{\hat{k}})} 2^{-n(H(X)-\epsilon)}\\
    &= \sum\limits_{z^{\hat{k}}\in A_\epsilon(Z^{\hat{k}}|Y^k_t,\hat{S}^{\hat{k}})} |T(z^{\hat{k}})| 2^{-n(H(X)-\epsilon)}\label{eq:tsetseparation2}\\
    &\le  \sum\limits_{z^{\hat{k}}\in A_\epsilon(Z^{\hat{k}}|Y^k_t,\hat{S}^{\hat{k}})}  2^{-n(H(X)-\epsilon)} \notag\\&\hspace{6.2em} 2^{n \left[H_b(\frac{\hat{k}}{n})+(1-\frac{\hat{k}}{n})(H(X)+\tilde{\epsilon})\right]}\\
    &= |A_\epsilon(Z^{\hat{k}}|Y^k_t,\hat{S}^{\hat{k}})| 2^{-\left[\hat{k}H(X)- n \epsilon-H_b(\frac{\hat{k}}{n})-(n-\hat{k})\tilde{\epsilon}\right]} \\
    &\le 2^{\hat{k}\left[H(X)-\frac{I(X;Y^S,S)}{1-\delta} +\epsilon\right]}\notag\\&\hspace{4em} 2^{-\left[\hat{k}H(X)- n \epsilon-n H_b(\frac{\hat{k}}{n})-(n-\hat{k})\tilde{\epsilon}\right]}\\
    &= 2^{-n\left[\frac{1-\delta-\epsilon}{1-\delta} I(X;Y^S,S)-H_b(\delta+\epsilon)-(\delta+\epsilon)(\epsilon+\tilde{\epsilon}) \right]}\\
    &=2^{-n\left[\frac{1-\delta-\epsilon}{1-\delta} I(X;Y^S,S)-H_b(\delta+\epsilon)-2\epsilon \right]}
\end{align}
}
\endgroup
Thus, we have the following upper bound on the error probability
\iftoggle{singlecolumn}{
\begin{align}
    P_e &\le 2\epsilon+ 2^{n R} 2^{-n\left[\frac{1-\delta-\epsilon}{1-\delta} I(X;Y^S,S)-H_b(\delta+\epsilon)-2\epsilon \right]} +\kappa_n+\rho_n
\end{align}
}{
\begin{align}
    P_e &\le 2\epsilon+ 2^{n R} 2^{-n\left[\frac{1-\delta-\epsilon}{1-\delta} I(X;Y^S,S)-H_b(\delta+\epsilon)-2\epsilon \right]} \notag\\&\hspace{1em}+\kappa_n+\rho_n
\end{align}
}
By LLN, we have $\kappa_n\to0$ and from Lemma~\ref{lem:noisyreplicadetection}, we have $ \rho_n\to0$ as $n\to\infty$. Hence, we can argue that any database growth rate $R$ satisfying 
\iftoggle{singlecolumn}{
\begin{align}
    R<I(X;Y^S,S)-H_b(\delta)
\end{align}
}{
\begin{align}
    R<I(X;Y^S,S)-H_b(\delta)
\end{align}
}
is achievable by taking $\epsilon$ small enough.

Now, we investigate repetition distributions with $\delta\le 1-\frac{1}{|\mathfrak{X}|}$. Recall from Appendix~\ref{proof:achievabilityW1} the counting function $F(n,\hat{k},|\mathfrak{X}|)$ denoting the number of $|\mathfrak{X}|$-ary sequences of length $n$, which contain a fixed $|\mathfrak{X}|$-ary sequence of length $\hat{k}$ as a subsequence. From~\cite{chvatal1975longest,diggavi1603788}, we have
\iftoggle{singlecolumn}{
\begin{align}
    F(n,\hat{k},|\mathfrak{X}|)&\le n 2^{n\left[ H_b\left(\frac{\hat{k}}{n}\right)+(1-\frac{\hat{k}}{n})\log(|\mathfrak{X}|-1))\right]}.
\end{align}
}{
\begin{align}
    F(n,\hat{k},|\mathfrak{X}|)&\le n 2^{n\left[ H_b\left(\frac{\hat{k}}{n}\right)+(1-\frac{\hat{k}}{n})\log(|\mathfrak{X}|-1))\right]}.
\end{align}
}
Furthermore, disregarding the typicality constraint, we can trivially bound the cardinality of $T(z^{\hat{k}})$ as
\iftoggle{singlecolumn}{
\begin{align}
    |T(z^{\hat{k}})|&\le |\{x^n\in\mathfrak{X}^n: x^{n}  \text{ contains } z^{\hat{k}} \}|\\
    &\le F(n,\hat{k},|\mathfrak{X}|)\\
    &\le n 2^{n\left[ H_b\left(\frac{\hat{k}}{n}\right)+(1-\frac{\hat{k}}{n}\log(|\mathfrak{X}|-1))\right]}\label{eq:tsettrivialbound2}
\end{align}
}{
\begin{align}
    |T(z^{\hat{k}})|&\le |\{x^n\in\mathfrak{X}^n: x^{n}  \text{ contains } z^{\hat{k}} \}|\\
    &\le F(n,\hat{k},|\mathfrak{X}|)\\
    &\le n 2^{n\left[ H_b\left(\frac{\hat{k}}{n}\right)+(1-\frac{\hat{k}}{n}\log(|\mathfrak{X}|-1))\right]}\label{eq:tsettrivialbound2}
\end{align}
}
Plugging \eqref{eq:tsettrivialbound2} into \eqref{eq:tsetseparation2} and following the same steps, one can show that any rate $R$ satisfying 
\iftoggle{singlecolumn}{
\begin{align}
    R&<\left[I(X;Y^S,S)+\delta(H(X)-\log(|\mathfrak{X}|-1))-H_b(\delta)\right]^+
\end{align}
}{
\begin{align}
    R&<\Big[I(X;Y^S,S)-H_b(\delta)\notag\\&\hspace{5em}+\delta(H(X)-\log(|\mathfrak{X}|-1))\Big]^+
\end{align}
}
is achievable. Simply taking the maximum of the two proven achievable rates when $\delta\le 1-\nicefrac{1}{|\mathfrak{X}|}$ yields the desired achievability result. This concludes the proof. \qed


\section{Proof of Lemma~\ref{lem:histogramuncollapsed}}\label{proof:histogramuncollapsed}
For brevity, we let $\mu_n$ denote ${\Pr(\exists i,j\in [n],\: i\neq j,H^{(1)}_i=H^{(1)}_j)}$. Notice that since the entries of $\mathbf{X}$ are \emph{i.i.d.}, $H^{(1)}_i$ are \emph{i.i.d.} Multinomial$(m_n,p_X)$ random variables. Then,
\iftoggle{singlecolumn}{
\begin{align}
    \mu_n&\le n^2 \Pr(H^{(1)}_1=H^{(1)}_2)\\
        &=n^2 \sum\limits_{h^{|\mathfrak{X}|}} \Pr(H^{(1)}_1=h^{|\mathfrak{X}|})^2
\end{align}
}{
\begin{align}
    \mu_n&\le n^2 \Pr(H^{(1)}_1=H^{(1)}_2)\\
        &=n^2 \sum\limits_{h^{|\mathfrak{X}|}} \Pr(H^{(1)}_1=h^{|\mathfrak{X}|})^2
\end{align}
}
where the sum is over all vectors of length $|\mathfrak{X}|$, summing up to $m_n$. Let $m_i\triangleq h(i)$, $\forall i\in\mathfrak{X}$. Then,
\iftoggle{singlecolumn}{
\begin{align}
    \Pr(H^{(1)}_1=h^{|\mathfrak{X}|})&= \binom{m_n}{m_1,m_2,\dots,m_{|\mathfrak{X}|}} \prod\limits_{i=1}^{|\mathfrak{X}|} p_X(i)^{m_i}
\end{align}
}{
\begin{align}
    \Pr(H^{(1)}_1=h^{|\mathfrak{X}|})&= \binom{m_n}{m_1,m_2,\dots,m_{|\mathfrak{X}|}} \prod\limits_{i=1}^{|\mathfrak{X}|} p_X(i)^{m_i}
\end{align}
}
Hence, we have
\iftoggle{singlecolumn}{
\begin{align}
    \mu_n &\le n^2\sum\limits_{m_1+\dots+m_{|\mathfrak{X}|}=m_n} \binom{m_n}{m_1,m_2,\dots,m_{|\mathfrak{X}|}}^2 \prod\limits_{i=1}^{|\mathfrak{X}|} p_X(i)^{2 m_i}\label{eq:multinomial}
\end{align}
}{
\begin{align}
    \mu_n &\le n^2\sum\limits_{m_1+\dots+m_{|\mathfrak{X}|}=m_n} \binom{m_n}{m_1,m_2,\dots,m_{|\mathfrak{X}|}}^2 \notag\\&\hspace{8em}\prod\limits_{i=1}^{|\mathfrak{X}|} p_X(i)^{2 m_i}\label{eq:multinomial}
\end{align}
}
where $\smash{\binom{m_n}{m_1,m_2,\dots,m_{|\mathfrak{X}|}}}$ is the multinomial coefficient corresponding to the $|\mathfrak{X}|$-tuple $(m_1,\dots,m_{|\mathfrak{X}|})$ and the summation is over all possible non-negative indices $m_1,\dots,m_{|\mathfrak{X}|}$ which add up to $m_n$.

From~\cite[Theorem 11.1.2]{cover2006elements}, we have
\iftoggle{singlecolumn}{
\begin{align}
    \prod\limits_{i=1}^{|\mathfrak{X}|} p_X(i)^{2 m_i}=2^{-2m_n(H(\Tilde{p})+D(\Tilde{p}\|p_X))}\label{eq:covertype}
\end{align}
}{
\begin{align}
    \prod\limits_{i=1}^{|\mathfrak{X}|} p_X(i)^{2 m_i}=2^{-2m_n(H(\Tilde{p})+D(\Tilde{p}\|p_X))}\label{eq:covertype}
\end{align}
}
where $\Tilde{p}$ is the type corresponding to $|\mathfrak{X}|$-tuple ${(m_1,\dots,m_{|\mathfrak{X}|})}$: 
\iftoggle{singlecolumn}{
\begin{align}
    \Tilde{p}&=\left(\frac{m_1}{m_n},\dots,\frac{m_{|\mathfrak{X}|}}{m_n}\right).
\end{align}
}{
\begin{align}
    \Tilde{p}&=\left(\frac{m_1}{m_n},\dots,\frac{m_{|\mathfrak{X}|}}{m_n}\right).
\end{align}
}
From Stirling's approximation~\cite[Chapter 3.2]{cormen2022introduction}, we get
\iftoggle{singlecolumn}{
\begin{align}
    \binom{m_n}{m_1,m_2,\dots,m_{|\mathfrak{X}|}}^2\le \frac{e^2} {(2\pi)^{|\mathfrak{X}|}} m_n^{1-|\mathfrak{X}|} \Pi_{\Tilde{p}}^{-1} 2^{2m_n H(\Tilde{p})}\label{eq:stirling}
\end{align}
}{
\begin{align}
    \binom{m_n}{m_1,m_2,\dots,m_{|\mathfrak{X}|}}^2\le \frac{e^2} {(2\pi)^{|\mathfrak{X}|}} m_n^{1-|\mathfrak{X}|} \Pi_{\Tilde{p}}^{-1} 2^{2m_n H(\Tilde{p})}\label{eq:stirling}
\end{align}
}
where $\Pi_{\Tilde{p}}=\prod_{i=1}^{|\mathfrak{X}|} \Tilde{p}(i)$.

Combining \eqref{eq:multinomial}-\eqref{eq:stirling}, we get
\iftoggle{singlecolumn}{
\begin{align}
    \mu_n\le \frac{e^2} {(2\pi)^{|\mathfrak{X}|}} n^2 m_n^{1-|\mathfrak{X}|} \sum\limits_{\Tilde{p}} \Pi_{\Tilde{p}}^{-1} 2^{-2m_n D_{KL}(\Tilde{p}\|p_X)}
\end{align}
}{
\begin{align}
    \mu_n\le \frac{e^2} {(2\pi)^{|\mathfrak{X}|}} n^2 m_n^{1-|\mathfrak{X}|} \sum\limits_{\Tilde{p}} \Pi_{\Tilde{p}}^{-1} 2^{-2m_n D_{KL}(\Tilde{p}\|p_X)}
\end{align}
}

Let 
\iftoggle{singlecolumn}{
\begin{align}
    \tilde{T}=\sum\limits_{\Tilde{p}} \Pi_{\Tilde{p}}^{-1} 2^{-2m_n D_{KL}(\Tilde{p}\|p_X)} = \tilde{T}_1 + \tilde{T}_2
\end{align} where
\begin{align}
    \tilde{T}_1&=\sum\limits_{\Tilde{p}:D_{KL}(\Tilde{p}\|p_X)> \frac{\epsilon_n^2}{2\log_e 2}} \Pi_{\Tilde{p}}^{-1} 2^{-2m_n D_{KL}(\Tilde{p}\|p_X)}\label{eq:T1iid}\\
    \tilde{T}_2&=\sum\limits_{\Tilde{p}:D_{KL}(\Tilde{p}\|p_X)\le\frac{\epsilon_n^2}{2\log_e 2}} \Pi_{\Tilde{p}}^{-1} 2^{-2m_n D_{KL}(\Tilde{p}\|p_X)},\label{eq:T2iid}
\end{align}
}{
\begin{align}
    \tilde{T}=\sum\limits_{\Tilde{p}} \Pi_{\Tilde{p}}^{-1} 2^{-2m_n D_{KL}(\Tilde{p}\|p_X)} = \tilde{T}_1 + \tilde{T}_2
\end{align} where
\begin{align}
    \tilde{T}_1&=\sum\limits_{\Tilde{p}:D_{KL}(\Tilde{p}\|p_X)> \frac{\epsilon_n^2}{2\log_e 2}} \Pi_{\Tilde{p}}^{-1} 2^{-2m_n D_{KL}(\Tilde{p}\|p_X)}\label{eq:T1iid}\\
    \tilde{T}_2&=\sum\limits_{\Tilde{p}:D_{KL}(\Tilde{p}\|p_X)\le\frac{\epsilon_n^2}{2\log_e 2}} \Pi_{\Tilde{p}}^{-1} 2^{-2m_n D_{KL}(\Tilde{p}\|p_X)},\label{eq:T2iid}
\end{align}
}
$\epsilon_n$, which is described below in more detail, is a small positive number decaying with $n$.

First, we look at $\tilde{T}_2$. From Pinsker's inequality~\cite[Lemma 11.6.1]{cover2006elements}, we have
\iftoggle{singlecolumn}{
\begin{align}
    D_{KL}(\Tilde{p}\|p_X)\le \frac{\epsilon_n^2}{2\log_e 2}\Rightarrow \mathbb{V}(\Tilde{p},p_X)\le \epsilon_n
\end{align}
}{
\begin{align}
    D_{KL}(\Tilde{p}\|p_X)\le \frac{\epsilon_n^2}{2\log_e 2}\Longrightarrow \mathbb{V}(\Tilde{p},p_X)\le \epsilon_n
\end{align}
}
where $\mathbb{V}$ denotes the (unnormalized) total variation distance. Therefore
\iftoggle{singlecolumn}{
\begin{align}
    \left|\{\Tilde{p}:D_{KL}(\Tilde{p}\|p_X)\le \frac{\epsilon_n^2}{2\log_e 2}\}\right|&\le |\{\Tilde{p}:\mathbb{V}(\Tilde{p},p_X)\le \epsilon_n\}| \\
    &= O(m_n^{|\mathfrak{X}|-1}\epsilon_n^{|\mathfrak{X}|-1})
\end{align}
}{
\begin{align}
    \Big|\{\Tilde{p}:D_{KL}(\Tilde{p}\|p_X)\le &\frac{\epsilon_n^2}{2\log_e 2}\}\Big|\notag\\&\le |\{\Tilde{p}:\mathbb{V}(\Tilde{p},p_X)\le \epsilon_n\}| \\
    &= O(m_n^{|\mathfrak{X}|-1}\epsilon_n^{|\mathfrak{X}|-1})
\end{align}
}
where the last equality follows from the fact in a type we have $|\mathfrak{X}|-1$ degrees of freedom, since the sum of the $|\mathfrak{X}|$-tuple $(m_1,\dots,m_{|\mathfrak{X}|})$ is fixed. 
Furthermore, when $\mathbb{V}(\Tilde{p},p_X)\le \epsilon_n$, we have 
\iftoggle{singlecolumn}{
\begin{align}
    \Pi_{\Tilde{p}} &\ge \prod\limits_{i=1}^{|\mathfrak{X}|} (p_X(i)-\epsilon_n)\ge \Pi_{p_X}-\epsilon_n \sum\limits_{i=1}^{|\mathfrak{X}|} \prod\limits_{j\neq i} p_X(j)
\end{align}
}{
\begin{align}
    \Pi_{\Tilde{p}} &\ge \prod\limits_{i=1}^{|\mathfrak{X}|} (p_X(i)-\epsilon_n)\ge \Pi_{p_X}-\epsilon_n \sum\limits_{i=1}^{|\mathfrak{X}|} \prod\limits_{j\neq i} p_X(j)
\end{align}
}
Hence
\iftoggle{singlecolumn}{
\begin{align}
    \Pi_{\Tilde{p}}^{-1}&\le \frac{1}{\Pi_{p_X}-\epsilon_n \sum\limits_{i=1}^{|\mathfrak{X}|} \prod\limits_{j\neq i} p_X(j)}
\end{align}
}{
\begin{align}
    \Pi_{\Tilde{p}}^{-1}&\le \frac{1}{\Pi_{p_X}-\epsilon_n \sum\limits_{i=1}^{|\mathfrak{X}|} \prod\limits_{j\neq i} p_X(j)}
\end{align}
}
and 
\iftoggle{singlecolumn}{
\begin{align}
    \tilde{T}_2 &\le \frac{1}{\Pi_{p_X}-\epsilon_n \sum\limits_{i=1}^{|\mathfrak{X}|} \prod\limits_{j\neq i} p_X(j)} O(m_n^{|\mathfrak{X}|-1}\epsilon_n ^{|\mathfrak{X}|-1})\\
    &= O(m_n^{|\mathfrak{X}|-1}\epsilon_n^{|\mathfrak{X}|-1})
\end{align}
}{
\begin{align}
    \tilde{T}_2 &\le \frac{1}{\Pi_{p_X}-\epsilon_n \sum\limits_{i=1}^{|\mathfrak{X}|} \prod\limits_{j\neq i} p_X(j)} O(m_n^{|\mathfrak{X}|-1}\epsilon_n ^{|\mathfrak{X}|-1})\\
    &= O(m_n^{|\mathfrak{X}|-1}\epsilon_n^{|\mathfrak{X}|-1})
\end{align}
}
for small $\epsilon_n$.
    
Now, we look at $\tilde{T}_1$. Note that since $m_i\in \mathbb{Z}_+$, we have ${\Pi_{\Tilde{p}}\le m_n^{|\mathfrak{X}|}}$, suggesting the multiplicative term in the summation in~\eqref{eq:T1iid} is polynomial with $m_n$. If $m_i=0$ we can simply discard it and return to Stirling's approximation with the reduced number of categories. Furthermore, from~\cite[Theorem 11.1.1]{cover2006elements}, we have
\iftoggle{singlecolumn}{
\begin{align}
    \left|\{\Tilde{p}:D_{KL}(\Tilde{p}\|p_X)> \frac{\epsilon_n^2}{2\log_e 2}\}\right|&\le |\{\Tilde{p}\}|\\&\le (m_n+1)^{|\mathfrak{X}|}
\end{align}
}{
\begin{align}
    \left|\{\Tilde{p}:D_{KL}(\Tilde{p}\|p_X)> \frac{\epsilon_n^2}{2\log_e 2}\}\right|&\le |\{\Tilde{p}\}|\\&\le (m_n+1)^{|\mathfrak{X}|}
\end{align}
}
suggesting the number of terms which we take the summation over in~\eqref{eq:T1iid} is polynomial with $m_n$ as well. Therefore, as long as ${m_n \epsilon_n^2\to\infty}$, $\tilde{T}_1$ has a polynomial number of elements that decay exponentially with $m_n$. Thus
\iftoggle{singlecolumn}{
\begin{align}
    \tilde{T}_1\to0\text{ as }n\to\infty.\label{eq:t1iid}
\end{align}
}{
\begin{align}
    \tilde{T}_1\to0\text{ as }n\to\infty.\label{eq:t1iid}
\end{align}
}
    
Define 
\iftoggle{singlecolumn}{
\begin{align}
    U_i&=e^2 (2\pi)^{-|\mathfrak{X}|} m_n^{1-|\mathfrak{X}|} \tilde{T}_i,\quad i=1,2\label{eq:ui}
\end{align}
}{
\begin{align}
    U_i&=e^2 (2\pi)^{-|\mathfrak{X}|} m_n^{1-|\mathfrak{X}|} \tilde{T}_i,\quad i=1,2\label{eq:ui}
\end{align}
}
and choose ${\epsilon_n=m_n^{-\frac{1}{2}} V_n}$ for some $V_n$ satisfying ${V_n=\omega(1)}$ and ${V_n=o(m_n^{1/2})}$. Thus, $U_1$ vanishes exponentially fast since ${m_n\epsilon_n^2=V_n^2\to\infty}$ and \iftoggle{singlecolumn}{
\begin{align}
    U_2&=O(\epsilon_n^{|\mathfrak{X}|-1})=O(m_n^{(1-|\mathfrak{X}|)/2} V_n^{(|\mathfrak{X}|-1)}).\label{eq:u2}
\end{align}
}{
\begin{align}
    U_2&=O(\epsilon_n^{|\mathfrak{X}|-1})=O(m_n^{(1-|\mathfrak{X}|)/2} V_n^{(|\mathfrak{X}|-1)}).\label{eq:u2}
\end{align}
}
Combining \eqref{eq:t1iid}-\eqref{eq:u2}, we have 
\iftoggle{singlecolumn}{
\begin{align}
    U=U_1+U_2=O(m_n^{(1-|\mathfrak{X}|)/2} V_n^{(|\mathfrak{X}|-1)})
\end{align}
}{
\begin{align}
    U=U_1+U_2=O(m_n^{(1-|\mathfrak{X}|)/2} V_n^{(|\mathfrak{X}|-1)})
\end{align}
}
and we get 
\iftoggle{singlecolumn}{
\begin{align}
    \mu_n\le n^2 O(m_n^{(1-|\mathfrak{X}|)/2} V_n^{(|\mathfrak{X}|-1)})
\end{align}
}{
\begin{align}
    \mu_n\le n^2 O(m_n^{(1-|\mathfrak{X}|)/2} V_n^{(|\mathfrak{X}|-1)})
\end{align}
}
By the assumption ${m=\omega(n^\frac{4}{|\mathfrak{X}|-1})}$, we have ${m_n=n^\frac{4}{|\mathfrak{X}|-1} Z_n}$ for some $Z_n$ satisfying  ${\lim\limits_{n\to\infty} Z_n=\infty}$. Now, taking ${V_n=o(Z_n^{1/2})}$ (e.g.~$V_n=Z_n^{1/3}$), we get
\iftoggle{singlecolumn}{
\begin{align}
    \mu_n&\le O(n^2 n^{-2} Z_n^{(1-|\mathfrak{X}|)/2} V_n^{(|\mathfrak{X}|-1)})
    = o(1)
\end{align}
}{
\begin{align}
    \mu_n&\le O(n^2 n^{-2} Z_n^{(1-|\mathfrak{X}|)/2} V_n^{(|\mathfrak{X}|-1)})
    = o(1).
\end{align}
}
Thus $m_n=\omega(n^\frac{4}{|\mathfrak{X}|-1})$ is enough to have $\mu_n\to0$ as $n\to\infty$. \qed
\section{Proof of Proposition~\ref{prop:uniformhistogram}}\label{proof:uniformhistogram}
For brevity, we let $\mu_n$ denote ${\Pr(\exists i,j\in [n],\: i\neq j,H^{(1)}_i=H^{(1)}_j)}$. Then,
\iftoggle{singlecolumn}{
\begin{align}
    \mu_n&= n(n-1)\Pr(H^{(1)}_1=H^{(1)}_2)\\
    &= n(n-1)\sum\limits_{h^{|\mathfrak{X}|}} \Pr(H^{(1)}_1=h^{|\mathfrak{X}|})^2\\
    &=n(n-1)\sum\limits_{m_1+\dots+m_{|\mathfrak{X}|}=m_n} \binom{m_n}{m_1,\dots,m_{|\mathfrak{X}|}}^2 |\mathfrak{X}|^{-2m_n}\\
    &= n(n-1) |\mathfrak{X}|^{-2m_n} \sum\limits_{m_1+\dots+m_{|\mathfrak{X}|}=m_n} \binom{m_n}{m_1,\dots,m_{|\mathfrak{X}|}}^2\\
    &= n(n-1) |\mathfrak{X}|^{|\mathfrak{X}|/2}(4\pi m_n)^{(1-|\mathfrak{X}|)/2}(1+o_{m_n}(1))(1-o_n(1)) \label{eqn:multinomialuniform}\\
    &= n^2 m_n^{\frac{1-|\mathfrak{X}|}{2}} (4\pi)^{(1-|\mathfrak{X}|)/2} |\mathfrak{X}|^{|\mathfrak{X}|/2} (1+o_{m_n}(1))(1-o_n(1))
\end{align}
}{
\begin{align}
    \mu_n&= n(n-1)\Pr(H^{(1)}_1=H^{(1)}_2)\\
    &= n(n-1)\sum\limits_{h^{|\mathfrak{X}|}} \Pr(H^{(1)}_1=h^{|\mathfrak{X}|})^2\\
    &=n(n-1) \hspace{-1em} \sum\limits_{m_1+\dots+m_{|\mathfrak{X}|}=m_n} \binom{m_n}{m_1,\dots,m_{|\mathfrak{X}|}}^2 |\mathfrak{X}|^{-2m_n}\\
    &= n(n-1) |\mathfrak{X}|^{-2m_n} \hspace{-1em} \sum\limits_{m_1+\dots+m_{|\mathfrak{X}|}=m_n} \binom{m_n}{m_1,\dots,m_{|\mathfrak{X}|}}^2\\
    &= n(n-1) |\mathfrak{X}|^{|\mathfrak{X}|/2}(4\pi m_n)^{(1-|\mathfrak{X}|)/2}\notag\\&\hspace{7em}(1+o_{m_n}(1))(1-o_n(1)) \label{eqn:multinomialuniform}\\
    &= n^2 m_n^{\frac{1-|\mathfrak{X}|}{2}} (4\pi)^{(1-|\mathfrak{X}|)/2} |\mathfrak{X}|^{|\mathfrak{X}|/2}\notag\\&\hspace{7em} (1+o_{m_n}(1))(1-o_n(1))
\end{align}
}
where \eqref{eqn:multinomialuniform} follows from \cite[Theorem 4]{richmond2008counting}.\qed
\end{appendices}

\end{document}